\documentclass[journal]{IEEEtran}
\usepackage{amsmath}
\usepackage{algorithmic}
\usepackage{array}
\usepackage[caption=false,font=normalsize,labelfont=sf,textfont=sf]{subfig}
\usepackage{textcomp}
\usepackage{stfloats}
\usepackage{url}
\usepackage{verbatim}
\usepackage{graphicx}
\usepackage{cite}
\usepackage{amsthm}
\usepackage{amsfonts,mathrsfs}
\usepackage{algorithmic}
\usepackage{graphicx}
\usepackage{textcomp}
\usepackage{xcolor}
\usepackage{times,arydshln}
\usepackage{fancyhdr,graphicx,amsmath,amssymb}
\usepackage[ruled,vlined]{algorithm2e}
\usepackage{enumerate}
\usepackage{blkarray}
\usepackage{bbm}
\usepackage{textcomp}
\usepackage{mathtools,cuted}
\def\BibTeX{{\rm B\kern-.05em{\sc i\kern-.025em b}\kern-.08em
		T\kern-.1667em\lower.7ex\hbox{E}\kern-.125emX}}
\usepackage{balance}

\usepackage[ruled,vlined]{algorithm2e}
\usepackage{algorithmic}
\SetKwInput{KwInput}{Input}                
\SetKwInput{KwOutput}{Output}   

\def \a {{\bf a}}
\def \b {{\bf b}}
\def \c {{\bf c}}

\def \r {{\bf r}}

\newtheorem{lemma}{\textbf{Lemma}}
\newtheorem{definition}{\textbf{Definition}}
\newtheorem{remark}{\textbf{Remark}}

\newtheorem{proposition}{\textbf{Proposition}}
\newtheorem{theorem}{\textbf{Theorem}}

\usepackage{xparse}

\newcounter{example}
\newenvironment{example}[1][]{\refstepcounter{example}\par\medskip
	\noindent \textbf{Example~\theexample. #1} \rmfamily}{\medskip}

\newcounter{Remark}

\newcounter{Corollary}

\begin{document}
	\title{ISI-Aware Code Design: A Linear Approach Towards Reliable Molecular Communication}
	
	\author{Tamoghno Nath,  \IEEEmembership{Student Member, IEEE}, Krishna Gopal Benerjee,  \IEEEmembership{Member, IEEE}, and \\ Adrish Banerjee, \IEEEmembership{Senior Member, IEEE}
		\thanks{The authors are affiliated with the Department of Electrical Engineering, Indian Institute of Technology Kanpur, India (email: \{tamoghno, kgopal, adrish\}@iitk.ac.in).}
	}
	\maketitle
	
	\begin{abstract}
		Intersymbol Interference (ISI) is one of the major bottlenecks in Molecular Communication via Diffusion (MCvD) systems resulting in degraded system performance.
		This paper first introduces two new families of linear channel codes to minimize the effect of ISI: linear Zero Pad Zero Start (ZPZS) and linear Zero Pad (ZP) codes, ensuring that every codeword is devoid of consecutive bit-1s. 
		Subsequently, the ZPZS linear and ZP linear codes are combined to form a binary ZP code, aiming for a higher code rate compared to the linear ZP codes, which can be decoded with a simple Majority Location Rule (MLR) algorithm. 
		Additionally, a linear Leading One Zero Pad (LOZP) code is proposed, which relaxes the zero padding constraints considering the placement of bit-1s in the codeword as an important metric to have an improved code rate than the ZP code.
		Finally, a closed-form expression is deduced to compute the expected ISI for the proposed codes and to demonstrate that the expected ISI is a function of the average bit-1 density of the codewords in a code.
		To compare the ISI and BER performance with average bit-1 density of the proposed codes, two types of MCvD channel are considered: (i) channel without refresh, where the previously transmitted bits persist for a longer duration and (ii) channel with refresh, where the channel is cleared after each successful reception of the message. 
		The ISI comparison, across different sequence distributions for a given length and weight, shows that the linear LOZP code exhibits superior resilience against ISI in a channel with refresh due to the placement of bit-1s at the initial positions, whereas the ZP code performs better in channel without refresh by reducing the average bit-1 density of the code.
		The asymptotic upper bound of the code rate is derived for the proposed codes, which depicts that a trade-off exists between the ISI and code rate. 
		The simulation results show that the proposed family of ZP and linear LOZP codes can improve the Bit Error Rate (BER) performance by controlling the bit-1 locations and the average bit-1 density of the code, specifically where the ISI is more pronounced over the channel noise, thus providing a better reliability compared to the conventional error correcting codes at different data rate regimes.
	\end{abstract}
	
	
	\section{Introduction}
	\IEEEPARstart{W}{ith} the evolution of nano-technology, chemical signalling has been one of the most promising solutions for the communication between bio-nanomachines. 
	Since the bio-nanomachines are made of biological materials and are highly reluctant to electromagnetic wave signals, molecular communication has been an obvious choice for a large variety of domains, such as biomedical instrumentation, smart drug delivery, lab-on-a-chip \cite{6122529} and also in industrial settings, such as monitoring in confined environments \cite{6881284}.
	
	Influenced by natural molecular communication systems, MCvD is one of the most preferred transmission mechanisms for short to medium-range communication \cite{suda2005exploratory}.  
	Particularly for a bio-engineering field, the external energy requirement and ease of bio-nanomachine structure have become two major constraints to establishing a secure communication link with limited computational resources. 
	In an MCvD system, specifically for bitwise communication \cite{nakano_eckford_haraguchi_2013}, the transmitter (Tx) releases molecules into the environment over a certain period, which freely diffuse in the medium.
	However, only a small number of transmitted molecules arrive at the fully absorbing receiver (Rx) in the current bit interval following Brownian motion. The rest of the molecules persist in the channel, leading to the channel memory and cause ISI \cite{7841486}. 
	This fully absorbing receiver model is biologically motivated by scenarios where small signalling molecules (e.g., Ca$^{2+}$) passively diffuse through the cell membrane and are permanently internalized upon reception, such as certain targeted drug delivery applications. However, there are other passive and reactive receiver models to depict different biological applications in an MCvD system as discussed in \cite{8742793}.
	Due to the ISI in the diffusion model with an absorbing receiver, the communication system experiences a performance drop in throughput and data reliability. Therefore, reducing the ISI effect is one of the main concerns in an MCvD system.

	\textit{Related Work}: In the existing literature, various ISI mitigation techniques have been proposed.
	The existing literature has mainly discussed four types of ISI mitigation aspects, namely, (i) modulation schemes \cite{9184816,10334472}, (ii) detection and equalizing techniques \cite{9839216,9768128,9840365,10088444}, (iii) channel codes design \cite{7273857,7248461,8633972,7859349,Minimum,6708566,8648429,8972472,10041114,10356127}, and (iv) source code design \cite{10250855,9336654}. Additionally, \cite{6868273,noel2014improving} have discussed several channel models for ISI-mitigation in an MCvD system (e.g., MCvD channel with drift and enzymes). 
	
	In this paper, we have considered the channel coding approach to reduce the effect of ISI in the MCvD system. The existing literature mostly considered the conventional error correcting codes employed in an MCvD system, such as Hamming code, LDPC code, Cyclic Reed-Muller code \cite{7273857}, Self Orthogonal Convolution code\cite{7248461} and Reed Solomon code \cite{7859349,8633972,10356127}. 
	The authors in \cite{7273857,7248461} emphasized the significance of critical distance, a measure of the actual transmission distance, and coding gain as key parameters, demonstrating the necessity of selecting a code for achieving optimal performance when the transmission distance from the source to the target is known.
	In \cite{7859349}, Reed Solomon codes were proposed for an MCvD channel to enhance transmission reliability. The authors demonstrated that these codes can decode codewords with longer block lengths in less decoding time compared to the Hamming code. 
	However, traditional error correcting codes have not integrated any ISI-reducing constraints while employed in the MCvD systems, consequently, motivating researchers to develop ISI-reducing channel codes in MCvD systems.
	
	ISI-free codes were initially introduced to eliminate ISI between codewords with very low decoding complexity \cite{6708566}. However, these codes show a better BER performance in low data rate regimes, where the effect of ISI is comparatively less \cite{6708566,8648429}.
	Also, the minimum energy channel code was explored to reduce the effect of ISI by considering energy dissipation as an essential metric \cite{Minimum}. 
	Another method, called Crossover Resistant Coding with Time Gap, has been proposed by the authors to reduce ISI for a one-dimensional noiseless channel \cite{keshavarz2019inter}. 
	In this channel code, the idea is to exploit a time gap between two consecutive codewords using two distinct types of molecules. 
	This motivates the construction of the ISI-mtg code \cite{8972472}, which constrains the consecutive bit-1s in the proposed codebook, and also, every codeword in the code starts with a bit-0.
	Hyun et al. devised an algorithm based on Huffman coding to construct a codebook that also avoids consecutive bit-1s \cite{10250855}. 
	In \cite{9336654}, the authors showed that controlling the ratio of bit-1 to bit-0 in the transmitted sequence using the inverse Huffman code can mitigate the effect of ISI in the channel.
	Furthermore, the authors introduced a non-linear \textit{ISI-red} code, leveraging constraints on the `maximum weight' and `location of bit-1s' in \cite{nath2023novel}. This approach enabled achieving a higher code rate compared to the existing ISI-reducing codes while maintaining similar ISI performance. Similarly, Tang et. al., in \cite{10356127}, showed that placing high-energy symbols at codeword beginnings in a Reed Solomon code reduces ISI. Whereas, in \cite{bhattacharjee2022channel,10134568}, the authors identified specific sequences that are highly affected by ISI within the code, emphasizing the need to avoid them.
	
	\textit{Motivation and Contribution}:
	It is well understood that the encoding and decoding complexity of the channel codes also have to be minimal for the bio-nanomachines in an MCvD system. However, to the best of our knowledge, the existing ISI-reducing channel codes have not considered any linear mapping to simplify the decoding complexity at the receiver. 
	Therefore, driven by the run-length constraints and the code constraints proposed in \cite{8972472}, first a
		 ZPZS linear code and then a linear approach for
		the ZP code have been proposed to combat the effect of ISI
		in the MCvD channel.
	Additionally, there are different sequence patterns for a given length and weight to reduce ISI in an MCvD channel.
	We propose a linear code construction based on the weight property as discussed in \cite{9840783} and compare the ISI performance of the proposed linear codes under specific code parameters. Note that bit-0s experience detrimental ISI, potentially causing bit flipping and consequent errors, whereas ISI has a constructive effect on bit-1, helping its correct detection at the receiver. 
	Therefore, the location of bit-1s in the code, the average ISI, the cumulative ISI by all the bit-0s and the maximum ISI experienced by a bit-0 in the codeword are some of the key parameters to construct an ISI-reducing linear code.
	The main contributions of this paper are:
	\begin{enumerate}
		\item We propose different families of linear channel codes based on two constraints: (i) ZP constraint, and (ii) Zero Start (ZS) constraint, which is defined as linear ZPZS code and a shifted construction of the ZPZS linear code or ZP linear code (Lemma \ref{ZPZE code parameters}, Lemma \ref{ZP code parameter} and Lemma \ref{constr2}). We also construct a family of ZP codes, a union of the linear ZPZS and the shifted linear ZP codes. Hence, the code rate has been improved over the linear ZPZS code, while each codeword in the codebook still satisfies the ZP constraint  (Lemma \ref{zp_non_linear_constr} and Lemma \ref{lemma_constrt2_non_linear_zp}).
		\item We additionally design a binary linear LOZP code based on the Ones-at-starting-position sequence described in \cite{9840783}, which relaxes the ZP constraint within the codeword up to a certain position (Lemma \ref{constr_3} and Lemma \ref{constr4}). The generalized construction of the LOZP code provides adaptability between ISI and the position of the maximum number of allowable consecutive bit-1s in the codeword.
		\item We derive a closed-form expression to compute the expected ISI in terms of bit-$1$-density of the code in Theorem \ref{ISI_density_theorem}. Then, we discuss the expected ISI of the proposed family of codes (Lemma \ref{ZPZS_sequence_weight}, Remark \ref{ZP_sequence_weight}, Remark \ref{ZP_and_ZPZS_sequence_weight} and Remark \ref{isi_zp_codes}) and validate the analytical expressions with the simulation result. Furthermore, we compare the expected ISI and the total ISI on bit-0 for the codeword experiencing the maximum ISI (Lemma \ref{isi_expected_lozp_code} and Remark \ref{isi_lozp_code_bound}). Our analysis indicates that a LOZP code with the same length and maximum weight as an ISI-mtg code is less affected by ISI in an MCvD channel (Lemma \ref{lemma_lozp_exist_for_isi_mtg} and Lemma \ref{bit_1_position_lemma}). 
		\item We introduce a location-based simple decoding mechanism for the proposed ZP codes and name it MLR decoding (Algorithm \ref{Encoding Algo}). This decoding technique is motivated by linear decoding, and therefore, the decoding complexity is also minimal for the bio-nanomachines. We also compare the simulation framework of the proposed decoding mechanism with the theoretical analysis.
		\item Finally, we derive an upper bound on the code rate for the family of ZP codes and also evaluate the asymptotic code rate of the proposed codes. The BER results show that the proposed channel codes perform better than the existing channel codes in the presence of noise and channel memory for different data rate regimes.
	\end{enumerate}
	
	\textit{Organization}:
	Section \ref{Sec 1} illustrates some basic definitions and notations used in this paper. 
	In Section \ref{Sec 2}, an MCvD system model and the channel characteristics are discussed in the presence of ISI. 
	Section \ref{Sec 3} describes ZPZS, ZP and LOZP code constructions. In Section \ref{Sec 5_1}, we have derived the closed-form expression of the expected ISI for the proposed codes followed by the encoding and decoding mechanisms in Section \ref{Sec 4}. Section \ref{Sec 5} compares the bounds of the code rate and also discusses the asymptotic code rates on the proposed codes. 
	Finally, in Section \ref{Sec 6}, the simulation results have been discussed for the proposed codes followed by the conclusion in Section \ref{sec 7}.
	
	\section{Notation and Preliminaries}\label{Sec 1}
	In this section, we give an overview of the notations and also discuss some basic definitions that have been used throughout this paper.
	
	A one dimensional array $\a$ = $a_1a_2\ldots a_n\in\mathbb{Z}_2^n$ of length $n$ over the binary alphabet $\mathbb{Z}_2=\{0,1\}$ is called a binary sequence, where $\in\mathbb{Z}_2^n$ is the set of all the binary sequences of length $n$.
	For two binary sequences $\a$ = $a_1a_2\ldots a_n$ of length $n$ and $\b$ = $b_1b_2\ldots b_m$ of length $m$, the binary sequence $\a\b$ = $a_1a_2\ldots a_nb_1b_2\ldots b_m$ of length $n+m$ is the concatenated sequence of the binary sequences $\a$ and $\b$. 
	For a binary sequence $\a$ of length $n$, the $nm$-length sequence $\a^m$ is obtained from $m$ times concatenation of sequence $\a$.  
	For a binary sequence $\a$ = $a_1a_2\ldots a_n$ and $1\leq i\leq j\leq n$, the sequence $\a(i,j)=a_ia_{i+1}\ldots a_j$ of length $j-i+1$ is called a sub-sequence of the sequence $\a$.
	Also, in this paper, we define an all-zero block with $i$ rows and $j$ columns as $\mathbf{0}_{i,j}$.
	Again, a two-dimensional array over the binary alphabet $\{0,1\}$ is called a binary matrix.
	For a binary sequence and binary matrix, the shifted sequence and shifted matrix can be defined as follows.
	\begin{definition}
		For any binary sequence $\a$ = $a_1a_2\ldots a_n$ of length $n$, the sequence $T(\a)$ = $a_2a_3\ldots a_na_1$ is called shifted binary sequence or shifted sequence. 
		Similarly, for any matrix $G$ with $k$ rows $\textbf{g}_i$ ($i = 1,2,\ldots,k$), the shifted matrix $T(G)$ is a matrix with $k$ rows $T(\textbf{g}_i)$.  
		\label{Def shifted}
	\end{definition}
	\begin{example}
		For the sequence $\a$ = $0 1 1 1$, the shifted sequence is $T(\a)$ = $1 1 1 0$.
		Similarly, for the matrix 
		$G = 
		\begin{bmatrix}
			0 & 1 & 1 & 0\\
			0 & 1 & 0 & 1
		\end{bmatrix}$, 
		the shifted matrix is $T(G) = 
		\begin{bmatrix}
			1 & 1 & 0 & 0\\
			1 & 0 & 1 & 0
		\end{bmatrix}$.
	\end{example}

	\begin{definition}\label{ZP_definition}
		A binary sequence $\textbf{c}$ = $c_1c_2\ldots c_n$ of length $n$ satisfies
		\begin{itemize}
			\item ZP constraint if the sequence $\textbf{c}$ is free from consecutive bit-1s, $i.e.$, for $i$ = $1,2,\ldots,n-1$, if $c_i = 1$ then $c_{i+1} = 0$, and
			\item ZS constraint if the binary sequence $\textbf{c}$ starts with bit-0,  $i.e.$, $c_1$ = $0$.
		\end{itemize}
		\label{Def ZS and ZE sequences}
	\end{definition}
	For example, the binary sequence $\a$ = $01001010101000$ of length $14$ satisfies both the ZP and ZS constraints.
	For any positive integers $n$ and $\mathcal{S}$, a set $\mathcal{C}\subseteq\mathbb{Z}_2^n$ of size $\mathcal{S}$ is called an ($n,\mathcal{S}$) binary code $\mathcal{C}$ or \textit{simply} ($n,\mathcal{S}$) code $\mathcal{C}$.
	In particular, if the binary code $\mathcal{C}$ is the row span of a full rank matrix $G$ with $k$ rows and $n$ columns then the code is called binary linear code with the generator matrix $G$ and the parameter $[n,k]$, where the size of the code is $\mathcal{S}$ = $2^k$.
	
	Again, for any ($n,\mathcal{S}$) code $\mathcal{C}$, the code $T(\mathcal{C})$ = $\{c_1c_2\ldots c_n:c_nc_1c_2\ldots c_{n-1}\in\mathcal{C}\}$ is called Shifted code.
	For example, the code $\mathcal{C}$ = $\{000,001,010,110,111\}$ and the shifted code $T(\mathcal{C})$ = $\{000,010,100,101,111\}$ are ($3,5$) binary codes. 
	\begin{proposition}
		For any ($n,\mathcal{S}$) code $\mathcal{C}$, the shifted code $T(\mathcal{C})$ is also an ($n,\mathcal{S}$) code. 
		\label{Prop Shifted code Parameters}
	\end{proposition}
	\begin{proposition}
		For any [$n,k$] linear code $\mathcal{C}$ with the generator matrix $G$, the shifted code $T(\mathcal{C})$ is also an [$n,k$] linear code with the generator matrix $T(G)$. 
		\label{Prop Linear Shifted code Parameters}
	\end{proposition}
	\begin{definition}
		If each codeword of an ($n,\mathcal{S}$) binary code satisfies 
		\begin{itemize}
			\item ZP constraint then the binary code is called ZP code,
			\item ZS constraint then the binary code is called ZS code, and 
			\item both ZP and ZS constraints then the binary code is called ZPZS code.
		\end{itemize}  
		If the binary code is linear then ZP code, ZS code and ZPZS code are called ZP linear code, ZS linear code and ZPZS linear code, respectively.
		\label{Def ZS and ZE code}
	\end{definition}
	\begin{example}
		For binary matrices
		$G_1 = \begin{bmatrix}
			0 & 0 & 0 & 1\\
			0 & 1 & 1 & 0
		\end{bmatrix},$ 
		$G_2 = 
		\begin{bmatrix}
			1 & 0 & 0 & 0\\
			0 & 0 & 1 & 0
		\end{bmatrix}$ 
		$\mbox{ and }
		G_3 = 
		\begin{bmatrix}
			0 & 1 & 0 & 0\\
			0 & 0 & 0 & 1
		\end{bmatrix},$
		the [$4,2$] binary linear codes with the generator matrix $G_1$ is a ZS linear code, with the generator matrix $G_2$ is a ZP linear code, and with the generator matrix $G_3$ is a ZPZS linear code, where the ZPZS code satisfies both ZP and ZS constraints simultaneously.		
		Also, the binary code $\mathcal{C}$ = $\langle G_3\rangle\cup\langle G_2\rangle$ is a ($4,8$) ZP code, where $\langle G\rangle$ is the row space of the binary matrix $G$.
	\end{example}
	\begin{lemma}
		For any ZPZS code $\mathcal{C}$, the shifted code $T(\mathcal{C})$ is the ZP code.
		\label{ZPZE to Zp lemma}
	\end{lemma}
	\begin{proof}
		For any $\b$ = $b_1b_2\ldots b_n$ in $\mathcal{C}$, there exists a sequence $T(\b)$ = $b_2b_3\ldots b_n b_1$ in $T(\mathcal{C})$. 
		From the definition of the ZPZS sequence, $b_1=0$ and $b_2$ is not necessarily bit-0. Then, the sequence $T(\b)$ is a ZP sequence.
		Hence, the code $T(\mathcal{C})$ is a ZP code.
	\end{proof}
	If $\c(r)$ = $c_1(r)c_2(r)\ldots c_n(r)$ is the $r$-th codeword of length $n$ in an ($n,\mathcal{S}$) binary code $\mathcal{C}$, then the average density of bit-1 in the $t$-th position for the code is $\Delta_t(\mathcal{C}) = \frac{1}{\mathcal{S}}\sum_{r = 1}^{\mathcal{S}}c_t(r)$. Therefore, the average density of bit-1 of the ($n,\mathcal{S}$) code $\mathcal{C}$ is
	\begin{align}\label{average_bit1_density}
		\Delta (\mathcal{C}) = \frac{1}{n}\sum_{t = 1}^n\Delta_t(\mathcal{C}).
	\end{align}
	Note that, for any binary linear code $\mathcal{C}$, $\Delta_t(\mathcal{C}) = 0.5$, if the $t$-th bit in any of the codeword is a bit-1. 
	
	In the following definition, we propose LOZP sequence that depends on the distribution of ones in the codeword \cite{9840783}, which relaxes both the ZP and ZS constraints, respectively.
	\begin{definition}\label{lozp_seq}
		For any positive integers $n$ and $\tau$ ($<n$), consider any binary sequence $\c = c_1c_2\ldots c_{\tau}c_{\tau+1}\ldots c_{n}$ of length $n$. The sequence $\c$ is defined as a LOZP sequence if the following conditions hold:
		\begin{itemize} 
			\item for $i$ = $\tau,\tau+1,\ldots,n-1$, if $c_i = 1$ then $c_{i+1} = 0$, and
			\item the sub-sequence $\c(1,\tau)\in\mathbb{Z}_2^{\tau}$.
		\end{itemize}
	\end{definition}
	For example, $\a$ $=1110101$ is a LOZP sequence of length 7, where $\tau = 3$.
	\begin{definition}
		If each codeword of an ($n,\mathcal{S}$) binary code is a LOZP sequence then the code is called LOZP code. If the binary code is linear then it is defined as linear LOZP code.
	\end{definition}
	
	For example, the code $\mathcal{C} = \{1110101,1100101,0110001,$ $1010101\}$ is $(7,4)$ LOZP code with $\tau = 3$.
	
	\section{System Model}\label{Sec 2}
	This work considers a 3-dimensional diffusion-based molecular communication system in an unbounded environment with a point transmitter and a fully absorbing spherical receiver from \cite{6807659}. In this considered model, the point transmitter releases the molecules impulsively at the beginning of the symbol duration. Finally, a fraction of these molecules, propagated through a simple Brownian motion, gets absorbed by the receiver. Thus, the capture probability of the molecule until time $t$ is $F_{\mathrm{cap}}(t) = \frac {r_{0}}{d_{\mathrm{tr}}} \text {erfc}\!\left ({\!\frac {d_{\mathrm{tr}} - r_{0}}{\sqrt {4D^{\mathrm{ch}}t}}\!}\right)$ for $d_{\mathrm{tr}}>r_0$, where $D^{\mathrm{ch}}$ is the diffusion coefficient of the molecule in the channel, $d_{\mathrm{tr}}$ is the distance between the Tx to the centre of the Rx and $r_0$ is the radius of the Rx. 
	Thus, the channel coefficients are $p^{\mathrm{ch}}_{i} = 	F_{\mathrm{cap}}(it_{s})-F_{\mathrm{cap}}((i-1)t_{s})$ for $i =1,2,\ldots,L,\label{p_coef2}$ where $L$ denotes the channel memory with $t_s$ being the symbol time interval for the $i$-th bit. 
	The term $p^{\mathrm{ch}}_{i}$ is also defined as the capture probability of a molecule in the $i$-th symbol slot.
	
	Consider that a binary sequence $\c$ = $c_1c_2\ldots c_n$ of length $n$ is transmitted from the point Tx over the channel.
	The received bit at the $i$-th interval depends on the $i$-th transmitted bit $c_i$ as well as on the past $L$ bits, which is defined as the channel memory of length $L$.
	In this MCvD system, an On-Off keying modulation scheme is considered, where binary bit-1 is represented by transmitting a constant $M$ number of molecules and no molecules for binary bit-0 \cite{8412141}. 
	It is also assumed that the probability of transmitting binary bit-1 is equal to the probability of transmitting binary bit-0 in each symbol interval.
	Now, due to the probabilistic movement of the molecules, the molecules transmitted at the $k$-th slot can be received by the Rx at the $i$-th slot ($i\geq k$). 
	We denote this received number of molecules at the $i$-th slot by $M_{i,k}^{\text{Rx}}$, which follows a Binomial distribution \cite{8648429}.
	Consequently, the number of received molecules ($M_i^{\text{\scriptsize Rx}}$) at the end of $i$-th symbol from all the previous time slots is $M_i^{ \text{\scriptsize Rx}} = \sum_{k=1}^iM_{i,k}^{\text{Rx}}$. Therefore, for a large value of $M$, the received number of molecules at the $i$-th interval can be estimated by a Gaussian distribution \cite{8972472}
	\begin{align}\label{rx_molecule_dist}
		M^{\text{\scriptsize Rx}}_i\sim &\!\mathcal {N} \Big( \sum_{k=1}^{i}{M c_kp^{\mathrm{ch}}_{i-k+1}},\notag\\
		&\hspace{1.5cm}\sum_{k=1}^{i}{Mc_k p^{\mathrm{ch}}_{i-k+1} \left({1-p^{\mathrm{ch}}_{i-k+1}}\right)} \!+\! \sigma ^{2}_{n} \Big),
	\end{align}
	where the additive noise at the Rx follows a
	Gaussian distribution with mean $0$ and variance $\sigma^{2}_{n}$. 
	This Rx noise can be either environmental or counting noise and the authors in \cite{farsad2014channel, zhai2018anti} have shown that the Gaussian noise can fairly approximate the non-linearity observed in experimental setups for an MCvD system.
	This paper considers the following two scenarios: 
	\begin{enumerate}
		\item Without channel refresh: The receiver cannot erase the previously transmitted bits, resulting in a channel memory of length $L$ before the transmitted codeword. For example, if the length of both the first transmitted sequence $\c^{(1)} = c^{(1)}_1c^{(1)}_2\ldots c^{(1)}_n$ and the second transmitted sequence $\c^{(2)} = c^{(2)}_1c^{(2)}_2\ldots c^{(2)}_n$ are $n$, then the $n$-th bit of the second received sequence $\c^{(2)}$ will experience interference from the concatenated sequence $\c^{(1)}(2n-L,n)\c^{(2)}(1,n-1) = c^{(1)}_{2n-L}c^{(1)}_{2n-L+1}\ldots c^{(1)}_{n}c^{(2)}_1c^{(2)}_2\ldots c^{(2)}_{n-1}$ for $n<L<2n$.
		\item With channel refresh: In this scenario, after successfully receiving a message, the remaining molecules from the channel are removed. This leads to a channel memory length of $L = i-1$ for the $i$-th received bit at the Rx.
		This case considers the release of enzymes by the receiver that react with the persisting molecules in the environment \cite{noel2014improving}.
		For instance, acetylcholinesterase enzymes at neuromuscular junctions rapidly degrade acetylcholine molecules after their reception to avoid prolonged interference with subsequent signals, and eventually increases the channel capacity. 
		Note that this model is effective to analyze the ISI performance with bit-1 locations in the codeword for a given length and weight of the sequence.
	\end{enumerate}
To further analyze the ISI, consider a channel where the received $i$-th bit of the codeword $\c$ = $c_1c_2\ldots c_n \in \mathcal{C}$ is affected by the $L$ earlier bits denoted by $\hat{\mathbf{\c}} = \hat{c}_1\hat{c}_2\ldots \hat{c}_L$. For different memory lengths $L$, one can obtain the sequence $\mathbf{\hat{c}}$ as follows:
		\begin{itemize}
			\item Case 1 ($1\leq L\leq i-1$): In this case, $\hat{\mathbf{\c}} = \c(i-L,i-1)$, and thus, $\hat{{c}}_j = c_{i-L+j-1}$ for $j=1,2,\ldots, L$.
			This case also aligns with the channel with refresh, where $L = i-1.$
			\item Case 2 ($i-1<L\leq n+i-1$): In this case, $\hat{\mathbf{\c}} = \c^*(n-L+i,n)\c(1,i-1)$, and thus,        
			\begin{align}
				\hat{c}_j = 
				\begin{cases}
					c^*_{n-L+i+j-1} & \mbox{if }1\leq j\leq L-i+1,\\
					c_{j-L+i-1} & \mbox{if }L-i+1<j\leq L,
				\end{cases}
			\end{align}
			where $\c^*$ is any valid codeword in the code $\mathcal{C}$.
			\item Case 3 ($n+i-1<L$): In this case, $\hat{\mathbf{c}}$ = $\c^*(n
			\left\lfloor\frac{L}{n}\right\rfloor-L+i,n)\c^{*\left\lfloor\frac{L}{n}\right\rfloor-1}\c(1,i-1)$, where $\c^*\in\mathcal{C}$.    
			Therefore, 
			\begin{itemize}
				\item for $1\leq j\leq L-n\left\lfloor\frac{L-i+1}{n}\right\rfloor-i+1$, \\ $\hat{c}_j$ = $c^*_{n-L+n\left\lfloor\frac{L-i+1}{n}\right\rfloor+i+j-1}$,
				\item for $L-n\left\lfloor\frac{L-i+1}{n}\right\rfloor-i+1<j\leq L-i$, $\hat{c}_j$ = $c^*_{j*-n\left\lfloor\frac{j*}{n}\right\rfloor}$, where $j*=j-L+n\left\lfloor\frac{L-i+1}{n}\right\rfloor+i-1$,
				\item for $L-i+1\leq j\leq L$, $\hat{c}_j$ = $c_{j-L+1}$.
			\end{itemize}
			Also, ${\c^*}^{\lfloor\frac{L}{n}\rfloor-1}$ represents $\lfloor\frac{L}{n}\rfloor-1$ times concatenation of the sequence $\c^*$ and for a specific case if a sequence is transmitted repeatedly then $\mathbf{c}^* = \mathbf{c}$.
		\end{itemize}
		Therefore, the mean ISI on the $i$-th bit in the codeword $\c$ = $c_1c_2\ldots c_n$ of length $n$, generated from earlier sequence $\hat{\mathbf{\c}} = \hat{c}_1\hat{c}_2\ldots \hat{c}_L$, is a function of past symbols and probability co-efficient $p^{\mathrm{ch}}$ describing how many molecules are expected to reach the
		receiver, and subsequently this can be expressed as
		\begin{equation}\label{isi_sequence}
			\mathrm{ISI}^{L}_{i}=\sum\limits_{k=2}^{L+1}\hat{c}_{L-k+2}p^{\mathrm{ch}}_{k}, \mbox{ for }i = 1,2,\ldots,n.
		\end{equation}
	    
		Therefore, from \eqref{isi_sequence} and \cite{8972472}, the expected ISI, generated from the earlier sequence $\hat{\mathbf{c}}$ can be computed by considering all possible concatenated sequences $\c^*\c$. Hence, the expected ISI on the $i$-th bit in the codeword $\c$ for $i = 1, 2, \ldots, n$ is
		\begin{align}\label{expected_ISI_i}
			\mathbb{E}\left[\mathrm{ISI}^L_{i}\right] = 
			\begin{cases}
				\frac{1}{2}\sum\limits_{k=2}^{L+1}p^{\mathrm{ch}}_{k},&\mbox{for un-coded case}\\
				\frac{1}{\mathcal{S^*}}\sum\limits_{\hat{\mathbf{c}}\in\mathcal{C}^*}\sum\limits_{k=2}^{L+1}\hat{c}_{L-k+2}p^{\mathrm{ch}}_{k},&\mbox{for ($n,\mathcal{S}$) code}~ \mathcal{C},
			\end{cases}
		\end{align} 
		where $\mathcal{C}^*$ denotes the set $\{\c^*\c:\c\in\mathcal{C}\}$ of size $\mathcal{S}^*$.
	Please note that the expected ISI for the uncoded case is independent of the index $i$.
	Also, for a large number of all possible transmitted sequences, $\mathcal{S}^*$ approaches to the code size of the code $\mathcal{C}$.
	Now, from \eqref{expected_ISI_i}, we define the average ISI of any $(n,\mathcal{S})$ code $\mathcal{C}$ with the channel memory $L~(>n)$ as 
	\begin{align}\label{equation_avg_isi_code}
		\mathrm{ISI}_{\mathrm{avg}}(\mathcal{C}) = \frac{1}{n}\sum_{i = 1}^{n}\mathbb{E}\left[\mathrm{ISI}^L_{i}\right].
	\end{align}
	Note that with channel refresh and for $L = n$, the average ISI of the code can be denoted as $ \mathrm{ISI}_{\mathrm{avg}}(\mathcal{C}) = \frac{1}{n}\sum_{i = 2}^{n+1}\mathbb{E}\left[\mathrm{ISI}^n_{i}\right]$.
	We define the total ISI on the bit-0s in a codeword $\c=c_1c_2\ldots c_n$ as $\mathrm{ISI}^{L,0}(\c)$, i.e., 
	\begin{equation}
		\mathrm{ISI}^{L,0}(\c)=\sum_{\substack{i; \\ c_i=0 \mbox{ and } \\ i\in\{1,2,\ldots,n\}}}\mathrm{ISI}^L_i.
	\end{equation} 
	While the maximum ISI experienced by a bit-0 in the codeword $\c$ is 
	\begin{equation}
		\mathrm{ISI}^{L,0}_{\mathrm{max}}(\c)=\max\{\mathrm{ISI}^L_i:c_i=0\mbox{ and }i=1,2,\ldots,n\}.
	\end{equation}
	Therefore, if $c_i = 0$, then we define the $i$-th bit expected ISI in the code $\mathcal{C}$ as $\mathbb{E}[\mathrm{ISI}_i^{L,0}]$. 
	
	For a $(5,2)$ binary code $\mathcal{C} = \{00100,10100\}$, we first compute the expected ISI on the $i$-th bit of the code for a channel memory $L = 4$ in Table \ref{table_expected_isi_example} (without channel refresh). 
	Considering $k_1, k_2 = 1,2$, the first bit of the message sequence $\c^{(k_2)}$ experience the effect of ISI from the past $L = 4$ bits, $i.e.,$ the subsequence $\c^{(k_1)}(2,5)$ of length $L = 4$.
		While, the sequence $\c^{(k_1)}(1+i,5)\c^{(k_2)}(1,i-1)$ of length $L = 4$ affect the $i$-th bit of the message sequence $\c^{(k_2)} = c^{(k_2)}_1c^{(k_2)}_2\ldots c^{(k_2)}_5$ for $i =  2, 3, 4$. In this example, we consider codewords $\c^{(1)} = 00100$ and $\c^{(2)} = 10100$.
	
	\begin{table}[ht]
		\centering
		\caption{$\mathrm{ISI}_i^L$ for code $\mathcal{C} = \{00100,10100\}$ with $L = 4.$}
		\begin{tabular}{|p{0.02cm}|p{0.8cm}|p{0.8cm}|p{1.1cm}|p{1.1cm}|p{1.1cm}|p{1cm}|}
			\hline
			& \multicolumn{4}{c|}{$\mathrm{ISI}_i^L$} & & \\ \cline{2-6} 
			$i$ & $\c^{(1)}\c^{(1)}$ & $\c^{(2)}\c^{(1)}$ & $\c^{(1)}\c^{(2)}$ & $\c^{(2)}\c^{(2)}$ & $\mathbb{E}[\mathrm{ISI}_i^4(\mathcal{C})]$ & $\mathrm{ISI}_{\mathrm{avg}}(\mathcal{C})$ \\	\cline{1-7}
			$1$ & $p^{\mathrm{ch}}_4$ & $p^{\mathrm{ch}}_4$ & $p^{\mathrm{ch}}_4$ & $p^{\mathrm{ch}}_4$ & $p^{\mathrm{ch}}_4$ & $0.3p^{\mathrm{ch}}_2 + $\\	\cline{1-6}
			$2$ & $p^{\mathrm{ch}}_5$ & $p^{\mathrm{ch}}_5$ & $p^{\mathrm{ch}}_2+p^{\mathrm{ch}}_5$ & $p^{\mathrm{ch}}_2+p^{\mathrm{ch}}_5$ & $\frac{p^{\mathrm{ch}}_2+2p^{\mathrm{ch}}_5}{2}$ & $0.3p^{\mathrm{ch}}_3 +$ \\	\cline{1-6} 
			$3$ & $0$ & $0$ & $p^{\mathrm{ch}}_3$ & $p^{\mathrm{ch}}_3$ & $\frac{p^{\mathrm{ch}}_3}{2}$ & $0.3p^{\mathrm{ch}}_4 + $\\	\cline{1-6}
			$4$ & $p^{\mathrm{ch}}_2$ & $p^{\mathrm{ch}}_2$ & $p^{\mathrm{ch}}_2 + p^{\mathrm{ch}}_4$ & $p^{\mathrm{ch}}_2 + p^{\mathrm{ch}}_4$ & $\frac{2p^{\mathrm{ch}}_2 + p^{\mathrm{ch}}_4}{2}$ & $0.3p^{\mathrm{ch}}_5$\\	\cline{1-6}
			$5$ & $p^{\mathrm{ch}}_3$ & $p^{\mathrm{ch}}_3$ & $p^{\mathrm{ch}}_3 + p^{\mathrm{ch}}_5$ & $p^{\mathrm{ch}}_3 + p^{\mathrm{ch}}_5$ & $\frac{2p^{\mathrm{ch}}_3 + p^{\mathrm{ch}}_5}{2}$ & \\	\hline
		\end{tabular}
		\label{table_expected_isi_example}
	\end{table}
	
	We also compute the following parameters for a channel memory length of $L = 4$.
	\begin{itemize}
		\item Without channel refresh: it is assumed that before the transmission of the first message, the channel does not contain any existing molecules, and consequently, $\mathbf{0}_{1,4}$ is the channel memory for the sequence $\c^{(1)}$. Therefore
		\begin{itemize}
			\item $\mathrm{ISI}^{4,0}(\c^{(1)}) = p^{\mathrm{ch}}_2+p^{\mathrm{ch}}_3$. Then, as shown in Fig. \ref{fig_isi_effect}(b), we have $\mathrm{ISI}^{4,0}(\c^{(2)}) = 2p^{\mathrm{ch}}_2 + p^{\mathrm{ch}}_3 + p^{\mathrm{ch}}_4 + 2p^{\mathrm{ch}}_5$, where the previous bits are from $\c^{(1)}$, and hence the resultant sequence becomes $010010100$.
			\item $\mathrm{ISI}^{4,0}_{\mathrm{max}}(\c^{(1)}) = p^{\mathrm{ch}}_2$, and $\mathrm{ISI}^{4,0}_{\mathrm{max}}(\c^{(2)}) = p^{\mathrm{ch}}_2+p^{\mathrm{ch}}_4$ correspond to $c_4 = 0$ in both the sequences. Thus, $\mathbb{E}[\mathrm{ISI}_4^{4,0}(\mathcal{C})] = \frac{2p^{\mathrm{ch}}_2+p^{\mathrm{ch}}_4}{2}$.
		\end{itemize}
		\item With channel refresh: there is no \textit{a-priori} message before the current transmitted message. Therefore
		\begin{itemize}
			\item $\mathrm{ISI}^{4,0}(\c^{(1)}) = p^{\mathrm{ch}}_2+p^{\mathrm{ch}}_3$. Then, as shown in Fig. \ref{fig_isi_effect}(a), we observe that $\mathrm{ISI}^{4,0}(\c^{(2)}) = 2p^{\mathrm{ch}}_2 + p^{\mathrm{ch}}_3 + p^{\mathrm{ch}}_4 + p^{\mathrm{ch}}_5$.
			\item $\mathrm{ISI}^{4,0}_{\mathrm{max}}(\c^{(1)}) = p^{\mathrm{ch}}_2$ and $\mathrm{ISI}^{4,0}_{\mathrm{max}}(\c^{(2)}) = p^{\mathrm{ch}}_2+p^{\mathrm{ch}}_4$.
			\item $\mathrm{ISI}_{\mathrm{avg}}(\mathcal{C}) = 0.3p^{\mathrm{ch}}_2 + 0.3p^{\mathrm{ch}}_3 + 0.1p^{\mathrm{ch}}_4 +$ $ 0.1p^{\mathrm{ch}}_5$.
		\end{itemize}
	\end{itemize}
	Now, one can revisit the expression for expected ISI as given in  \eqref{expected_ISI_i} in Theorem \ref{ISI_density_theorem}.
	\begin{theorem}
		For any $(n,\mathcal{S})$ binary code $\mathcal{C}$ with bit-$1$ density $\Delta_i(\mathcal{C})$ ($i=1,2,\ldots,n$), 
		\begin{equation}\label{eq_isi_with_density}
			\mathbb{E}\left[\mathrm{ISI}^L_{i}(\mathcal{C})\right]=\sum_{k=2}^{L+1}\Delta_{i-k+1}(\mathcal{C})p^{\mathrm{ch}}_{k},
		\end{equation}
		where $\Delta_0$ = $\Delta_n$ and $\Delta_{-j}$ = $\Delta_{n-j}$ for any positive integer $j$. 
		\label{ISI_density_theorem}
	\end{theorem}
	\begin{proof}
		For $(n,\mathcal{S})$ binary code $\mathcal{C}$, $\mathbb{E}\left[\mathrm{ISI}^L_{i}\right]$ is calculated in (\ref{expected_ISI_i}) using finite double summation. Then, by interchanging the summation, we have 
		\begin{align*}
			\mathbb{E}\left[\mathrm{ISI}^L_{i}(\mathcal{C})\right] = \sum\limits_{k=2}^{L+1}\left(p^{\mathrm{ch}}_{k}\left(\frac{1}{\mathcal{S}}\sum\limits_{\Hat{\mathbf{c}}\in\mathcal{C}}\hat{c}_{L-k+2}\right)\right).
		\end{align*}
		Note that the sequence $\hat{\c}$ is the sub-sequence of a sequence in  $\mathcal{C}^{\lceil\frac{L}{n}\rceil}$. Also, recall that the average bit-$1$ density for $\mathcal{C}$ and $\mathcal{C}^{\lceil\frac{L}{n}\rceil}$ are equal. Thus, $\Delta_{i-k+1}(\mathcal{C})=\frac{1}{\mathcal{S}}\sum_{\Hat{\mathbf{c}}\in\mathcal{C}}\hat{c}_{L-k+2}$. This completes the proof.
	\end{proof}
	Using Theorem \ref{ISI_density_theorem} and \eqref{equation_avg_isi_code}, one can observe that the average ISI of the code is a function of the code weight density. 
	Also the channel memory length $L$ is a function of the symbol duration ($t_s$) and the time required to reach some negligible hitting
		probability $\alpha$, $i.e.,$ $t_{\alpha}$, and therefore can be given as $L = \lceil{t_{\alpha}}/{t_s} \rceil$, where $t_{\alpha}$ can be determined from \cite[Eq. (5)]{10361887} 
		\begin{align}\label{eq_prob_t_alpha}
			\frac{r_0}{d_{\mathrm{tr}}}\left(\text{erfc}\left(\frac{d_{\mathrm{tr}}-r_0}{ \sqrt{4D^{\mathrm{ch}}\left(t_{\alpha}+t_{s}\right)}}\right)-\text{erfc}\left(\frac{d_{\mathrm{tr}}-r_0}{\sqrt{4D^{\mathrm{ch}} t_{\alpha}}}\right)\right)=\alpha
		\end{align}	
		for a given probability $\alpha$ and symbol duration $t_s$.
	\begin{figure}
		\centering
		\includegraphics[width = 1\linewidth]{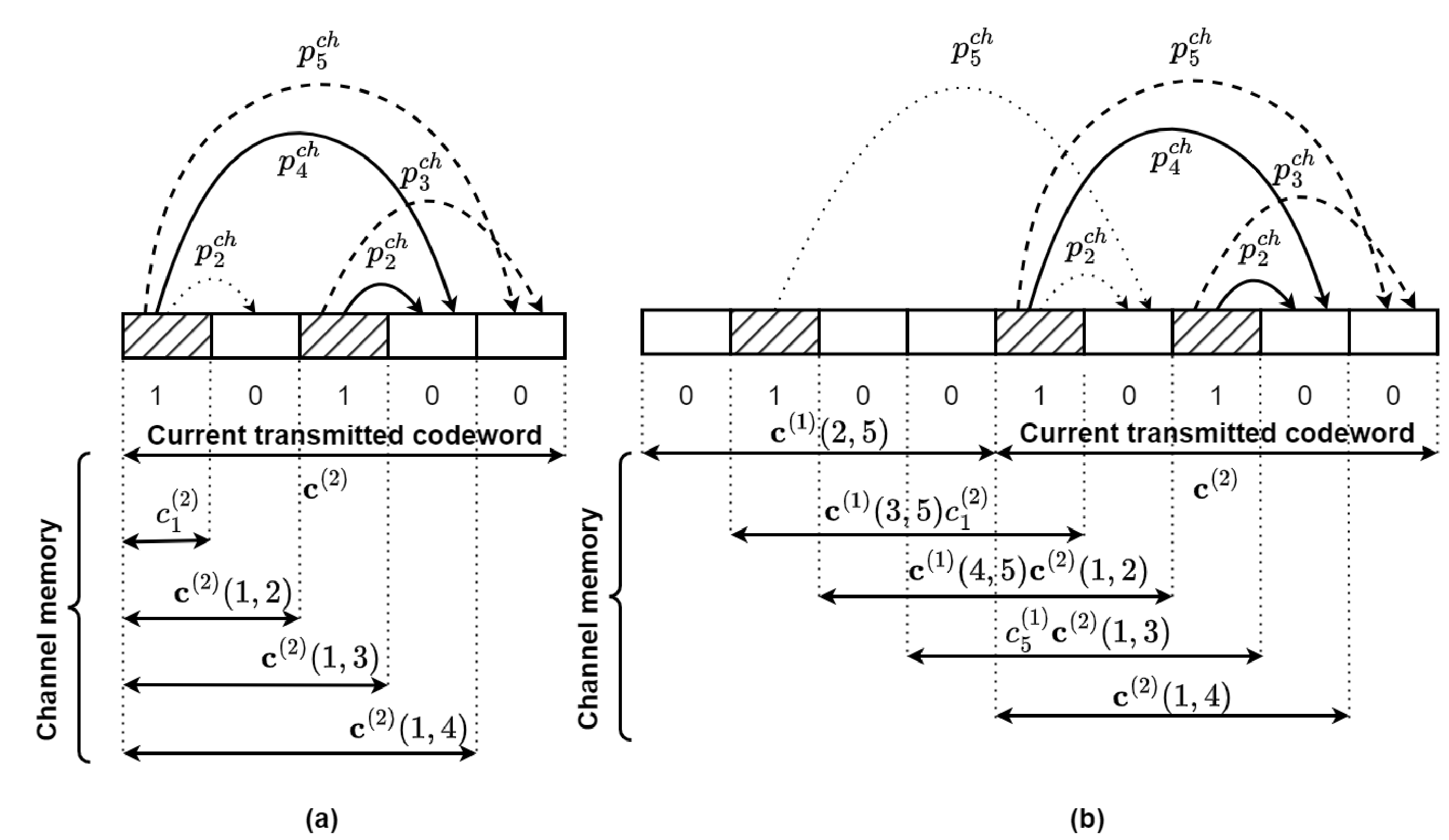}
		\caption{ISI effect on bit-0 (a) with channel refresh and (b) without channel refresh for channel memory $L = 4$.}
		\label{fig_isi_effect}
	\end{figure}
	
	In the demodulation technique, we use the threshold detector for the nano-networks to compare the number of received molecules with some pre-determined threshold at the receiver taken by the weighted sum detector \cite{6868273}. 
	And the threshold $\zeta$ is computed by sending the pilot symbols over the MCvD channel.
	Thus, in the molecular channel, for an un-coded case with channel memory $L$, the average probability of error can be given by \cite{7273857} 
	\begin{align}\label{prob_error}
		P_{\mathrm{e}} = \frac{1}{2^{L+1}}\left[\sum\limits_{j=1}^{2^L}\left(Q\left(\frac{\zeta-\mu_{0,L}[j]}{\sigma_{0,L}[j]}\right) + Q\left(\frac{\mu_{1,L}[j]-\zeta}{\sigma_{1,L}[j]}\right)\right)\right],
	\end{align}
	where $\mu_{k, L}[j]$ and $\sigma_{k, L}^2[j]$ are the mean and the variance of the variable $M_L^{\mathrm{Rx}}[j]$, respectively, corresponding to the $j$-th codeword.
	Here, $k = 0$ represents that the $L$-th transmitted symbol is a bit-0 and $k = 1$ if the $L$-th transmitted symbol is a bit-1. It is considered that the transmission probability of both the symbols bit-0 and bit-1 are equally likely. 
	However, the existing literature has no closed-form expression for the one-to-one relation between the average BER and expected ISI.
	From \eqref{prob_error}, we can observe that the average BER, at the optimum threshold, depends on the channel memory $L$, codewords and symbol duration $t_s$.
	Since ISI is also a function of the weight distribution in the code $\mathcal{C}$ \cite{nath2023novel} and 
		channel memory $L$ (or symbol duration $t_s$) \cite{10361887}, it follows that the average BER of the system will
		increase as the expected ISI increases.
	
	\section{Zero Pad Codes} \label{Sec 3}
	In this section, we construct a family of binary linear ZPZS codes and a family of binary linear ZP codes for the given code parameters. Further, we have improved code rate by considering the ZP code obtained by taking union of the linear ZP code and the linear ZPZS code. In addition, we also have constructed linear LOZP code by weakening the zero-padding constraint on some initial positions in certain codewords, further enhancing the code rate. The code constructions and their properties are discussed below.
	
	\subsection{Construction of Zero Pad Codes $\mathcal{C}_{d,q}$}
	In this section, we first construct a family of ZP codes, where for any codewords in the code, the minimum number of zeros between two consecutive bit-1s is always either 1 or multiple of any positive integer $d-1\ (>1)$. In Lemma \ref{constr_1} and Lemma \ref{ZPZE code parameters}, the construction of a generator matrix followed by a linear ZPZS code is proposed for a given positive integer $d$.
	
	For any prime integer $d$ ($\geq 2$) and positive integer $q$, consider a matrix $G_{d, q}$ with the initial condition $G_{d,0}$ = $[0\ 1]$, where
	\begin{align}
		G_{d, q} = 
		\begin{bmatrix}
			G_{d,q-1} & \mathbf{0}_{q,d}\\
			\mathbf{0}_{1,dq+1} & 1
		\end{bmatrix}.
		\label{G_n_1,gen}
	\end{align}
	To provide a comprehensive approach, we associated parameters $d$ and $q$ with the matrix notation $G_{d, q}$. 
	But, note that the matrix $G_{d, q}$ is constructed recursively over $q$ only.
	\begin{lemma}[Construction 1]\label{constr_1}
		For any prime integer $d$ ($\geq 2$) and positive integer $q$, any binary linear code $\mathcal{C}_{d,q}$ with the generator matrix $G_{d,q}$ (as given in (\ref{G_n_1,gen})) is an [$n,k$] linear code, where $n$ = $qd + 2$ and $k$ = $q+1$.
	\end{lemma}
	\begin{proof}
		For any [$n,k$] binary linear code $\mathcal{C}_{d,q}$, the result on the code parameters $n$ and $k$ can be proved using mathematical induction on $q$ ($\geq1$). 
		
		For $q$ = $0$, consider the code $\mathcal{C}_{d,0}$, where the code parameters are $n$ = $2$ and $k$ = $1$, and therefore, the base case is true.
		Now, for $q$ = $t$, consider the code $\mathcal{C}_{d,t}$.
		Assume that, for $\mathcal{C}_{d,t}$, the code parameters are $n$ = $td + 2$ and $k$ = $t+1$.
		For $q$ = $t+1$, from generator matrix $G_{d,t+1}$ (as given in (\ref{G_n_1,gen})), code length $n$ = $(td + 2)+d$ = $(t+1)d+2$, and message length $k$ = $(t+1)+1$. 
		Hence, the identity is also true for $q$ = $t+1$.
		
		Thus, from mathematical induction, the code parameters are $n$ = $qd + 2$ and $k$ = $q+1$.
	\end{proof}
	\begin{lemma}
		For any prime integer $d$ ($\geq 2$) and non-negative integer $q$, any binary linear code $\mathcal{C}_{d,q}$ with the generator matrix $G_{d,q}$ is an [$n,k$] ZPZS linear code.
		\label{ZPZE code parameters}
	\end{lemma}
	\begin{proof}
		For an [$n,k$] binary linear code $\mathcal{C}_{d,q}$, properties of satisfying ZP and ZS constraints are proved using mathematical induction on $q$ ($\geq0$). 
		
		For $q$ = $0$, consider the linear code $\mathcal{C}_{d,0}$ = $\{00,01\}$, and it can be observed that ZP and ZS constraints are followed by the code $\mathcal{C}_{d,0}$.
		Thus, the base case is true.
		Now, for $q$ = $t$, assume that the linear code $\mathcal{C}_{d,t}$ (and therefore, each row of the generator matrix $G_{d,q}$) satisfies ZP and ZS constraints.
		Also, recall that $\textbf{0}_{a,b}$ is an all-zero block of $a$ rows and $b$ columns.
		Thus, for $q$ = $t+1$, each row of generator matrix $G_{d,t+1}$ (from (\ref{G_n_1,gen})) follows ZP and ZS constraints.
		Hence, from mathematical induction, the linear code satisfies ZP and ZS constraints, and thus, the code $\mathcal{C}_{d,q}$ is a ZPZS linear code.
	\end{proof}
	\begin{example}		\label{C_2,2 example}
		For $q = 2$ and $d$ = $2$, the code $\mathcal{C}_{2,2}$ =  $\left\{000000,000001,010000,000100,010100,010001,000101,\right.$ $\left.010101\right\}$ is a [$6,3$] ZPZS linear code with generator matrix 
		\begin{align*}
			G_{2, 2} = 
			\begin{bmatrix}
				0 & 1 & 0 & 0 & 0 & 0 \\
				0 & 0 & 0 & 1 & 0 & 0 \\
				0 & 0 & 0 & 0 & 0 & 1 \\			
			\end{bmatrix}.
		\end{align*}
	\end{example}
	\begin{example}
		For $q = 2$ and $d = 3$, the code $\mathcal{C}_{3,2}$ = $\{00000000$,	
		$01000000,00001000,00000001,01001000,01000001,00001$- 		
		$001,01001001\}$ is a [$8,3$] ZPZS linear code with the generator matrix 
		\\
		$\mbox{\hspace{1.5cm}}G_{3, 2} = 
		\begin{bmatrix}
			0 & 1 & 0 & 0 & 0 & 0 & 0 & 0\\
			0 & 0 & 0 & 0 & 1 & 0 & 0 & 0\\			
			0 & 0 & 0 & 0 & 0 & 0 & 0 & 1\\
		\end{bmatrix}$.
		\label{C_3,2 example}
	\end{example}
	
	In Lemma \ref{zp_constr}, we construct a family of linear ZP codes for a given positive integer $d$.
	\begin{lemma}\label{zp_constr}
		For any prime integer $d$ ($\geq 2$) and non-negative integer $q$, the binary code $T(\mathcal{C}_{d,q})$ is a [$qd + 2,q+1$] ZP linear code.
		\label{ZP code parameter}
	\end{lemma}
	\begin{proof}
		For any given prime integer $d$ ($\geq 2$) and given non-negative integer $q$, consider the binary code $\mathcal{C}_{d,q}$.
		From Proposition \ref{Prop Shifted code Parameters} and Definition \ref{Def shifted}, the parameters and generator matrix for the code $T(\mathcal{C}_{d,q})$ are [$qd + 2,q+1$] and $T(G_{d,q})$, respectively. 
		Now, from Lemma \ref{ZPZE to Zp lemma}, for the ZPZS code $\mathcal{C}_{d,q}$, the code $T(\mathcal{C}_{d,q})$ is a ZP code. 
		Hence, the code $T(\mathcal{C}_{d,q})$ is a [$qd + 2,q+1$] ZP linear code.
	\end{proof}
	\begin{example}
		For $q = 2$ and $d$ = $2$, code $T(\mathcal{C}_{2,2})$ =  $\{000000,100000,001000,000010,101000,100010,001010,$
		$ 101010\}$ is a [$6,3$] ZP linear code  with generator matrix 
		$\mbox{\hspace{2cm}}T(G_{2, 2}) = 
		\begin{bmatrix}
			1 & 0 & 0 & 0 & 0 & 0\\
			0 & 0 & 1 & 0 & 0 & 0\\
			0 & 0 & 0 & 0 & 1 & 0\\			
		\end{bmatrix}$.
		\label{T(C_2,2) example}
	\end{example}
	\begin{example}
		For $q = 2$, $d = 3$, $T(\mathcal{C}_{3,2})$ = $\{00000000,$ 
		$10000000,00010000,00000010, 10010000, 10000010, 00010$- $010,10010010\}$ is a [$8,3$] ZP linear code with matrix    
		\begin{align*}
			T(G_{3, 2}) = 
			\begin{bmatrix}
				1 & 0 & 0 & 0 & 0 & 0 & 0 & 0\\
				0 & 0 & 0 & 1 & 0 & 0 & 0 & 0\\			
				0 & 0 & 0 & 0 & 0 & 0 & 1 & 0\\
			\end{bmatrix}.
		\end{align*}
		\label{T(C_3,2) example}
	\end{example}
	In the following Lemma \ref{zp_non_linear_constr}, we construct a ZP code, which is the union of both the linear ZPZS code and the linear ZP code for the given code parameters.
	\begin{lemma}\label{zp_non_linear_constr}
		For any prime integer $d$ ($\geq 2$) and non-negative integer $q$, the binary code $\mathcal{C}_{d,q}\cup T(\mathcal{C}_{d,q})$ is a ($qd + 2,2^{q+2}-1$) ZP code.
	\end{lemma}
	\begin{proof}
		The parameters of the code $\mathcal{C}_{d,q}\cup T(\mathcal{C}_{d,q})$ are followed from Lemma \ref{ZPZE code parameters}, Lemma \ref{ZP code parameter} and the property of inclusion-exclusion between the code $\mathcal{C}_{d,q}$ and the code $T(\mathcal{C}_{d,q})$.
	\end{proof}
	\begin{example}
		The code $\mathcal{C}_{3,2}\cup T(\mathcal{C}_{3,2})$ is the ($8,15$) ZP code, and the code $\mathcal{C}_{2,2}\cup T(\mathcal{C}_{2,2})$ is the ($6,15$) ZP code, where $\mathcal{C}_{2,2}$, $\mathcal{C}_{3,2}$, $T(\mathcal{C}_{2,2})$, and $T(\mathcal{C}_{3,2})$ are given in Example \ref{C_2,2 example}, Example \ref{C_3,2 example}, Example \ref{T(C_2,2) example} and Example \ref{T(C_3,2) example}, respectively.
	\end{example}
	\subsection{Construction of Zero Pad Codes $\mathcal{C}_q$}
	In this section, we propose a generalized construction of the ZP codes, where for any codewords in the code, the minimum number of zeros between two consecutive bit-1s can be any positive integer $d_i$. In the following Lemma \ref{constr2}, the generator matrix for the generalized linear ZPZS code is first given.
	\begin{lemma}[Construction 2]\label{constr2}
		For any given positive $q$ integers $d_1,d_2,\ldots,d_q$ ($\geq 2$), the binary linear code $\mathcal{C}_q$ with the generator matrix  
		
		\begin{align}
			G_q = 
			\begin{bmatrix}
				0 & 1 & \textbf{0}_{1,d_1-1} & 0     & \textbf{0}_{1,d_2-1} & \ldots & 0      & \textbf{0}_{1,d_q-1} & 0 \\
				0 & 0 & \textbf{0}_{1,d_1-1} & 1      & \textbf{0}_{1,d_2-1} & \ldots & 0      & \textbf{0}_{1,d_q-1} & 0 \\
				\vdots & \vdots         & \vdots & \vdots     & \vdots            & \ddots & \vdots         & \vdots & \vdots     \\
				0 & 0 & \textbf{0}_{1,d_1-1} & 0      & \textbf{0}_{1,d_2-1} & \ldots & 1      & \textbf{0}_{1,d_q-1} & 0 \\
				0 & 0 & \textbf{0}_{1,d_1-1} & 0      & \textbf{0}_{1,d_2-1} & \ldots & 0      & \textbf{0}_{1,d_q-1} & 1 
			\end{bmatrix}
			\label{G_n_d,gen}
		\end{align}
		is a $[2+\sum_{i=1}^qd_i,q+1]$ ZPZS code. 
		\label{Gen ZPZE code parameter}
	\end{lemma}
	\begin{proof}
		From \eqref{G_n_d,gen}, the proof follows the fact that the number of rows and the number of columns of the matrix $G_q$ are $q+1$ and $2+\sum_{i=1}^qd_i$, respectively, and the first column of the matrix $G_q$ is an all-zero column.
	\end{proof}
	Lemma \ref{lemma_constr2_zp} gives the construction of the generalized linear ZP codes, which can be defined from Lemma \ref{constr2}.
	\begin{lemma}\label{lemma_constr2_zp}
		For any given positive $q$ integers $d_1,d_2,\ldots,d_q ~(\geq 2)$, the binary code $T(\mathcal{C}_q)$ is a $[2+\sum_{i=1}^qd_i,q+1]$ ZP linear code.
		\label{Gen ZP code parameter}
	\end{lemma}
	\begin{proof}
		From Proposition \ref{Prop Shifted code Parameters} and Definition \ref{Def shifted}, the parameters and generator matrix for the code $T(\mathcal{C}_q)$ are $[2+\sum_{i=1}^qd_i,q+1]$ and $T(G_q)$, respectively. 
		Now, from Lemma \ref{ZPZE to Zp lemma}, for the ZPZS code $\mathcal{C}_q$, the code $T(\mathcal{C}_q)$ is a ZP code. 
		So, the code $T(\mathcal{C}_q)$ is a $[2+\sum_{i=1}^qd_i,q+1]$ ZP linear.
	\end{proof}
	In Lemma \ref{lemma_constrt2_non_linear_zp}, we give the generalized construction of the ZP codes for the given code parameters. 
	\begin{lemma}\label{lemma_constrt2_non_linear_zp}
		For any given positive $q$ integers $d_1,d_2,\ldots,d_q ~(\geq 2)$, the binary code $\mathcal{C}_q\cup T(\mathcal{C}_q)$ is a $(2+\sum_{i=1}^qd_i,2^{q+2}-1)$ ZP code.
	\end{lemma}
	\begin{proof}
		The parameters of the code $\mathcal{C}_q\cup T(\mathcal{C}_q)$ are followed from Lemma \ref{Gen ZPZE code parameter}, Lemma \ref{Gen ZP code parameter} and the property of inclusion-exclusion between the code $\mathcal{C}_q$ and the code $T(\mathcal{C}_q)$.
	\end{proof}
	\begin{example}\label{example_various_d}
		For $q$ = $2$, $d_1$ = $5$ and $d_2$ = $2$, the [$9,3$] ZPZS linear code with the generating matrix $G_2$ is $\mathcal{C}_2$ = $\{000000000,010000000, 000000100, 000000001, 010000100, $ 
		$010000001,000000101,010000101\}$, and the [$9,3$] ZP linear code with the generator matrix $T(G_2)$ is $T(\mathcal{C}_2)$ = $\{000000000,100000000,000001000, 000000010,$ $100001000, 100000010, 000001010, 100001010\}$, where 
		\begin{align*}
			& G_2 = 
			\begin{bmatrix}
				0 & 1 & 0 & 0 & 0 & 0 & 0 & 0 & 0\\
				0 & 0 & 0 & 0 & 0 & 0 & 1 & 0 & 0\\	
				0 & 0 & 0 & 0 & 0 & 0 & 0 & 0 & 1\\			
			\end{bmatrix},
			\mbox{ and }\\ \notag
			&T(G_2) = 
			\begin{bmatrix}
				1 & 0 & 0 & 0 & 0 & 0 & 0 & 0 & 0\\
				0 & 0 & 0 & 0 & 0 & 1 & 0 & 0 & 0\\	
				0 & 0 & 0 & 0 & 0 & 0 & 0 & 1 & 0\\			
			\end{bmatrix}.
		\end{align*}
		Also, the code $\mathcal{C}_2\cup T(\mathcal{C}_2)$ is a $(9,15)$ ZP code.
	\end{example}
	
	In addition to the codes based on zero padding constraints, we also propose a linear LOZP code in the following subsection. Note that a linear ZPZS code can achieve a maximum code rate of $0.5$ with $d = 2$, as constructed in Lemma \ref{constr_1}. Therefore, we need to properly choose the bit-1 locations in the generator matrix to achieve a higher code rate compared to linear ZP codes, and consequently, the placement of bit-1s within the codeword becomes an important metric for code design.
	
	\subsection{Construction of Leading One Zero Pad Codes}
	In this sub-section, we propose a family of linear LOZP codes for the given code parameters by relaxing both the ZP and ZS constraints.
	We first construct a family of linear LOZP codes, where for any codewords in the code, the minimum number of zeros between two consecutive bit-1s is always multiple of any positive integer $d~(\geq 2)$ after the $\tau$-th position. Observe that as the value of $\tau$ increases, we can allow more consecutive bit-1s at the initial positions of the codeword, thereby increasing the code rate at the expense of some ISI. In Lemma \ref{constr_3}, we discuss the code parameters for the proposed linear LOZP codes.
	
	\begin{lemma}[Construction 3]\label{constr_3}
		For any prime integer $d$ ($\geq 2$) and non-negative integer $q$ and $\tau$, any linear code $\mathcal{C}^{\tau}_{d,q}$ with the generator matrix $G^{\tau}_{d,q}$ satisfying the following recursive relation
		\begin{align}
			G^{\tau}_{d, q+1} = 
			\begin{bmatrix}
				G^{\tau}_{d,q} & \mathbf{0}_{\tau+q,d}\\
				\mathbf{0}_{1,d(q+1)+\tau-1} & 1
			\end{bmatrix},
			\label{G_n_1_lozp}
		\end{align}		
		is a $[qd + \tau, q+\tau]$ linear LOZP code with  the initial condition $G^{\tau}_{d,0}$ = $
		\mathbf{I}_{\tau}$, where $\mathbf{I}_{\tau}$ denotes an identity matrix of order $\tau$.
	\end{lemma}
	\begin{proof}
		The proof follows similar arguments from the proof of Lemma \ref{constr_1}.
	\end{proof}
	\begin{example}
		For $\tau = 2$, $q = 1$ and $d = 3$, the code $\mathcal{C}^{2}_{3,1}$ = $\{00000,10000,01000,00001,11000,10001,01001,11001\}$ is a [$6,3$] LOZP linear code with the generator matrix
		\begin{align*}
			&G^{2}_{3, 1} = 
			\begin{bmatrix}
				1 & 0 & 0 & 0 & 0\\
				0 & 1 & 0 & 0 & 0\\
				0 & 0 & 0 & 0 & 1\\			
			\end{bmatrix}.
		\end{align*}
		\label{lozp_d2_example}
	\end{example}
	
	We propose a generalized construction of the LOZP codes, where for any codewords in the code, the minimum number of zeros between two consecutive bit-1s after the $\tau$-th position can be any positive integer $d_i ~(\geq 2)$.
	\begin{lemma}[Construction 4]\label{constr4}
		For any given positive $q$ integers $d_1\geq d_2\geq \ldots \geq d_q$ ($\geq 2$), the binary linear code $\mathcal{C}^{\tau}_q$ with the generator matrix  
		\begin{align}
			G^{\tau}_q = 
			\begin{bmatrix}
				\begin{array}{cc}
					\mathbf{I}_{\tau} &  \mathbf{0}_{\tau,\sum_{i=1}^qd_i} \\
					\begin{array}{ccc}
						\mathbf{0}_{1,\tau-1+\sum_{i=1}^1d_i} & 1 &  
					\end{array} & \mathbf{0}_{1,\sum_{i=2}^qd_i} \\
					\begin{array}{ccc}
						\mathbf{0}_{1,\tau-1+\sum_{i=1}^2d_i} & 1 &  
					\end{array}
					& \mathbf{0}_{1,\sum_{i=3}^qd_i} \\ 
					\vdots & \vdots \\
					\mathbf{0}_{1,\tau-1+\sum_{i=1}^qd_i} & 1
				\end{array}
			\end{bmatrix}
			\label{G_n_d_gen_lozp}
		\end{align}
		is a $[\tau+\sum_{i=1}^qd_i,q+\tau]$ LOZP code. 
		\label{Gen_lozp_construction}
	\end{lemma}
	\begin{proof}
		The proof follows similar arguments from the proofs of Lemma \ref{constr2} and Lemma \ref{lemma_constr2_zp}.
	\end{proof}
	
	\begin{example}\label{lozp_generalized_example}
		For $\tau = 2$, $q = 2$, $d_1$ = $3$ and $d_2$ = $2$, the code $\mathcal{C}^2_2$ = $\{0000000,1000000,0100000,0000100,0000001,1100000, $ $1000100,1000001,0100100,0100001,0000101,1100100, 110$ $0001,1000101,0100101,1100101\}$, is a [$7,4$] LOZP linear code with the generating matrix 
		\begin{align*}
			G^2_2 = 
			\begin{bmatrix}
				1 & 0 & 0 & 0 & 0 & 0 & 0\\
				0 & 1 & 0 & 0 & 0 & 0 & 0\\
				0 & 0 & 0 & 0 & 1 & 0 & 0\\	
				0 & 0 & 0 & 0 & 0 & 0 & 1\\			
			\end{bmatrix}.
		\end{align*}
	\end{example}
	
	Observe that, for any given integer $d$ ($\geq2$), the maximum weight codeword of the code $\mathcal{C}_{d,q}$ have multiple of $d-1$ number of bit-$0$s in between any two bit-$1$s. 
	Then the property of a uniform number of zeros between two bit-$1$s is weakened and obtained the code $\mathcal{C}_q$ for given $d_i$ ($\geq 2$) $i=1,2,\ldots,q$. In this case, one can observe the following.
	\begin{itemize}
		\item If $d_1\geq d_2\geq\ldots\geq d_q\geq2$, then the bit-1s in the maximum weight sequence of the code $\mathcal{C}_q$ are pushed towards the $n$-th position. Such sequences are defined as Ones-at-end-position (OEP) sequences \cite{9840783}. Therefore, the code containing these sequences is defined as OEP code.
		Note that there can be consecutive bit-1s towards the end positions of the codewords in an OEP code.
		
		\item If $d_1\geq d_2\geq\ldots\geq d_j\geq2$ and $2\leq d_j\leq d_{j+1}\leq\ldots\leq d_q$, then the bit-1s are pushed towards the $j$-th position for some $j$ in between 1 and $n$ in the maximum weight sequence of the code $\mathcal{C}_q$. Such sequences are defined as Ones-at-middle-position (OMP) sequences in \cite{9840783}, and therefore, the code containing these sequences can be named  as OMP code.
		Note that there can be consecutive bit-1s towards the middle positions of the codewords in an OMP code.
	\end{itemize}
	Now, we calculate the expected ISI for the ZP code $\mathcal{C}_q$ and LOZP code $\mathcal{C}^{\tau}_q$ in Section \ref{Sec 5_1} as follows.
	
	\section{Expected ISI Computation}\label{Sec 5_1}
	In this section, we deduce the analytical expected ISI expression on the last bit of any codes $\mathcal{C}_q$ and $\mathcal{C}^{\tau}_q$. However, messages can be transmitted in any order in any communication system. Therefore, we compute the expected ISI over a large number of transmitted messages and show that the simulated result exactly matches the analytical ISI expression for a given channel memory. We also derive the analytical expressions of expected ISI for the proposed ZP codes. Subsequently, we compute this expression for a LOZP code and compare the ISI performance with linear OMP code and linear OEP code.
	
	\subsection{Expected ISI for ZP codes}
	In this section, we compute the analytical ISI expressions for the family of ZP codes with given parameters. To evaluate the ISI impact on the current bit considering the channel memory, we first need to determine the average density of bit-1s for the family of ZP codes.
	For any code if $\c$ is the maximum weight sequence of length $n$ then the sequence $\r$ = $r_1r_2\ldots r_L$ = $\c^{\lceil L/n\rceil}(n\lceil L/n\rceil-L,n\lceil L/n\rceil)$ and $D$ is the set of all indices of bit-$1$s in the maximum weighted sequence $\r$. 
	Since we are interested in all the possible positions of bit-$1$s in the sequence $\r$, and thus, we have considered the maximum weight sequence. 
	In Example \ref{example_various_d}, for the code $\mathcal{C}_2$, the maximum weight sequence is $\c$ = $010000101$. Now, consider the following three cases. 
	\begin{enumerate}
		\item For $L=27$ ($L\geq n$ and $L$ is multiple of $n$), $\r$ = $\c^3$ and $D$ = $\{2,7,9,11,16,18,20,25,27\}$.
		\item For $L=22$ ($L>n$ and $L$ is not multiple of $n$), $\r$ = $0101010000101010000101$ and $D$ = $\left\{2,4,6,11,13,15,\right.$ $\left.20,22\right\}$.
		\item For $L=6$ ($L<n$), $\r$ = $000101$ and $D$ = $\{4,6\}$. 
	\end{enumerate} 
	The set $D$ depends on the code, and thus the set of all indices of bit-$1$s in the maximum weighted sequence $\r$ is denoted by $D_1$ for ZPZS code $\mathcal{C}_q$, $D_2$ for ZP code $T(\mathcal{C}_q)$ and $D_3$ for code $\mathcal{C}_q\cup T(\mathcal{C}_q)$.  
	These sets $D_1$, $D_2$ and $D_3$ are calculated in Lemma \ref{ZPZS_sequence_weight}, Remark \ref{ZP_sequence_weight} and Remark \ref{ZP_and_ZPZS_sequence_weight} as follows. 
	\begin{lemma}
		For any given positive $q$ integers $d_1,d_2,\ldots,d_q$ ($\geq 2$), and channel memory $L$, consider the $[n,q+1]$ ZPZS code $\mathcal{C}_q$ with $n=2+\sum_{i=1}^qd_i$. 
		For a given channel memory $L$, if the transmitted $i$-th bit is bit-1, i.e. $c_i=1$, then $i\in D_1$ where $D_1$ is as follows.
		\begin{itemize}
			\item For $L<n$, 
			\begin{align}
				D_1=
				\begin{cases}
					\{1,1+\sum\limits_{m=1}^{\tau}d_i:\tau=1,2,\ldots,q\} & \mbox{ if } n-L=1 \\
					\left\{\right. \alpha+\sum\limits_{m=1}^{\tau}d_i:\alpha+\sum\limits_{m=1}^{\tau}d_i>0 & \\ 
					\left.\mbox{\hspace{1.5cm} and }\tau=1,2,\ldots,q\right\} & \mbox{ if } n-L>1
				\end{cases}
			\end{align}
			and $\alpha=2+L-n$.
			\item For $L\geq n$ and $L$ is a multiple of $n$,
			\begin{align}
				D_1=&
				\left\{\right.\tau^{*},\tau^{*}+\sum_{m=1}^\tau d_m:\tau^{*}=2+n\ell, \notag\\ 
				&\hspace{0.3cm}\ell=0,1,\ldots,(L/n)-1\mbox{ and }\tau=1,2,\ldots,q\left.\right\}.
			\end{align}
			\item For $L> n$ and $L$ is not a multiple of $n$,  
			\begin{align*}
				D_1 = & 
				B_1^*\cup\left\{\right.\tau^{*},\tau^{*}+\sum_{m=1}^\tau d_m:\tau^{*}=2+n\ell+L-n\left\lfloor\frac{L}{n}\right\rfloor, \notag\\ 
				&\ell=0,1,\ldots,\lfloor L/n\rfloor-1\mbox{ and }\tau=1,2,\ldots,q\left.\right\},\mbox{ and }
			\end{align*}
			\begin{align}
				B_1^*=
				\begin{cases}
					\{1,1+\sum\limits_{m=1}^{\tau}d_i:\tau=1,2,\ldots,q\} \\ \hspace{3cm} \mbox{ if } n-L+n\lfloor L/n\rfloor=1 \\
					\left\{\right. \alpha^*+\sum\limits_{m=1}^{\tau}d_i:\alpha^*+\sum\limits_{m=1}^{\tau}d_i>0   \\ 
					\left.\mbox{\hspace{2cm} and }\tau=1,2,\ldots,q\right\} \\ \hspace{3cm} \mbox{ if } n-L+n\lfloor L/n\rfloor>1
				\end{cases}
			\end{align}
			and $\alpha^*=2+L-n\lfloor L/n\rfloor-n$.
		\end{itemize}
		\label{ZPZS_sequence_weight}
	\end{lemma}
	\begin{proof}
		Consider a ZPZS code $\mathcal{C}_q$ with the generator matrix $G_q$. 
		For the binary matrix $G_q$, if column support of the matrix is the set $S_q$ of column indices of the non-zero column of the matrix. 
		Note that the set of indices of bit-$1$s in the maximum weight sequence is the column support of the generator matrix $G_q$. 
		Now, consider a binary sequence $\c=c_1c_2\ldots c_n$ with the maximum weight such that $c_i=1$ if and only if $i\in S_q$. 
		Note that the sequence $\c$ is a codeword with a maximum weight of $q+1$.
		Also, observe that, for any codeword $\Bar{\c}=\Bar{c}_1\Bar{c}_2\ldots \Bar{c}_n$ of the ZPZS code $\mathcal{C}_q$, if $\Bar{c}_i=1$ then $i\in S_q$.
		From enumeration, one can get $S_q=\{2,2+\sum_{m=1}^\tau d_m:\tau=1,2,\ldots,q\}$ of size $q+1$.
		Consider the transmitted binary sequence $\r=r_1r_2\ldots r_L$ with the channel memory of length $L$. 
		Now, there are three cases possible.
		\begin{itemize}
			\item \textbf{Case 1} ($L<n$): If the memory-length $L$ is less than the code-length $n$ then the sequence $\r$ is obtained from the sub-sequence $\c(n-L+1,n)$, i.e., $\r=r_1r_2\ldots r_L=\c(n-L+1,n)$. 
			Then, we have the following two sub-cases for each $\tau=1,2,\ldots,q$.
			\begin{itemize}
				\item \textbf{Sub-Case 1} ($n-L=1$): If $n-L=1$ then the weight of both sequences $\r$ and $\c$ are the same but the bit-$1$ indices in $\r$ is reduced by one of the bit-$1$ indices in $\c$, i.e., $D_1$ = $\{1,1+\sum_{m=1}^{\tau}d_i:\tau=1,2,\ldots,q\}$. 
				\item \textbf{Sub-Case 2} ($n-L>1$): If $n-L>1$ then the bit-$1$ indices in $\r$ is reduced by $n-L$ of the bit-$1$ indices in $\c$, i.e., $D_1$ = $\{2+L-n+\sum_{m=1}^{\tau}d_i:\tau=1,2,\ldots,q\mbox{ and }2+L-n+\sum_{m=1}^{\tau}d_i>0\}$.
				Note that the condition $2+L-n+\sum_{m=1}^{\tau}d_i>0$ appears for the fact that all the indices of any sequence are positive integers. 
			\end{itemize}
			\item \textbf{Case 2} ($L\geq n$ and $L$ is a multiple of $n$): If the memory length $L$ is a multiple of code length $n$ then the sequence $\r$ is obtained from the concatenation of $L/n$ numbers of sequences $\c$, $i.e.$, $\r$ = $\c^{L/n}$. 
			For each $\ell$ = $0,1,\ldots,(L/n)-1$, the $\ell$-th block $\c$ contributes the index $\tau^{*}$ for the first bit-$1$ and indices $\tau^{*}+\sum_{m=1}^\tau d_m$ for remaining bit-$1$s in the block $\c$. 
			Therefore, $D_1=\{\tau^{*},\tau^{*}+\sum_{m=1}^\tau d_m:\tau^{*}=2+n\ell,\ell=0,1,\ldots,(L/n)-1\mbox{ and }\tau=1,2,\ldots,q\}$.
			\item \textbf{Case 3} ($L>n$ and $L$ is not a multiple of $n$): If the memory-length $L$ is not a multiple of code length $n$ and $L>n$ then the sequence $\r$ is obtained from the concatenation of the sub-sequence $\c(n-L+n\lfloor L/n\rfloor+1,n)$ and $\lfloor L/n\rfloor$ numbers of sequences $\c$, i.e., $\r=r_1r_2\ldots r_L=\c(n-L+n\lfloor L/n\rfloor+1,n)\c^{\lfloor L/n\rfloor}$. 
			Now, for the set $D_1$, the set of indices of bit-$1$s contributed by the block $\c(n-L+n\lfloor L/n\rfloor+1,n)$ is $B_1^*$ and by the block $\c^{\lfloor L/n\rfloor}$ is $B_1$.
			In this case, $B_1^*$ is equal to $D_1$ as calculated in Case 1 by replacing $L$ with $L-n\lfloor L/n\rfloor$ and $B_1$ is equal to $D_1$ as calculated in Case 2 by shifting indices with $L-n\lfloor L/n\rfloor$. 
		\end{itemize}
		This completes the proof.
	\end{proof}
	\begin{remark}
		For any given positive $q$ integers $d_1,d_2,\ldots,d_q$ ($\geq 2$), consider the $[n,q+1]$ ZP code $T(\mathcal{C}_q)$ with $n=2+\sum_{i=1}^qd_i$. 
		Recall that the bit-$1$ positions are shifted with $1$ and indices are always positive integers. 
		Therefore, for a given channel memory $L$, if the transmitted $i$-th bit is bit-1, i.e. $c_i=1$, then $i\in D_2$, where $D_2$ = $\{j-1:j\in D_1\}\backslash\{0\}$.
		\label{ZP_sequence_weight}
	\end{remark}
	\begin{remark}
		For any positive $q$ integers $d_1,d_2,\ldots,d_q$ ($\geq 2$), consider the $(2+\sum_{i=1}^qd_i,2^{q+2}-1)$ ZP code $\mathcal{C}_q\cup T(\mathcal{C}_q)$. 
		From Definition \ref{Def ZS and ZE code}, Lemma \ref{ZPZS_sequence_weight} and Remark \ref{ZP_sequence_weight}, we have $D_1\cap D_2$ = $\emptyset$. 
		Therefore, for a given channel memory $L$, if the transmitted $i$-th bit is bit-1, i.e. $c_i=1$, then $i\in D_3$ = $D_1\cup D_2$. 
		\label{ZP_and_ZPZS_sequence_weight}
	\end{remark}
	Therefore, using Lemma \ref{ZPZS_sequence_weight}, Remark \ref{ZP_sequence_weight} and Remark \ref{ZP_and_ZPZS_sequence_weight} with Theorem \ref{ISI_density_theorem}, we can compute the analytical expression for the last bit ISI and is given in Remark \ref{isi_zp_codes}.
	\begin{remark}\label{isi_zp_codes}
		For any given channel memory $L$, the expected ISI on the last bit in the code $\mathcal{C}$, $i.e.,$ $\mathbb{E}[\mathrm{ISI}^L_n(\mathcal{C})]$,  is approaching to $\sum\limits_{i\in D}\Delta_i(\mathcal{C})p^{\mathrm{ch}}_{i}$, where 
		\begin{itemize}
			\item for $\mathcal{C}=\mathcal{C}_q$, $D$ = $D_1$ and
			\begin{align}
				\Delta_i(\mathcal{C}_q) = 
				\begin{cases}
					0.5 & \mbox{ if }i\in D_1 \\
					0 & \mbox{ otherwise}
				\end{cases}
			\end{align} 
			\item for $\mathcal{C}=T(\mathcal{C}_q)$, $D$ = $D_2$ and
			\begin{align}
				\Delta_i(T(\mathcal{C}_q)) = 
				\begin{cases}
					0.5 & \mbox{ if }i\in D_2 \\
					0 & \mbox{ otherwise}
				\end{cases}
			\end{align} 
			\item for $\mathcal{C}=\mathcal{C}_q\cup T(\mathcal{C}_q)$, $D$ = $D_3$ and      
			\begin{align}
				\Delta_i(\mathcal{C}_q\cup T(\mathcal{C}_q)) = 
				\begin{cases}
					\frac{2^q}{2^{q+2}-1} & \mbox{ if }i\in D_3 \\
					0 & \mbox{ otherwise}
				\end{cases}
			\end{align} 
		\end{itemize}
		and $D_1$, $D_2$ and $D_3$ are given in Lemma \ref{ZPZS_sequence_weight}, Remark \ref{ZP_sequence_weight} and Remark \ref{ZP_and_ZPZS_sequence_weight}.
		Note that the average density of bit-1 in the $i$-th position for any linear binary code $\mathcal{C}$ is $\Delta_i(\mathcal{C}) = 0.5$ and both the codes $\mathcal{C}_q$ and $T(\mathcal{C}_q)$ are linear codes. 
		Also, from the fact $\mathcal{C}_q\cap T(\mathcal{C}_q)=\{\mathbf{0}_{1,n}\}$, and the fact that $\Delta_i(\mathcal{C}_q)=0$ if and only if $\Delta_{i-1}(T(\mathcal{C}_q))=0$ ($i=2,3,\ldots n$). Now, for any given $i$, there are either $2^q$ codewords or zero codewords that are bit-1 in $\mathcal{C}_q$ and $T(\mathcal{C}_q)$.  
		Thus, $\Delta_i(\mathcal{C}_q\cup T(\mathcal{C}_q))$ = $\frac{2^q}{2^{q+2}-1}$.
	\end{remark}
	From Remark \ref{ZP_and_ZPZS_sequence_weight} and Remark \ref{isi_zp_codes}, we can deduce Remark \ref{lemma_isi_comparison_with_ds} to illustrate the relation between the expected ISI of a code and the minimum number of zero padding between two consecutive bit-1s.
	
	\begin{remark}\label{lemma_isi_comparison_with_ds}
		For any given positive integers $q$, $d_1$ ($\geq 2$) and $d_2$ ($\geq 2$), consider two ZP codes $\mathcal{C}_{d_1,q}\cup T(\mathcal{C}_{d_1,q})$ and $\mathcal{C}_{d_2,q}\cup T(\mathcal{C}_{d_2,q})$.
		For given length $n$ and channel memory $L$, if $d_2>d_1$, then the last bit expected ISI $\mathbb{E}[\mathrm{ISI}^L_n\left(\mathcal{C}_{d_1,q}\cup T(\mathcal{C}_{d_1,q})\right)]$ of the ZP code $\mathcal{C}_{d_1,q}\cup T(\mathcal{C}_{d_1,q})$ is larger than the last bit expected ISI $\mathbb{E}[\mathrm{ISI}^L_n\left(\mathcal{C}_{d_2,q}\cup T(\mathcal{C}_{d_2,q})\right)]$ of the ZP code $\mathcal{C}_{d_2,q}\cup T(\mathcal{C}_{d_2,q})$. Therefore, the difference of the last bit expected ISI ($\delta_{\mathrm{ISI}}$) between the two codes is 
		\begin{align}
			\delta_{\mathrm{ISI}} & = \mathbb{E}[\mathrm{ISI}^L_n\left(\mathcal{C}_{d_1,q}\cup T(\mathcal{C}_{d_1,q})\right) - \mathrm{ISI}^L_n\left(\mathcal{C}_{d_2,q}\cup T(\mathcal{C}_{d_2,q})\right)]\notag \\
			&
			= \frac{2^q}{2^{q+2}-1}\Bigg[\sum\limits_{m=1}^q \Big(\left(p^{\mathrm{ch}}_{2+md_1}-p^{\mathrm{ch}}_{2+md_2}\right)+\notag \\
			&\mbox{\hspace{4.5cm}} \left(p^{\mathrm{ch}}_{1+md_1}-p^{\mathrm{ch}}_{1+md_2}\right)\Big)\Bigg] \notag\\
			& > \frac{2^q}{2^{q+2}-1}\left[p^{\mathrm{ch}}_{1+d_1}-p^{\mathrm{ch}}_{1+d_2}\right],
		\end{align}
		where the channel coefficient is 
		$p^{\mathrm{ch}}_{i}$ and is a decreasing function of $i$. Therefore, the differences $p^{\mathrm{ch}}_{1+d_j}-p^{\mathrm{ch}}_{1+d_k}$ and $p^{\mathrm{ch}}_{2+d_j}-p^{\mathrm{ch}}_{2+d_k}$ are always positive for all $j$ and $k$ ($ d_j<d_k$).
	\end{remark}
	From Remark \ref{lemma_isi_comparison_with_ds}, we can also observe that as $(d_2-d_1)$ increases, the difference in expected ISI also increases. This results in improved system performance with the ZP code $\mathcal{C}_{d_2,q}\cup T(\mathcal{C}_{d_2,q})$, despite a compromised code rate compared to the ZP code $\mathcal{C}_{d_1,q}\cup T(\mathcal{C}_{d_1,q})$ where $d_2>d_1$.
	
	\subsection{Expected ISI for LOZP codes}
	This section computes the analytical ISI expressions for the family of linear LOZP, OMP, and OEP codes. It is essential to ensure uniformity in length and a maximum weight of a codeword across all the codes for a fair comparison of ISI performance across different codeword patterns. Therefore, we consider the MCvD channel with a refresh to compute and compare the expected ISI performance for a given code.
	In Example \ref{lozp_generalized_example}, for the code $\mathcal{C}^2_2$, the maximum weighted sequence is $\c$ = $1100101$ = $\textbf{1}_{1,\tau}\mathbf{0}_{1,d_1-1}1\mathbf{0}_{1,d_2-1}1$ with the parameters $\tau = 2$, $d_1 = 3$, $d_2 = 2$ and $q=2$.
	Now, there are one consecutive bit-1 block $\textbf{1}_{1,\tau}$ and two ($q = 2)$ zero padding blocks $\mathbf{0}_{1,d_{\phi}-1} 1$ in the sequence $\c$ for $\phi = 1,2$.
	Therefore, the ISI effect on the $i$-th bit of the consecutive bit-1 block arises from the previous $(i-1)$ bit-1s. The ISI effect on the first zero padding block ($\phi= 1$) arises due to the first two ($\tau = 2$) consecutive bit-1s. 
	In comparison, the ISI effect on the second zero padding block ($\phi= 2$) arises due to (i) the first two consecutive bit-1s, and (ii) the previous one ($\phi-1 = 1$) bit-1, that is after the consecutive bit-1 block.
	Consequently, from Theorem \ref{ISI_density_theorem}, we derive the analytical expression of expected ISI for LOZP codes in the following lemma, which is a function of the bit-1 locations in the codeword.
	
	\begin{lemma}\label{isi_expected_lozp_code}
		For any positive integers $d_i~ (\geq 2)$, $q$, and $\tau$ $(\geq 2)$,  consider a $[\tau+\sum_{i=1}^qd_i,q+\tau]$ LOZP code $\mathcal{C}^{\tau}_{q}$ with channel memory $L=n-1$ (with channel refresh). The expected ISI on the $i$-th $(\leq n)$ bit is
		\begin{align}\label{avg_isi_lozp_code}
			\mathbb{E}[\mathrm{ISI}_i^L] = 
			\begin{cases}
				\frac{1}{2}\sum\limits_{k=2}^i p^{\mathrm{ch}}_k, & 1\leq i\leq \tau\\
				\frac{1}{2}\sum\limits_{k = i-\tau+1}^{i}p^{\mathrm{ch}}_k, & \tau<i \leq \tau+d_1\\
				\frac{1}{2}\left(\sum\limits_{k = i-\tau+1}^{i}p^{\mathrm{ch}}_k + \sum\limits_{j = 1}^{\phi-1} p^{\mathrm{ch}}_{\theta_i^j}\right), & \tau+d_1< i\leq n,
			\end{cases}
		\end{align}
		where $\phi$ can be determined by the inequality $\tau+\sum_{m = 1}^{\phi}d_m > i \geq \tau+\sum_{m = 1}^{\phi-1} d_m$, and $\theta_i^j = (i+1)-(\tau+\sum_{m = 1}^j d_m)$.
	\end{lemma}
	
	\begin{proof}
		From Lemma \ref{constr4}, consider the maximum weighted codeword $\c = c_1c_2\ldots c_n$ in a $[\tau+\sum_{i=1}^qd_i,q+\tau]$ LOZP code $\mathcal{C}^{\tau}_{q}$. If the $i$-th transmitted bit in the codeword $\c$ is a bit-1, i.e., $c_i = 1$ then $i\in D_4$, where $D_4$ is the set containing all possible bit-1 locations in the sequence $\c$, $i.e.,$
		$D_4 = \left\{1,2,\ldots,\tau, \tau+d_1,\tau+d_1+d_2,\ldots,\tau+\sum_{m=1}^{q}d_m\right\}$.
		Now, to derive the $i$-th bit ISI expression, we split the codeword into three regions: (i) the first consecutive bit-1 block, (ii) the first zero padding block ($\phi = 1$), and (iii) the second zero padding block onwards ($\phi \geq 2$). The details are given as follows:
		\begin{itemize}
			\item \textbf{Case 1} {($1\leq i \leq \tau$)}: From the code $\mathcal{C}^{\tau}_{q}$, the first $i ~(\leq \tau)$ positions of the maximum weighted codeword are always bit-1. Hence, the ISI on the initial $i~ (\leq \tau)$ positions in the codeword $\c$ is $\mathrm{ISI}_i =\sum_{k=2}^i p^{\mathrm{ch}}_k$. 
			\item \textbf{Case 2}{ ($\tau< i \leq \tau+d_1$)}:  The first zero padding block ($\phi = 1$) only experiences ISI from the first $\tau$ consecutive bit-1s. Therefore, the ISI on the $i~ (> \tau)$-th bit is $\sum_{k = i-\tau+1}^{i}p^{\mathrm{ch}}_k$.
			\item \textbf{Case 3} {($\tau+d_1< i \leq n$)}:  We can split the analysis into the following two sub-cases:
			\begin{itemize}
				\item \textbf{Sub-case 1:} {(Contribution from the first $\tau$ bit-1s):} ISI on the $i\ (>\tau)$-th bit from the first $\tau$ number of consecutive bit-1s is $\sum_{k = i-\tau+1}^{i}p^{\mathrm{ch}}_k$.
				\item \textbf{Sub-case 2} {(Contribution from the remaining bit-1s after the $\tau$-th position):} We first compute the number of bit-1s ($\phi-1$) before the $\phi ~(\geq 2)$-th zero padding block and after the $\tau$-th position. The value of $\phi$ can be determined by the inequality $\tau+\sum_{m = 1}^{\phi}d_m \geq i > \tau+\sum_{m = 1}^{\phi-1} d_m$ for $\phi = 2,3,\ldots,q$. Therefore, the ISI on the $i$-th ($\tau+d_1< i\leq n$) bit of the $\phi$-th zero padding block from the previous $(\phi-1)$ bit-1s is $\sum_{j = 1}^{\phi-1} p^{\mathrm{ch}}_{\theta_i^j}$, where $\theta_i^j = (i+1)-(\tau+\sum_{m = 1}^j d_m)$.
			\end{itemize}   
			Therefore, combining Sub-case 1 and Sub-case 2, we obtain the desired expected ISI expression on the $i$-th bit for $\tau+1\leq i\leq n$, as outlined in \eqref{avg_isi_lozp_code}.
		\end{itemize}
		Also, note that the average density of bit-1 in the linear LOZP code $\mathcal{C}^{\tau}_{q}$ is always 0.5 for $i\in D_4$.
		This concludes the proof for all three regions over $1\leq i \leq n$.
	\end{proof}
	Note that, for a given codeword length, an increase in the value of $\tau$ corresponds to an increase in the codeword weight and therefore an increase in the ISI as well. Conversely, when the number of consecutive bit-0s increases, followed by a sequence of bit-1s, the impact of ISI diminishes on the subsequent bits. 
	Hence, a distinct trade-off exists between the expected ISI and the prevalence of leading bit-1s, followed by the number of bit-0s in the sequence. 
	Therefore, it is essential to consider the following metrics for the codeword :
	\begin{enumerate}
		\item the expected ISI experienced by all the bit-0s in the code,
		\item the cumulative ISI experienced by all the bit-0s in the codeword that experiences the maximum ISI, and
		\item the maximum ISI experienced by a bit-0 in the codeword corresponding to the maximum ISI.
	\end{enumerate}
	
	We first consider an $(n,\mathcal{S})$ ISI-mtg code \cite{8972472} to compute these parameters on the expected ISI, as this is the best-known channel code available in the literature in terms of ISI mitigation while satisfying the zero padding constraint.
	Now, we first need to identify the codeword where bit-0 experiences the maximum ISI in the ISI-mtg code. In that case, the most affected bit-0 is the first bit-0 followed by the maximum number of possible bit-1s.
	Similarly, for the proposed LOZP code, the most affected bit-0 can either be the first bit-0 followed by a consecutive bit-1 block or the first bit-0 of the last zero padding block with the highest number of preceding bit-1s. In Remark \ref{isi_mtg_bound}, we first compute the closed-form expressions for specified metrics on ISI for the ISI-mtg code \cite{8972472}. Subsequently, we compute the same ISI metrics on bit-0 on the proposed LOZP code in Remark \ref{isi_lozp_code_bound}, to compare the performance with the ISI-mtg code.
	
	\begin{remark}\label{isi_mtg_bound}
		For any positive integers $n$ and $\mathcal{S}$, consider an $(n,\mathcal{S})$ ISI-mtg code \cite{8972472}, where the codeword $\c$ of length $n$ and weight $w(n)$ experience the maximum ISI in a channel with refresh. Therefore, for a channel memory length of $L$, the ISI experienced by all the bit-0s and the maximum ISI experienced by a bit-0 corresponding to the codeword $\c$ are  
		\begin{equation}
			\mathrm{ISI}^{L,0}(\c) = \begin{cases}
				\sum\limits_{k = 1}^{w(n)}\frac{(n-2k+1)}{2}p^{\mathrm{ch}}_{2k}, & \mbox{for odd } n\\
				\sum\limits_{k=1}^{w(n)}\left(p^{\mathrm{ch}}_{2k+1}+\frac{(n-2k)}{2}p^{\mathrm{ch}}_{2k}\right), & \mbox{for even } n,  
			\end{cases}
			\label{isi_total_0_isi_mtg}
		\end{equation}
		\begin{equation}
			\mathrm{ISI}_{\mathrm{max}}^{L,0}(\c) = \sum\limits_{k=1}^{w(n)}p^{\mathrm{ch}}_{2k}, 
			\label{isi_max_0_isi_mtg}
		\end{equation}
		where
		\[
		w(n) = \begin{cases}
			{\frac{n-1}{2}}, & \mbox{ for odd } n\\
			{\frac{n-2}{2}}, & \mbox{ for even } n.
		\end{cases}
		\]
	\end{remark}
	
	\begin{remark}\label{isi_lozp_code_bound}
		For a given positive $q$ integers $d_1\geq d_2\geq \ldots \geq d_q$ ($\geq 2$), consider a $[\tau+\sum_{i=1}^qd_i,q+\tau]$ LOZP code $\mathcal{C}^{\tau}_q$, where the codeword $\c$ of length $n$ and weight $w(n)$ experience the maximum ISI in a channel with refresh. Therefore, for a given channel memory length of $L$ and the weight of the codeword $\c$, the ISI experienced by all the bit-0s and the maximum ISI experienced by a bit-0 corresponding to the codeword $\c$ are
		\begin{equation}
			\mathrm{ISI}^{L,0}(\c) = \sum_{j\in \mathscr{B}} \mathrm{ISI}^L_j(\c), \mbox{ and}
			\label{isi_total_0_lozp}
		\end{equation}
		\begin{equation}
			\mathrm{ISI}_{\mathrm{max}}^{L,0}(\c) = \max\left\{\sum\limits_{k=2}^{\tau+1}p^{\mathrm{ch}}_k,\sum\limits_{k = i-\tau+1}^{i} p^{\mathrm{ch}}_k+\sum\limits_{j = 1}^{q-1}p^{\mathrm{ch}}_{\theta_i^j}\right\},
			\label{isi_max_0_lozp}
		\end{equation}
		where $w(n) = \tau+q-1$.
		In the $\mathrm{ISI}^{L,0}(\c)$ expression, $\mathscr{B}$ denotes the bit-0 locations in the codeword $\c$, $i.e.,$ $\mathscr{B} = \{n\}\cup \mathscr{A}\backslash D_4$, where $\mathscr{A} = \{1,2,\ldots,n\}$, $D_4 = \{1,2,\ldots,\tau,\tau^*: \tau^* = \tau+\sum_{m = 1}^{\phi} d_m \mbox{ for } \phi = 1,2,\ldots,q\}$, and $\mathrm{ISI}^L_j(\c)$ can be computed from Lemma \ref{isi_expected_lozp_code}. 
		In \eqref{isi_max_0_lozp}, the first term $\sum_{k=2}^{\tau+1}p^{\mathrm{ch}}_k$ corresponds to $i = \tau+1$ and the second term $\sum_{k = i-\tau+1}^{i} p^{\mathrm{ch}}_k+\sum_{j = 1}^{q-1}p^{\mathrm{ch}}_{\theta_i^j}$ corresponds to $i = 1+\tau+\sum_{m=1}^{q-1}d_m$, where $\theta_i^j = 2+\sum_{m=j+1}^{q-1}d_m$.
	\end{remark}
	We consider the following example to compute the parameters given in Remark \ref{isi_mtg_bound} and Remark \ref{isi_lozp_code_bound} for the linear LOZP code and ISI-mtg-code \cite{8972472}, respectively.
	
	\begin{example}\label{example_isi_mtg_isi_lozp_compare}
		Consider a $[9, 4]$ LOZP code $\mathcal{C}^2_{2}$ with $q = 2$, $d_1 = 3$, $d_2 = 2$, channel memory $L = 8$ (channel refresh), sampling time ($t_s$ = 0.3s) and other parameters from Table \ref{tab1} to compute the coefficients $p^{\mathrm{ch}}$. Note that the codeword $\c = 110001000$ of weight three ($\tau + q-1 = 3$) experiences the maximum ISI among all the codewords. Consequently, from \eqref{isi_total_0_lozp}, the total ISI experienced by all the bit-0s in the codeword $\c$ is $\mathrm{ISI}^{8,0} = 2p^{\mathrm{ch}}_{2}+3p^{\mathrm{ch}}_{3}+3p^{\mathrm{ch}}_{4}+p^{\mathrm{ch}}_{5}+p^{\mathrm{ch}}_{6}+2p^{\mathrm{ch}}_{7}+2p^{\mathrm{ch}}_{8}+p^{\mathrm{ch}}_{9} = 0.3640$. We also observe that the first bit-0 $((\tau+1)$-th position$)$, followed by the consecutive bit-1 block, is the most affected bit-0 by ISI in the codeword $\c$. Therefore, From \eqref{isi_max_0_lozp}, we can determine $\mathrm{ISI}_{\mathrm{max}}^{8,0}(\c) = p^{\mathrm{ch}}_{2}+p^{\mathrm{ch}}_{3} = 0.1035$.
		
		Now, let us consider a $(9, 54)$ ISI-mtg code, where we observe that the codeword $\c = 010101010$ of weight $4$ experiences the maximum ISI among all the codewords. Consequently, from \eqref{isi_total_0_isi_mtg}, the total ISI experienced by all the bit-0s in the codeword is $\mathrm{ISI}^{8,0}(\c) = 4p^{\mathrm{ch}}_{2}+3p^{\mathrm{ch}}_{4}+2p^{\mathrm{ch}}_{6}+p^{\mathrm{ch}}_{8}$ = $0.3698> 2p^{\mathrm{ch}}_{2}+3p^{\mathrm{ch}}_{3}+3p^{\mathrm{ch}}_{4}+p^{\mathrm{ch}}_{5}+p^{\mathrm{ch}}_{6}+2p^{\mathrm{ch}}_{7}+2p^{\mathrm{ch}}_{8}+p^{\mathrm{ch}}_{9}$ = $0.3640$. Similarly, we can observe that the last bit-0 experiences the maximum ISI among all the bits, which can be computed from \eqref{isi_max_0_isi_mtg} and is  $\mathrm{ISI}_{\mathrm{max}}^{8,0}(\c) = p^{\mathrm{ch}}_{2}+p^{\mathrm{ch}}_{4}+p^{\mathrm{ch}}_{6}+p^{\mathrm{ch}}_{8}$ = $0.1081> p^{\mathrm{ch}}_{2}+p^{\mathrm{ch}}_{3}$ = $0.1035$.		
	\end{example}		
	
	As observed in Example \ref{example_isi_mtg_isi_lozp_compare}, the codewords in the LOZP code can have consecutive bit-1s up to a certain position while still achieving a lower ISI compared to the ISI-mtg code. Also, for a specific sampling interval $t_s$ and any positive integers $b_1$, $b_2\ (< b_1)$, $b_3$ and $i$, if there exists an $i$, $s.t.,$ $b_1p^{\mathrm{ch}}_i\geq b_2p^{\mathrm{ch}}_{i}+b_3p^{\mathrm{ch}}_{i+1}$ then the codewords in the code $\mathcal{C}^{\tau}_{q}$ can exhibit consecutive bit-1s while maintaining the same ISI level as that of an ISI-mtg code. Therefore, we can deduce Remark \ref{lemma_lozp_exist_for_isi_mtg} using Remark \ref{isi_mtg_bound} and Remark \ref{isi_lozp_code_bound}.
	\begin{remark}\label{lemma_lozp_exist_for_isi_mtg}
		For any positive integers $d_i ~(\geq 2)$, $q$, and $\tau$ $(\geq 2)$, consider a $[n,q+\tau]$ LOZP code $\mathcal{C}^{\tau}_{q}$ where $n = \tau+\sum_{i=1}^qd_i$. Now, for any feasible sampling time $t_s$, as given in \cite{7331300}, if 
		\begin{itemize}
			\item $\sum_{k=3}^{\tau+1}p^{\mathrm{ch}}_k < \sum_{k = 2}^{w(n)}p^{\mathrm{ch}}_{2k}$, and
			\item $\sum_{k = i-\tau+1}^{i} p^{\mathrm{ch}}_k+\sum_{j = 1}^{q-1}p^{\mathrm{ch}}_{\theta_i^j}<\sum_{k = 1}^{w(n)}p^{\mathrm{ch}}_{2k}$
		\end{itemize}
		then there exists a codeword $\c$ with maximum weight $\tau+q-1$ in the LOZP code $\mathcal{C}^{\tau}_{q}$ such that $\mathrm{ISI}^{L,0}(\c)$ and $\mathrm{ISI}^{L,0}_{\mathrm{max}}(\c)$ are less than that of the maximum weighed codeword in an $(n,\mathcal{S})$ ISI-mtg code.
	\end{remark}

	Now, we need to check whether there are any feasible solutions of $\tau$ and $q$ such that the conditions from Remark \ref{lemma_lozp_exist_for_isi_mtg} hold. For instance, if $\tau = 2$ then the first inequality reduces to $p^{\mathrm{ch}}_{4} + p^{\mathrm{ch}}_{6} + \ldots + p^{\mathrm{ch}}_{2w(n)}>p^{\mathrm{ch}}_3$, and this inequality holds for a large value of $w(n)$. 
	Also, for a given value of $\tau$ and $d_i = c$ (any positive integer $c ~(\geq 2)$), $q = \lfloor(n-\tau)/c\rfloor)$, which reduces the second inequality of Remark \ref{lemma_lozp_exist_for_isi_mtg} into $p^{\mathrm{ch}}_{2} + p^{\mathrm{ch}}_{4} + p^{\mathrm{ch}}_{6} + \ldots + p^{\mathrm{ch}}_{2w(n)}> p^{\mathrm{ch}}_{3+(q-1)c} + p^{\mathrm{ch}}_{1+(q-1)c} + p^{\mathrm{ch}}_{2+(q-3)c} + p^{\mathrm{ch}}_{2+(q-3)c} + \ldots + p^{\mathrm{ch}}_{2}$. Hence, this condition holds if the value of $q$ is sufficiently small for a given $n$ and $\tau$, which leads to a large $c$. For instance, if $q = 3$ then the inequality $p^{\mathrm{ch}}_{4} + p^{\mathrm{ch}}_{6} + \ldots + p^{\mathrm{ch}}_{2w(n)}> p^{\mathrm{ch}}_{1+2c} + p^{\mathrm{ch}}_{3+2c} + p^{\mathrm{ch}}_{2+c} + p^{\mathrm{ch}}_{2+2c}$ holds for a large $w(n)$ with any $c\geq 3$ and at a feasible value of $t_s$.       
	From Fig. \ref{fig_isi_compare_mtg_lozp}, note that for a sufficiently large-length codeword, there exists at least a solution of $\tau$ and $q$ such that the ISI experience by the $(\tau+1)$-th bit and the $(n-d_i+1)$-th bit in the LOZP code is less than the bit-0 experiencing the maximum ISI in ISI-mtg code. For example, if $n = 9$ and $n = 55$, then $\tau = 2$ and $2\leq \tau \leq 4$ with $d_i = 3$ satisfy Remark \ref{lemma_lozp_exist_for_isi_mtg} for $t_s = 0.2$s, respectively.
			\begin{figure}
		\centering
		\includegraphics[width = 1\linewidth]{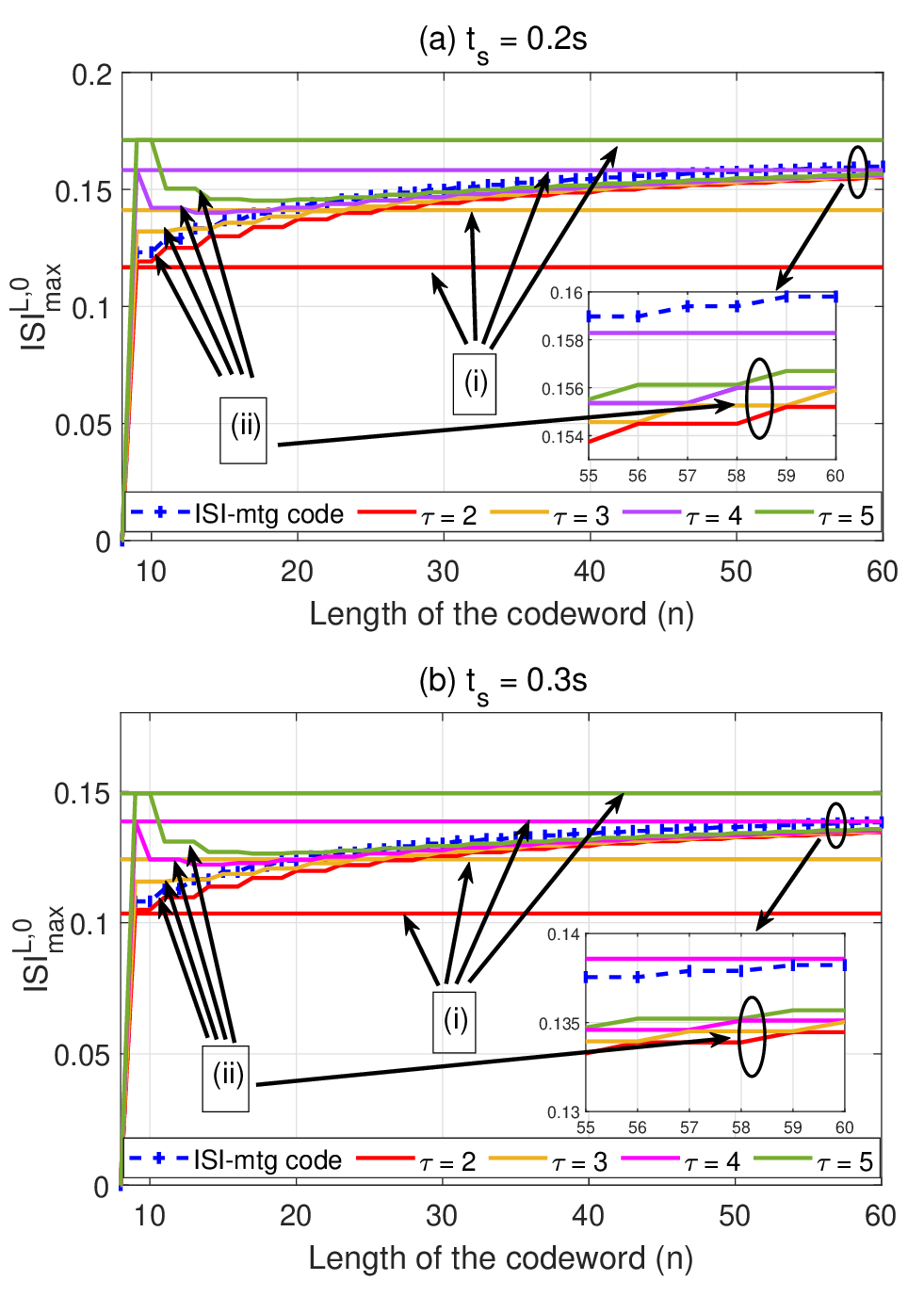}
		\caption{Comparison of the maximum ISI experienced by a bit-0 for ISI-mtg code $\textbf{CW}_n$ and LOZP code $\mathcal{C}^{\tau}_{q}$, where the term (i) represents $(\tau+1)$-th bit ISI and the term (ii) denotes $(n-d_q+1)$-th bit ISI in the LOZP code $\mathcal{C}^{\tau}_{q}$ with different sampling times $t_s$ (0.2s and 0.3s), consecutive bit-1 block size $\tau$ and zero padded block size $d_i=3$ $\forall i$ (with channel refresh and $L = n-1$).}
		\label{fig_isi_compare_mtg_lozp}
	\end{figure}
	\begin{figure}
		\centering
		\includegraphics[width = 0.86\linewidth]{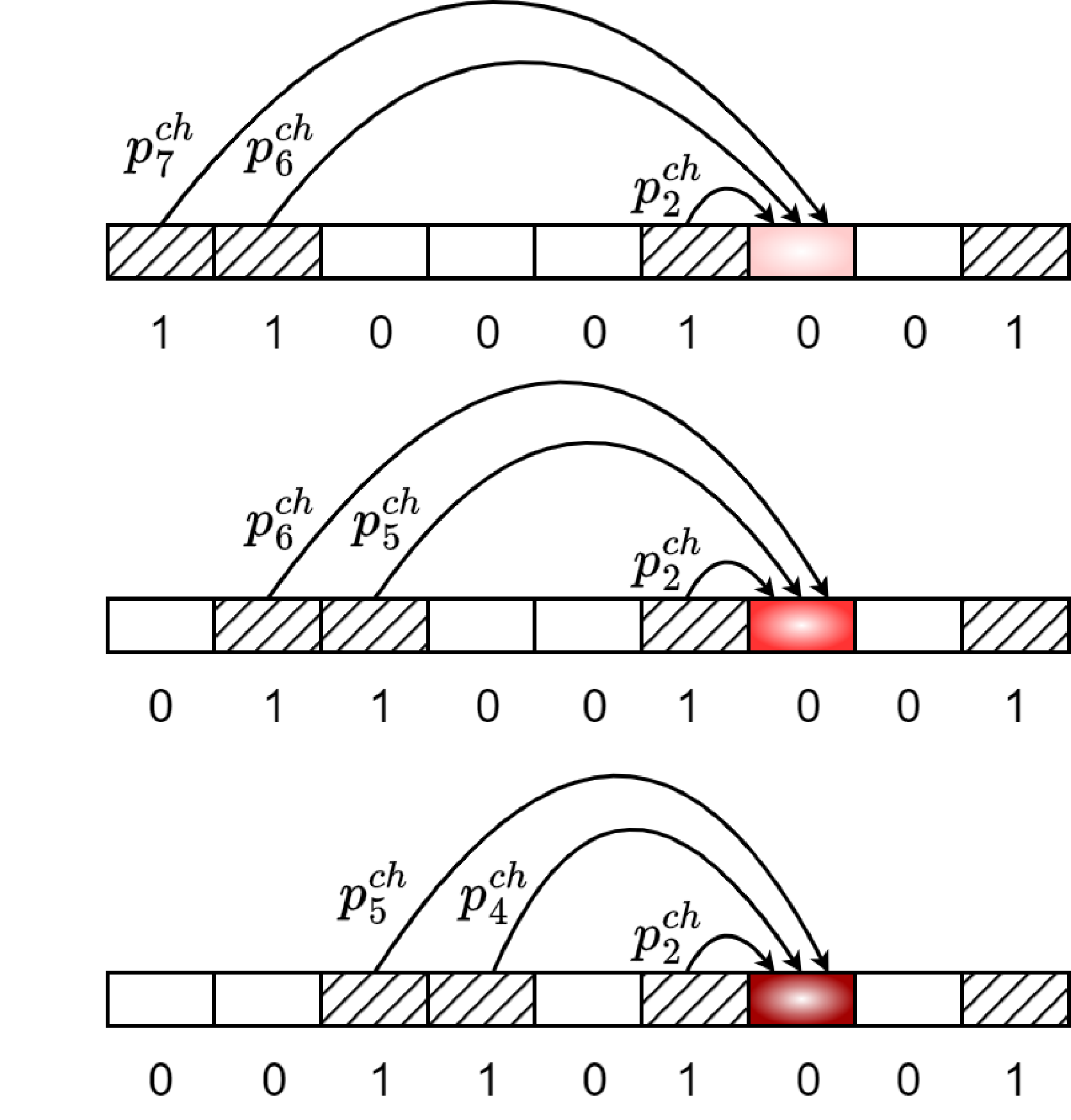}
		\caption{Comparison of ISI effect on the first bit-0 of the last zero padding block as the consecutive bit-1 block shifts rightwards with $L = 8$ (with channel refresh).}
		\label{fig_isi_effect_cons_bit1_block}
	\end{figure}
	From Fig. \ref{fig_isi_effect_cons_bit1_block}, we observe that as the consecutive bit-1 block $\mathbf{1}_{1,2}$ shifts right within any sequence $\mathbf{c} = 1 1 0 0 0 1 0 0 1$, the ISI effect on the $7$-th bit, i.e., the first bit-0 of the last zero padding block, increases. This holds for other bit-0 locations following the consecutive bit-1 block, as $p^{\mathrm{ch}}_k$ is a decreasing function of $k$ and the length of the zero padding block decreases.
	Therefore, from Remark \ref{isi_lozp_code_bound}, we can deduce the following lemma illustrating the advantage of the LOZP code over the OMP and OEP codes, considering a given length and maximum weight.
	
	\begin{figure*}[h!]
		\centering
		\includegraphics[width=1\linewidth]{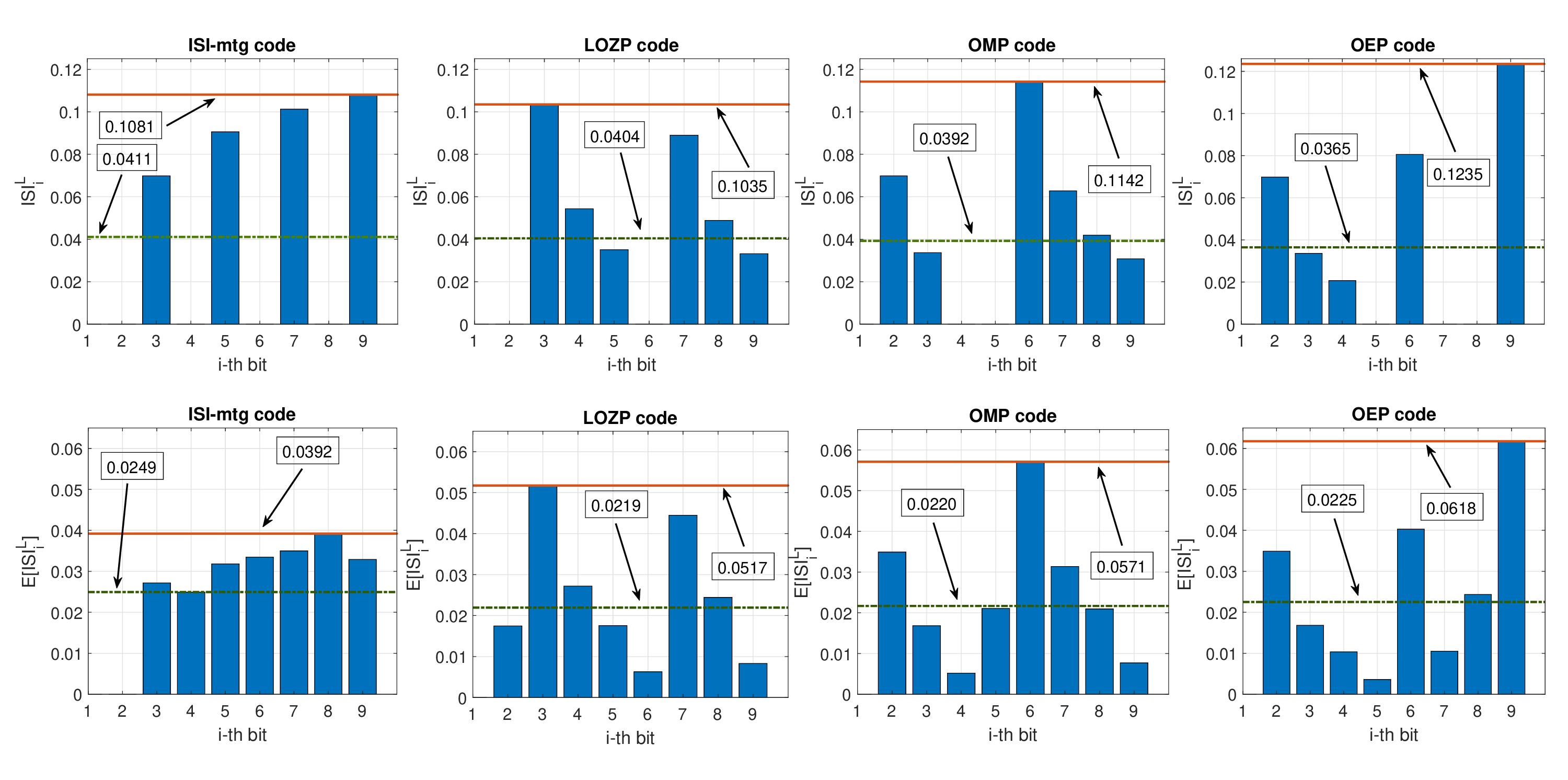}  
		\caption{Comparison of the $i$-th bit ISI with the codeword experiencing the maximum ISI and the $i$-th bit expected ISI of the respective binary codes with length $ (n) = 9$, maximum weight = 4  and channel memory $L = n-1$ at $t_s = 0.3$s.}
		\label{avg_isi_max_isi_bar_fig}
	\end{figure*}
	
	\begin{lemma}\label{bit_1_position_lemma}
		For the given positive integers $d_i\ (\geq 2)$, $q$ and $\tau\ (\geq 2)$, consider a codeword $\c$ of weight $\tau+q-1$ in a $[\tau+\sum_{i=1}^qd_i,q+\tau]$ LOZP code $\mathcal{C}^{\tau}_q$. Considering a channel with refresh, for any consecutive bit-1 block $\c(j,j+\tau-1) = \mathbf{1}_{1,\tau}$, $\mathrm{ISI}^{L,0}_{\mathrm{max}}(\c)$ increases with $j$ for $j = 1,2,\ldots,n-\tau-2q+1$. 
	\end{lemma}
	\begin{proof}
		For any positive integer $\tau\ (\geq 2)$, consider that the consecutive bit-1 block $\c(j,j+\tau-1) = \mathbf{1}_{1,\tau}$ shifts rightwards of the codeword. Therefore, from the $\mathrm{ISI}_{\mathrm{max}}^{L,0}(\c)$ expression of Remark \ref{isi_lozp_code_bound}, the maximum ISI experienced by the first bit-0 of the $q$-th zero padding block increases. This happens because of the ISI contributing term $\sum_{k = i-\tau+1}^{i} p^{\mathrm{ch}}_k+\sum_{j = 1}^{q-1}p^{\mathrm{ch}}_{\theta_i^j}$ increases as $d_q$ decreases for a constant $n$-length sequence, where $i = 1+\tau+\sum_{m=1}^{q-1}d_m$, $\theta_i^j = 2+\sum_{m=j+1}^{q-1}d_m$ and $n= \tau+\sum_{m=1}^q d_m$.
		
		Hence, for any consecutive bit-1 block $\c(j,j+\tau-1) = \mathbf{1}_{1,\tau}$ in the LOZP code $\mathcal{C}^{\tau}_q$, $\mathrm{ISI}^{L,0}_{\mathrm{max}}(\c)$ increases with $j$ for $j = 1,2,\ldots,n-\tau-2q+1$.
	\end{proof}
	
	From Lemma \ref{bit_1_position_lemma}, we can conclude that positioning $\tau$ number of bit-1s at the start of the sequence is essential to minimize the expected ISI on a bit-0 within the code $\mathcal{C}^{\tau}_q$. Fig. \ref{avg_isi_max_isi_bar_fig} presents a comparative analysis of the maximum ISI experienced by a codeword and the expected ISI between the ISI-mtg code, LOZP code, linear OMP code and linear OEP code for a given length ($n = 9$) and maximum weight (4). In this figure, we consider the generator matrices
	\begin{align*}
		&  G_1= \begin{bmatrix}
			\mathbf{I}_{2} & \mathbf{0}_{2,4} & \mathbf{0}_{2,3}\\
			\mathbf{0}_{1,5} & 1 & \mathbf{0}_{1,3}\\
			\mathbf{0}_{1,7} & 0 & 1\\
		\end{bmatrix},     
		G_2= \begin{bmatrix}
			1 & \mathbf{0}_{1,5} & \mathbf{0}_{1,3}\\
			\mathbf{0}_{2,4} & \mathbf{I}_{2} & \mathbf{0}_{2,3}\\
			\mathbf{0}_{1,7} & 0 & 1\\
		\end{bmatrix}, \mbox{ and }\\
		&\hskip 2cm G_3= \begin{bmatrix}
			1 & \mathbf{0}_{1,7} & 0\\
			\mathbf{0}_{1,4} & 1 & \mathbf{0}_{1,4}\\
			\mathbf{0}_{2,6} & \mathbf{I}_{2} &  \mathbf{0}_{2,1}\\
		\end{bmatrix}, 
	\end{align*}
	to construct the linear LOZP, OMP and OEP codes, respectively. Now, the codeword experiencing the maximum ISI for the ISI-mtg, LOZP, OMP and OEP codes are respectively $010101010,110001001,1000110000,100010110$. In the LOZP and OMP codes, the third and sixth bits in the codeword experience the highest ISI, respectively. Conversely, in the ISI-mtg and OEP code, it is the last bit that encounters the most ISI resulting in a probable bit-flip and subsequent error. From Fig. \ref{avg_isi_max_isi_bar_fig}, it is evident that among all the mentioned codes 
	\begin{itemize}
		\item the LOZP code exhibits the lowest expected ISI before transmitting a bit-0, and
		\item the codeword, subject to the most ISI in the LOZP code, experiences the least maximum ISI among the ISI-mtg, the linear LOZP, OMP and OEP codes, thereby leading to an enhanced BER performance.
	\end{itemize}
	
	\section{Encoding and Decoding of ZP and LOZP Codes}\label{Sec 4}
	This section discusses the encoding of linear ZPZS and non-linear ZP codes, followed by a simple location-based decoding mechanism of the proposed ZP codes.
	\subsection{Encoding Technique}
	For the linear ZP codes $\mathcal{C}_{d,q}$, $T(\mathcal{C}_{d,q})$, $\mathcal{C}_q$ and $T(\mathcal{C}_q)$, any message bit sequence can be encoded using the generator matrices $G_{d,q}$, $T(G_{d,q})$, $G_q$ and $T(G_q)$, respectively. Similarly, for the linear LOZP codes $\mathcal{C}^{\tau}_{d,q}$ and $\mathcal{C}^{\tau}_{q}$, any message bit sequence can be encoded using the generator matrices $G^{\tau}_{d,q}$ and $G^{\tau}_q$ respectively. 
	
	In general, for the encoding mechanism, consider an $(n,\mathcal{S})$ ZP binary code $\mathcal{C}\cup T(\mathcal{C})$, where the generator matrix of code $\mathcal{C}$ is $G$. 
	Also, consider that the message bit sequence $\textbf{m}$ = $m_1m_2\ldots m_{k+1}$ of length $k+1$ is encoded into the sequence $\mathbf{c}$ of length $n$ such that $\mathbf{c}$ = $[m_2\ m_3 \ldots\ m_{k+1}]\cdot G^*$, where $G^*$ = $G$ for $m_1=0$, and $G^*$ = $T(G)$ for $m_1=1$.
	Now, both the sequence $\textbf{0}_{1,k+1}$ and $1\textbf{0}_{1,k}$ of length $k+1$ is encoded to the sequence $\textbf{0}_{1,n}$. Therefore, we choose the code parameters $d_i$ and $q$ in such a way that we can avoid the all-zero message sequence ($\textbf{0}_{1,k+1}$) in the MCvD channel. 
	
	\subsection{Decoding Technique}
	For the linear ZP codes ($\mathcal{C}_{d,q}$, $T(\mathcal{C}_{d,q})$, $\mathcal{C}_q$, $T(\mathcal{C}_q)$) and linear LOZP codes ($\mathcal{C}^{\tau}_{d,q}$,  $\mathcal{C}^{\tau}_{q}$), the decoding can be done using the simple parity check matrix \cite{983680}. 
	However, we must systematically elaborate on the decoding techniques to decode the message using ZP codes.
	For decoding of the received sequence $\textbf{r}$ = $r_1r_2\ldots r_n$ of length $n$, consider $\textbf{r}$ is encoded by a ZP code $\mathcal{C}\cup T(\mathcal{C})$ where the generator matrix of the linear code $\mathcal{C}$ is denoted by $G$.
	Now, there are two cases as follows.
	\begin{enumerate}
		\item Case 1 ($\textbf{r}\in\langle G\rangle\cup\langle T(G)\rangle$): For any $\textbf{r}\neq\textbf{0}_{1,n}$ and $\textbf{r}\in\mathcal{C}_q$, one can find $[m_2\ m_3\ \ldots\ m_{k+1}]$ = $\textbf{r}\cdot (G^*)^t$, and the decoded sequence is $\mathbf{\hat{m}}$, where 
		\begin{align}
			\mathbf{\hat{m}}=
			\begin{cases}
				0\hat{m}_2\hat{m}_3\ldots\hat{m}_{k+1} & \mbox{if } r_{2+\chi(j)} = 1,  G^*=G, \\
				1\hat{m}_2\hat{m}_3\ldots \hat{m}_{k+1} & \mbox{if }r_{1+\chi(j)} = 1,  G^*=T(G),
			\end{cases}
		\end{align}
		
		where $\chi(j)=\sum_{i=1}^j d_i$ for $j = 0,1,\ldots,q$, the transpose of the matrix (or the array) $A$ is $A^t$, and $r_{qd+f}$ denotes the $(qd+f)$-th position in the received sequence $\textbf{r}$ for $f \in \{1,2\}$.
		\item Case 2 ($\textbf{r}\notin \langle G\rangle\cup\langle T(G)\rangle$): For any non-negative integers $q$ and $d$ $(\geq 2)$, and $j = 0,1,\ldots,q$, the number of bit-1s at $(2+\sum_{i=1}^jd_i)$-th position (denoted by $\mathcal{K}_2$), and at $(1+\sum_{i=1}^jd_i)$-th position (denoted by $\mathcal{K}_1$) are computed as follows.
		\begin{align}\label{loc_bit_1}
			&\mathcal{K}_1 = r_1+\sum\limits_{\ell=1}^q r_{1+\chi(\ell)}, \mbox{ and }  \mathcal{K}_2 = r_2 + \sum\limits_{\ell=1}^q r_{2+\chi(\ell)},
		\end{align}
		where $\chi(\ell)=\sum_{i=1}^\ell d_i$.
		There are two cases, depending on the conditions on $\mathcal{K}_1$ and $\mathcal{K}_2$.
		\begin{enumerate}
			\item Sub-Case 1 ($\mathcal{K}_1\neq\mathcal{K}_2$): For any non-zero sequence $\mathbf{r}$, one can find $[m_2\ m_3\ \ldots\ m_{k+1}]$ = $\textbf{r}\cdot(G^*)^t$, and the decoded sequence is $\mathbf{\hat{m}}$, where $\mathbf{\hat{m}}$ = $0\hat{m}_2\hat{m}_3\ldots \hat{m}_{k+1}$ for $\mathcal{K}_2>\mathcal{K}_1$ and $G^*=G$, and $\mathbf{\hat{m}}$ = $1\hat{m}_2\hat{m}_3\ldots \hat{m}_{k+1}$ for $\mathcal{K}_1>\mathcal{K}_2$ and $G^*=T(G)$.
			We define this decoding mechanism as Majority Location Rule (MLR) decoding in Algorithm \ref{Encoding Algo}. This algorithm checks whether the maximum number of bit-1s are coming from the generator matrix $G$ or $T(G)$ and then selects the appropriate generator matrix for decoding.
			
			\item Sub-Case 2  ($\mathcal{K}_1$ = $\mathcal{K}_2$): In this Sub-Case, we can not employ the MLR decoding operation directly. We first transform the received sequence $\mathbf{r}$ into $\Tilde{\mathbf{r}}$ using \eqref{map_11_to_10}, which is defined as a pre-decoding operation. Once, we obtain $\Tilde{\mathbf{r}}$, one can similarly find $[m_2\ m_3 \ldots\ m_{k+1}]$ = $\Tilde{\mathbf{r}}\cdot(G^*)^t$ from the MLR algorithm.	
		\end{enumerate}
	\end{enumerate}
	
	\begin{algorithm} 
		\caption{Majority Location Rule Decoding for $(2+\sum_{i=1}^{q} d_i,2^{(q+2)}-1)$ ZP Code $\mathcal{C}\cup T(\mathcal{C})$} 
		\begin{algorithmic} 
			\REQUIRE Received sequence $\textbf{r}$, Generator matrix $G$ and $\mathbf{d} = [d_1\ d_2\ \ldots\ d_q]$
			\ENSURE Message bit sequence $\mathbf{\hat{m}}$
			\STATE$1$: For $i\leq 2+\sum_{j=1}^{q}d_j$, if $r_i=1$ then $\mathrm{loc} = i$
			\STATE$2$: Compute $\mathcal{K}_1$ and $\mathcal{K}_2$ from \eqref{loc_bit_1}
			\STATE$3$: If $\mathcal{K}_2>\mathcal{K}_1$ then $\mathbf{\hat{m}} = [\hat{m}_1 ~\mathbf{\hat{m}}']$, where $\hat{m}_1 = 0$, $\mathbf{\hat{m}}'$ = $\mathbf{r}\cdot G^t$, and if $\mathcal{K}_1>\mathcal{K}_2$  then $\mathbf{\hat{m}} = [\hat{m}_1 ~\mathbf{\hat{m}}']$, where $\hat{m}_1$ $= 1$, $\mathbf{\hat{m}}'$ = $\mathbf{r}\cdot T(G)^t$
			\STATE$4$: If $\mathcal{K}_1=\mathcal{K}_2$ then 
			\begin{itemize}
				\item[4.1] Obtain the transformed sequence from \eqref{map_11_to_10} such that the new sequence does not contain any consecutive bit-1s
				\item[4.2] Go to step 2 and then follow step 3 until $\mathcal{K}_1\neq \mathcal{K}_2$ satisfies
			\end{itemize}
		\end{algorithmic}
		\label{Encoding Algo} 
	\end{algorithm}	
	
	\begin{figure*}
		\centering
		\includegraphics[width=0.95\linewidth]{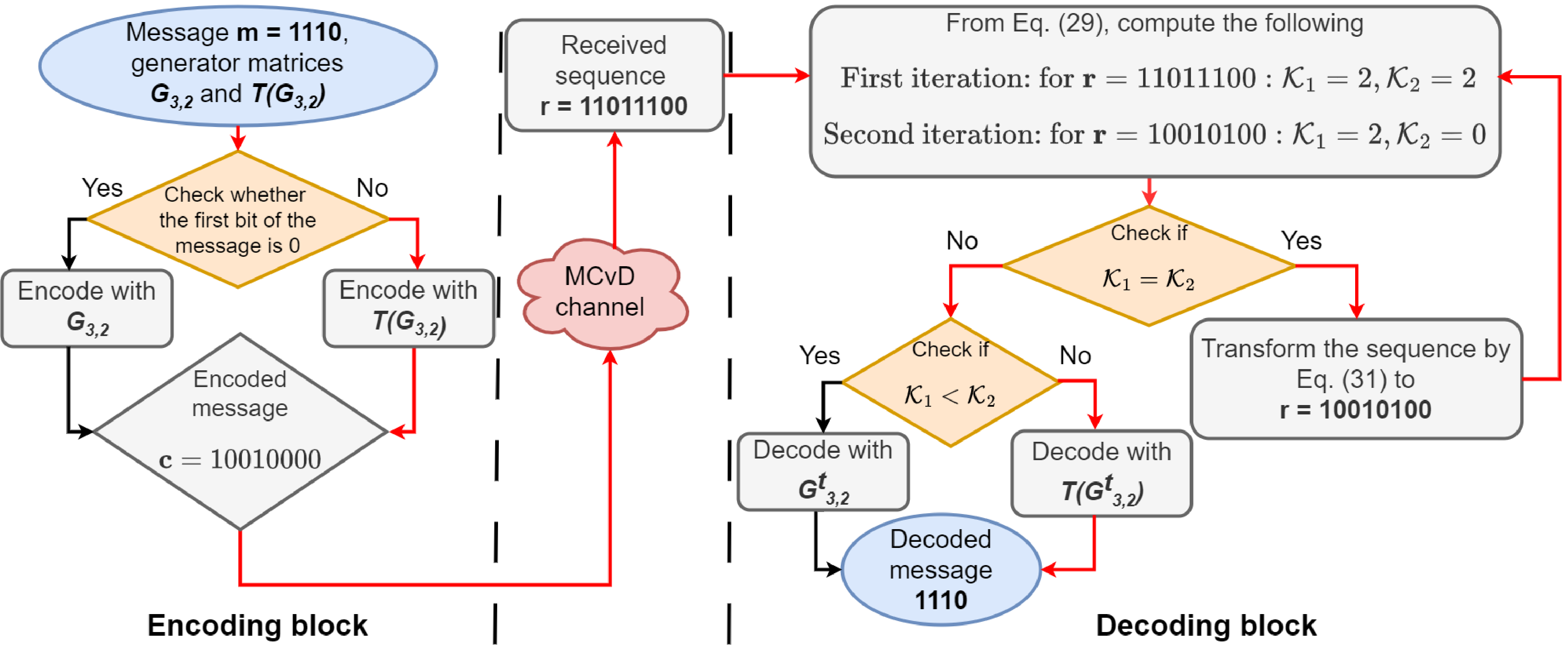}  
		\caption{An example of encoding and MLR decoding of $(8,15)$ ZP code $\mathcal{C}_{3,2}\cup T(\mathcal{C}_{3,2})$ (red line indicates the followed path for the considered example).}
		\label{flowchart_encoding_mlr_decoding}
	\end{figure*}
	
	Now, the effects from previous bits are essential in order to map the received sequence $\mathbf{r}$ into the most probable codeword $\mathbf{c}$. Now, as given in \cite[Eq. (15)]{8972472}, if the $i$-th bit is delayed by $j$ slot then the received number of molecules at the Rx due to these delayed bits is 
	\begin{align}\label{rx_mol_from_prev_time}
		N_{i,j} \!\sim \! \mathcal {N} \left({\! \sum _{k=1}^{i-1}{M  c_kp^{\mathrm{ch}}_{i,k,j} }, \sum _{k=1}^{i-1}{ M   c_kp^{\mathrm{ch}}_{i,k,j}  \left ({1\!-\!p^{\mathrm{ch}}_{i,k,j}  }\right) } \!}\right),
	\end{align}
	for $1\leq i\leq n$ and $0\leq j \leq n$, where $p^{\mathrm{ch}}_{i,k,j} = p^{\mathrm{ch}}_{i-k+j+1}$. Note that if $j = 0$ then the $i$-th transmitted bit arrives at the $i$-th slot itself and therefore, the expression \eqref{rx_mol_from_prev_time} aligns with the expression given in \eqref{rx_molecule_dist}.
	
	In the context of ZP codes, we consider a code $\mathcal{C}_{d,q}$ used for decoding a received bit sequence $\mathbf{r}$ of length $d+2$ with the condition $\mathcal{K}_1=\mathcal{K}_2$. 
	To illustrate the decoding process, consider the example: if the received sequence $\mathbf{r}$ is $11\mathbf{0}_{1,d-2}11$ then the most probable codeword for the received sequence $\mathbf{r}$ of length $n$ can be either $\mathbf{c}_1 = 1\mathbf{0}_{1,d-1}10$ or $\mathbf{c}_2 = 01\mathbf{0}_{1,d-1}1$. 
	The effect of ISI is understood to propagate from a leading bit-1 to the most recent bit-0. This characteristic, combined with the transmission probabilities described in \cite[Eq. (16a), (16b), (16c)]{8972472}, leads to the phenomenon: when two consecutive bit-1s are received, the last bit is transformed into a bit-0. Leveraging these properties leads to the following mapping for sequences of consecutive ones for $i = 1,2,\ldots,n-1$:

\begin{align}\label{map_11_to_10}
	\Tilde{r}_{i+1} =
	\begin{cases}
		0, & \text{if } {r}_{i} = 1 \text{ and } r_{i+1} = 1 \\
		r_{i+1}, & \text{otherwise},
	\end{cases} 
\end{align}
 where the transformed sequence is $\mathbf{\Tilde{r}} = 	\Tilde{r}_1	\Tilde{r}_2 \dots 	\Tilde{r}_n$ with the initial condition $\Tilde{r}_1 = r_1$. 
	Equation \eqref{map_11_to_10} leads to the condition $\mathcal{K}_1 \neq \mathcal{K}_2$. 
	Other received sequences with the condition $\mathcal{K}_1 = \mathcal{K}_2$  can be decoded by either $G^t$ or $T(G^t)$. 
	
	In the encoding algorithm, the dominant complex term involves the multiplication of the message of length $k$ with the generator matrix of dimension $k \times n$.
	Therefore, the complexity of this matrix multiplication is $\mathcal{O}(kn)$, where \( \mathcal{O} \) denotes Big-O notation. While, the decoding complexity is also similar to that of the encoding algorithm due to the matrix multiplication between the received sequence of length $n$ and the transpose of the generator matrix, $i.e.,$ $\mathcal{O}(kn)$.
	Note that for a given rate of the code $R~(=k/n)$, the encoding and decoding complexity depends only on the encoded message length of the sequence, $i.e.,$ $\mathcal{O}(n^2)$. 
	Also, the authors in \cite{10620232} have shown that matrix multiplication can be achieved through different chemical approaches such as controlled substance transport between compartments and precise volume ratios that determine matrix weights.
	Therefore, one can implement the ZP and LOZP codes through the generator and parity check matrix multiplications in an MCvD system.
	
	In Fig. \ref{flowchart_encoding_mlr_decoding}, we consider a $(8,15)$ ZP code $\mathcal{C}_{3,2}\cup T(\mathcal{C}_{3,2})$, where both the encoding and decoding flowchart is discussed with an example ($\mathbf{m} = 1110$).
	For this considered example, it is observed that we can successfully retrieve the original message error-free by our proposed MLR decoding algorithm.
	
	\section{Bounds}\label{Sec 5}
	In this section, we first derive an upper bound on the number of possible codewords for both the proposed families of ZPZS and ZP codes, respectively. We also determine the asymptotic code rate for the ZP codes for different constructions. 
	
	\begin{remark}
		In any ZPZS sequence, each bit-1 is padded with at last one bit-0, and therefore, if the number of ones in a ZPZS sequence is $w$ then $w\leq\lfloor n/2\rfloor$, for any integer $n$ $(>0)$. 
		\label{bound on weight ZPZE}
	\end{remark}
	\begin{lemma}
		For any positive integer $n$, the total number of ZPZS sequences is	$U_n = \sum\limits_{w=0}^{\lfloor n/2\rfloor}\binom{n-w}{w}$.
	\end{lemma}
	\begin{proof}
		Each bit-1 is padded with at least one bit-0 in each ZPZS sequence.
		Therefore, if the number of one's in a ZPZS sequence of length $n$ is $w$ then there are $n-w$ positions where bit-0 can be allotted. 
		There are $\binom{n-w}{w}$ ZPZS sequences with $w$ ones.
		Hence, the result follows from Remark \ref{bound on weight ZPZE}.
	\end{proof}	
	For any $(n,U_n)$ ZPZS code $\mathcal{C}_n$ can be constructed as $\mathcal{C}_n=\{\a01,\b0:\a\in\mathcal{C}_{n-2}\mbox{ and }\b\in\mathcal{C}_{n-1}\}, \mbox{ for }  n\geq 3 $ where $\mathcal{C}_1$ = $\{0\}$ and $\mathcal{C}_2$ = $\{00,01\}$.
	Therefore, for $3\leq n$, $	U_n=U_{n-1}+U_{n-2}$ with the initial conditions $U_1$ = $1$ and $U_2$ = $2$.
	After solving the recurrence relation, one can obtain 
	\begin{equation}
		U_n = \frac{\sqrt{5}-1}{2\sqrt{5}}\left(\frac{1-\sqrt{5}}{2}\right)^n + \frac{\sqrt{5}+1}{2\sqrt{5}}\left(\frac{1+\sqrt{5}}{2}\right)^n,
		\label{Code rate ZPZS sequences}
	\end{equation}
	with the code rate being defined as $R^{U_n} = \frac{1}{n}\log_2 U_n$.
	Hence, from \eqref{Code rate ZPZS sequences}, one can observe $R^{U_n} \to\log_2\left(\frac{\sqrt{5}+1}{2}\right)\approx 0.6942$ as $n\to\infty$. 	
	
	\begin{remark}
		In any ZP sequence, except the last bit, each bit-1 is padded with a bit-0, and therefore, if the number of ones in a ZP sequence is $w$ then $w\leq\lceil\frac{n}{2}\rceil$. 
		And thus, for any positive integer $n$, the total number of ZP sequences is $T_n = \sum_{w = 0}^{\left\lceil n/2\right\rceil}\binom{n-w+1}{w}$.
	\end{remark}
	
	Any $(n,T_n)$ ZP code $\mathcal{C}_n$ can be constructed as $\mathcal{C}_n=\{\a01,\b0:\a\in\mathcal{C}_{n-2}\mbox{ and }\b\in\mathcal{C}_{n-1}\}, \mbox{ for }3\leq n$ where $\mathcal{C}_1$ = $\{0,1\}$ and $\mathcal{C}_2$ = $\{00,10,01\}$.
	Therefore, for $3\leq n$, $	T_n=T_{n-1}+T_{n-2}$ with the initial conditions $T_1$ = $2$ and $T_2$ = $3$.
	After solving the recurrence relation, one can obtain
	\begin{equation}
		T_n = \frac{\sqrt{5}-3}{2\sqrt{5}}\left(\frac{1-\sqrt{5}}{2}\right)^n + \frac{\sqrt{5}+3}{2\sqrt{5}}\left(\frac{1+\sqrt{5}}{2}\right)^n,    
		\label{Code rate ZP sequences}
	\end{equation}
	with the code rate being defined as $R^{T_n} = \frac{1}{n}\log_2 T_n$.
	Hence, from \eqref{Code rate ZP sequences}, one can observe $R^{T_n}\to\log_2\left(\frac{\sqrt{5}+1}{2}\right)\approx 0.6942$ as $n\to\infty$.
	
	Following are some remarks on the asymptotic code rate of the proposed ZP codes.
	\begin{remark}\label{bound1}
		Let us consider any positive integer $q$ and $d_i$ $(\geq 2)$ for $i = 1,2,\ldots,q$, $s.t.,$ $d_i$ is bounded by some positive constant integer $\mathcal{D}$, $i.e.,$ $d_i\leq\mathcal{D}$. Now, for any positive integers $n$ and $\mathcal{S}$ with $n = 2+\sum_{i=1}^{q} d_i$ and $\mathcal{S} = 2^{(q+2)}-1$, the asymptotic code rate of an $(n,\mathcal{S})$ ZP code is
		
		\begin{align}
			\lim\limits_{n\to \infty}\left(\frac{1}{n}\log_2\mathcal{S}\right) = & \lim\limits_{q\to \infty} \frac{\log_2(2^{q+2}-1)}{2+\sum\limits_{i=1}^qd_i} \notag \\ = & \lim\limits_{q\to \infty} \frac{\log_2(2^{q+2}-1)}{2+q\Tilde{d}}\notag \\
			\leq & \lim\limits_{q\to \infty} \frac{q+2}{2+q\Tilde{d}} \notag \\ 
			= & \frac{1}{\Tilde{d}}.
		\end{align}		
		where $\Tilde{d} = \frac{1}{q}\sum\limits_{i=1}^q d_i$.
		Note that the condition $2\leq d_i\leq \mathcal{D}$ leads to the fact that $2\leq\lim\limits_{q\to\infty}\Tilde{d}\leq\mathcal{D}$.
	\end{remark}
	\begin{remark}
		From Remark \ref{bound1}, for $d_i = d,\ (i = 1,2,\ldots,q)$, the asymptotic code rate of the $(qd+2,2^{(q+2)}-1)$ ZP code is $\frac{1}{d}$.
	\end{remark}
	\begin{figure*}
		\centering
			\subfloat[ZPZS linear codes.]{\includegraphics[width=0.5\linewidth]{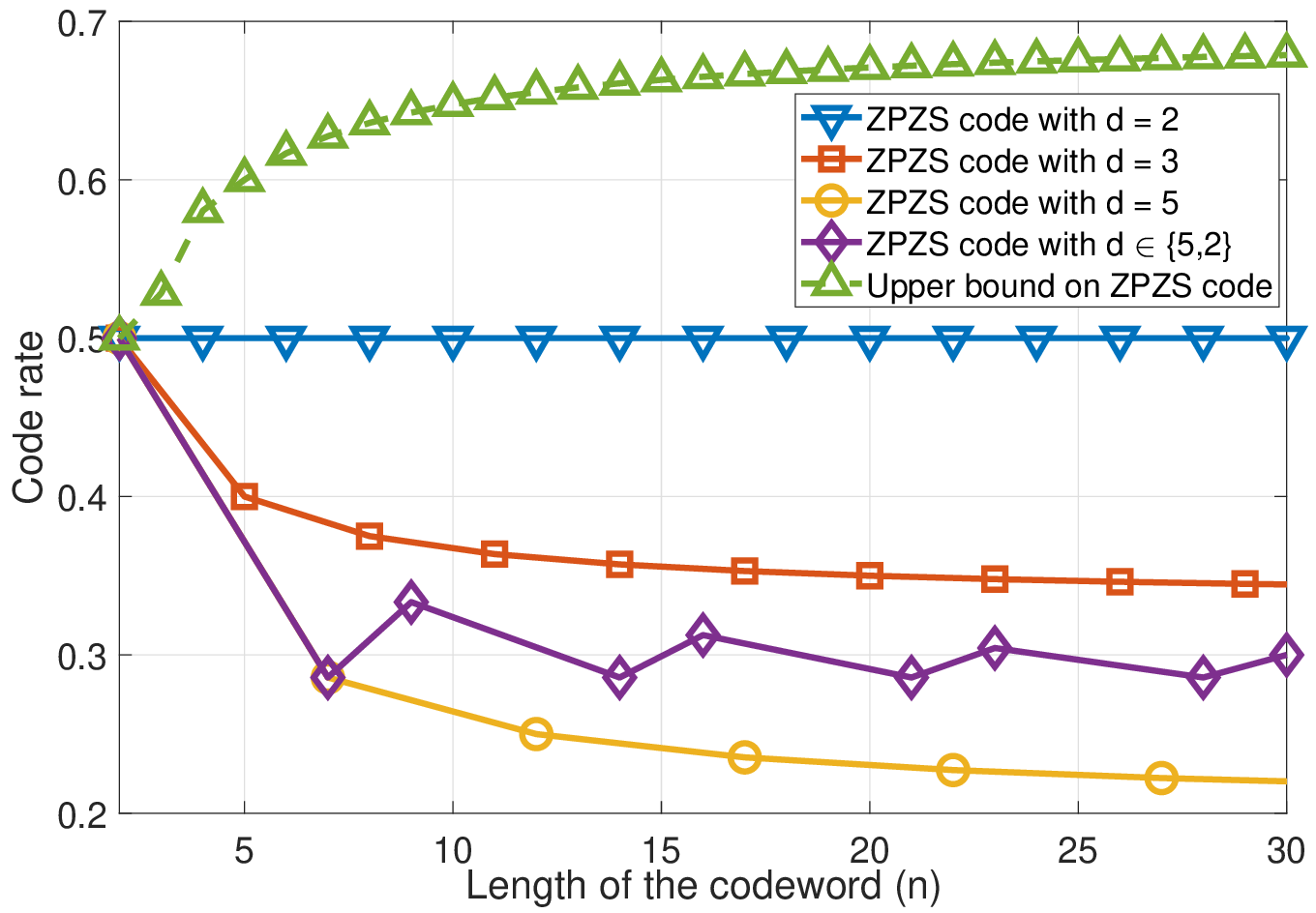}
			\label{code1}}
			\subfloat[ZP codes.]{\includegraphics[width=0.5\linewidth]{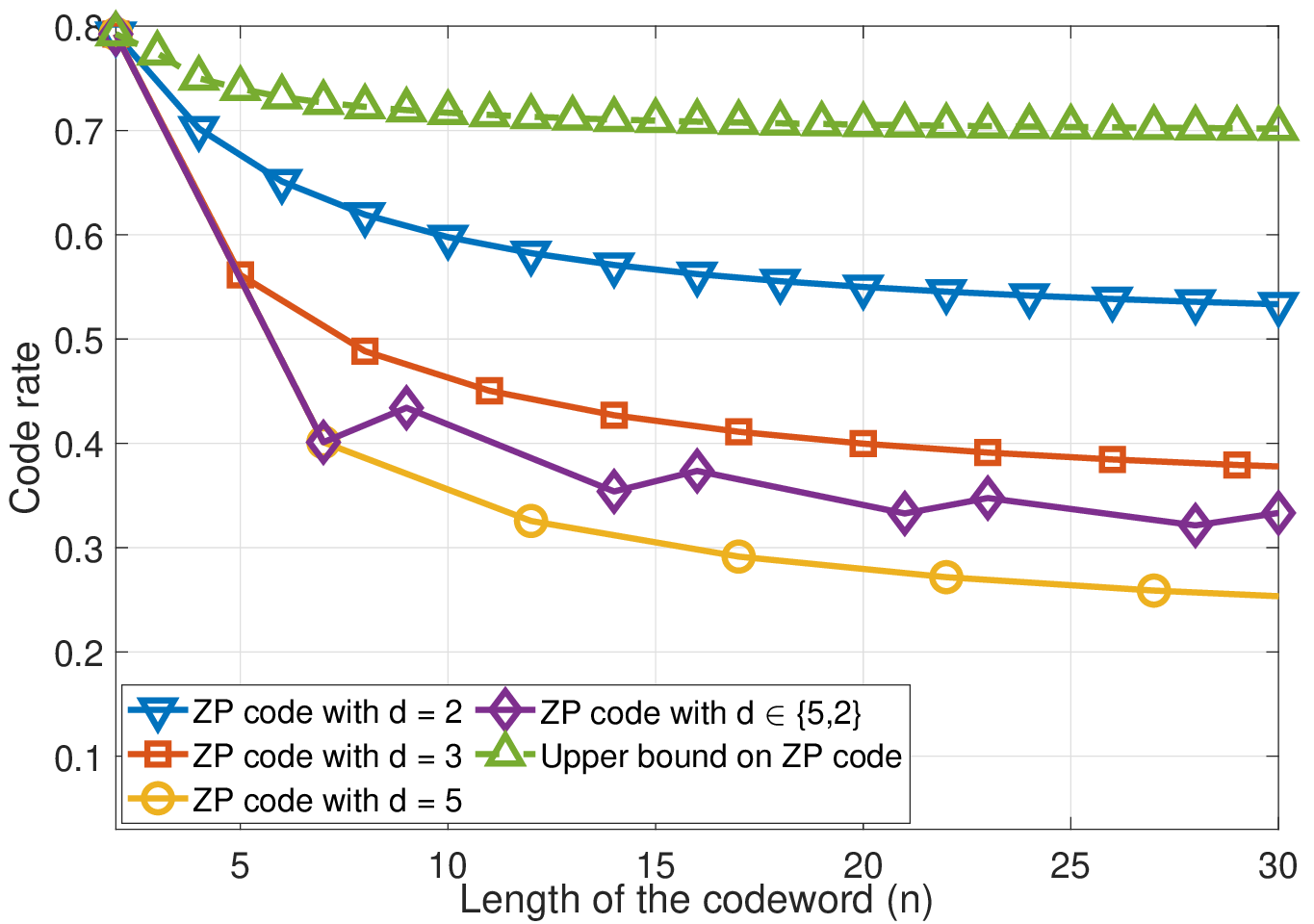}  
			\label{code2}}
		\caption{Comparison of code rate for the ZPZS and ZP codes.}
		\label{coderate}
	\end{figure*}
	In Fig. \ref{coderate} (a), we observe the code rate of ZPZS linear codes in comparison with the upper bounds on code rate with ZPZS constraints.
	On the other hand,  Fig. \ref{coderate} (b) illustrates the comparison of code rates between the proposed family of ZP codes and the upper bound on the code rate with ZP constraints. We observe that for the same code length, the ZP code achieves a better code rate compared to the ZPZS linear code, and with an increasing value of $d$, the code rate decreases for both the ZPZS and ZP codes. Note that the periodic nature (with a period of two) in the code rate for the ZP code with $d\in \{5,2\}$ arises due to two different number of zero padding between two consecutive bit-1s.
	
	Similarly, we can obtain the asymptotic code rate for the LOZP code as given in Remark \ref{asymptotic_rate_lozp}.
	\begin{remark}\label{asymptotic_rate_lozp}
		From Lemma \ref{constr4} and Remark \ref{bound1}, for any positive integers $q$, $\tau ~(\geq 2)$, $n$ and $k$ with $n = \tau+\sum_{i=1}^{q} d_i$, $k = \tau+q$ and $d_i\leq\mathcal{D}$ for $i=1,2,\ldots,q$, the asymptotic code rate of an $[n,k]$ LOZP code is
		\begin{align}
			&\lim\limits_{n\to \infty}\frac{k}{n} =  
			\begin{cases}
				\lim\limits_{\frac{\tau}{q}\to\infty}\frac{\frac{\tau}{q}+1}{\frac{\tau}{q}+\Tilde{d}} = 1, & \mbox{for } \tau\gg q \\
				\lim\limits_{\frac{q}{\tau}\to\infty}\frac{\frac{\tau}{q}+1}{\frac{\tau}{q}+\Tilde{d}} = \frac{1}{\Tilde{d}}, & \mbox{for } \tau\ll q,
			\end{cases} 
		\end{align}	
		where $\Tilde{d} = \frac{1}{q}\sum\limits_{i=1}^q d_i$.
	\end{remark}	
	
	\section{Performance Evaluation}\label{Sec 6}
	In this section, we compare the expected ISI and average BER performance for different linear codes: (i) single error correcting $[7, 4]$ Hamming code \cite{7273857}, (ii) two error correcting $[8, 4]$ Reed Solomon code \cite{7859349} and non-linear codes: (i) $(7, 20)$ ISI-mtg code, (ii) $(5, 7)$ ISI-mtg code \cite{8972472} and the (iii) un-coded case with our proposed families of channel codes: (i) $(6, 15)$ ZP code $\mathcal{C}_{2,2}\cup T(\mathcal{C}_{2,2})$, (ii) $(5, 7)$ ZP code $\mathcal{C}_{3,1}\cup T(\mathcal{C}_{3,1})$, (iii) $(9, 15)$ ZP code $\mathcal{C}_{2}\cup T(\mathcal{C}_{2})$ with $d \in \{5,2\}$, (iv) $(7, 7)$ ZP code $\mathcal{C}_{5,1}\cup T(\mathcal{C}_{5,1})$ and (v) $[8, 5]$ LOZP code $\mathcal{C}^2_{2,3}$, (vi) $[12, 7]$ LOZP code $\mathcal{C}^2_{2,5}$.
	Through the encoding and decoding mechanisms over $10^7$ transmitted message blocks with  Monte Carlo simulations, the average BER of the proposed codes is determined with the simulation parameters outlined in Table II. Note that each transmitted message block explicitly refers to one complete codeword generated by the respective channel codes.
		Also, the average BER of the codes is obtained with a fixed threshold detector with an optimized threshold.
	
	The subsequent sections are divided into four parts: (a) effect of ISI with different codebook (section \ref{subsubsec_channel_without_ref_isi}),  (b) average BER performance with transmitted number of molecules (section \ref{subsubsec_channel_without_ref_ber_mol}), (c) average BER performance with receiver noise (section \ref{subsubsec_channel_without_ref_ber_noise}) and (d) effect of channel refresh (section \ref{subsubsec_channel_with_refresh}.
	These sections are further divided into two parts: analyzing the performance of the coded system in the MCvD channel across two different data rate regions: (i) one corresponding to $t_s = 0.2$s, and (ii) another corresponding to $t_s = 0.3$s. 
	The choice of this symbol duration $(t_s)$ largely depends on the distance between the Tx and Rx to ensure most of the molecules reach at the Rx within this interval \cite{7331300}. 
	However, the authors also show that the data rate is inversely proportionate to the symbol duration, and hence the value of $t_s$ should be chosen accordingly.
	For the parameters in Table \ref{tab1}, the channel coefficients $p^{\mathrm{ch}}_1$, $p^{\mathrm{ch}}_2$ and $p^{\mathrm{ch}}_3$ are $0.1875$, $0.0777$ and $0.0390$ at $t_s = 0.2$s and $0.2344$, $0.0698$ and $0.0336$ at $t_s = 0.3$s, respectively. 
	Therefore, a small change of $0.1$s in symbol duration increases the probability of reception of the molecule at first interval while reducing it in subsequent intervals.
	Consequently, the average ISI and BER performance with $t_s = 0.2$s and $0.3$s will also vary due to this difference in channel coefficients, and is therefore discussed in the subsequent sections.

	\begin{table}[t]
		\centering
		\caption{Simulation parameters \cite{7331300}}
		\begin{tabular}{|l|c|c|}
			\hline				
			{{Parameter}}& {{Symbol}}& {{Value}} \\
			\hline\hline
			Radius of the Rx & $r_{0}$& $5 \:\mu$m \\
			Distance between PT and Rx & $d_{\mathrm{tr}}$& $10 \:\mu$m \\
			Diffusion coefficient & $D^{\mathrm{ch}}$& $79.4 \mu\text{m}^2/s$ \\
			Symbol duration & $t_s$& $\{0.2,0.3\}$ s   \\
			Number of transmitted molecules for bit-1 & $M$& $200-500$\\
			Channel memory & $L$ & $11-40$\\
			Noise variance & $\sigma_n^2$ & $0-40$ \\        
			\hline
		\end{tabular}
		\label{tab1}
	\end{table}
	
	\begin{figure}[t]
		\centering
		\includegraphics[width=1\linewidth,keepaspectratio]{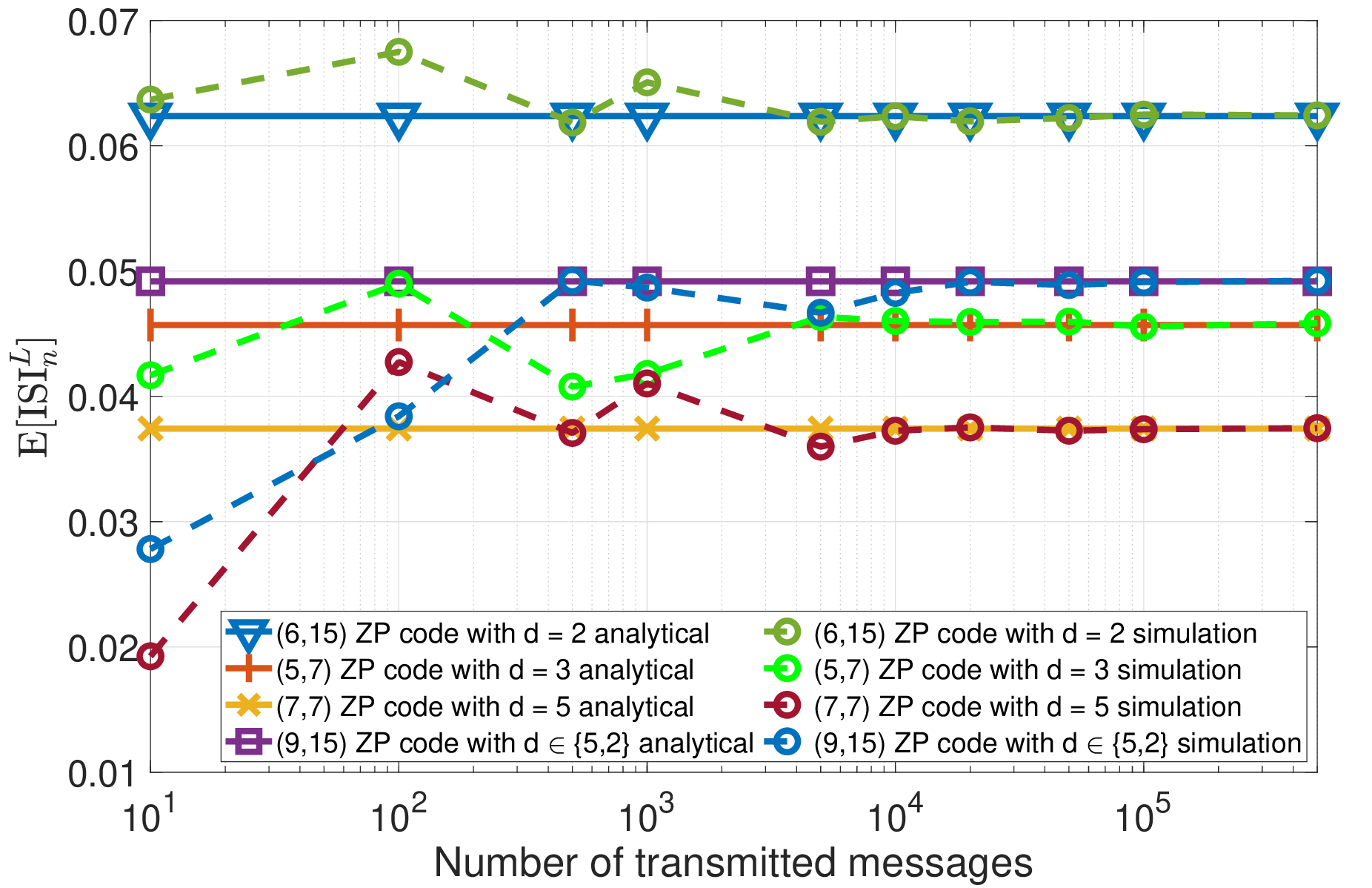} 
		\caption{Comparison of analytical and simulated last bit expected ISI with $t_s = 0.2$s, $L=20$ (without channel refresh).}\label{isi_simulation}
	\end{figure}
	
	\begin{table*}[ht]
		\centering
		\caption{ISI, average bit-1 density and code rate comparison for different channel codes at $t_s = 0.3$s and $L = 11.$}
			\begin{tabular}{|l|c|c|c|c|c|l|}
				\hline
				Code ($\mathcal{C}$)             & Parameter    &   Average bit-1     & $\mathrm{E}[\mathrm{ISI}^L_n]$ & ISI$_{\mathrm{avg}} (\mathcal{C})$  & Code rate         & Remark    \\
				& & density ($\Delta(\mathcal{C})$) & & & & \\ \hline\hline  
				$\mathcal{C}_{5,1}$         & $[7,2]$       & 0.1429  & 0.0088 & 0.0261 & 0.2857            &  ZPZS code with $d = 5$\\ \hline
				$\mathcal{C}_{2}$           & $[9,3]$       &  0.1667  & 0.0244 & 0.0305 & 0.3333            &  ZPZS code with $d \in \{5,2\}$\\ \hline
				$\mathcal{C}_{3,1}$         & $[5,2]$        & 0.2000  & 0.0206  & 0.0366 & 0.4000            &  ZPZS code with $d = 3$\\ \hline
				$\mathcal{C}_{3,2}$         & $[8,3]$        & 0.1875 & 0.0192  & 0.0343 & 0.3750            &  ZPZS code with $d = 3$\\ \hline
				$\mathcal{C}_{2,2}$         & $[6,3]$        &  0.2500  & 0.0331  & 0.0457 & 0.5000            &  ZPZS code with $d = 2$\\ \hline
				$\mathcal{C}_{2,3}$         & $[8,4]$        &  0.2500 & 0.0331   & 0.0457 & 0.5000            &  ZPZS code with $d = 2$\\ \hline
				$\mathcal{C}_{5,1}\cup T(\mathcal{C}_{5,1})$       & $(7,7)$      &  0.1633   & 0.0290  & 0.0298 & 0.4011            &  ZP code with $d = 5$\\	\hline		
				$\mathcal{C}_{2}\cup T(\mathcal{C}_{2})$ & $(9,15)$     & 0.1778    & 0.0398  & 0.0325 & 0.4341            &  ZP code with $d\in\{5,2\}$\\ \hline
				$\mathcal{C}_{3,2}\cup T(\mathcal{C}_{3,2})$       & $(8,15)$     &  0.2000   & 0.0358  & 0.0366 & 0.4884            &  ZP code with $d = 3$\\	\hline		
				$\mathcal{C}_{3,1}\cup T(\mathcal{C}_{3,1})$       & $(5,7)$      & 0.2286   & 0.0407   & 0.0417 & 0.5615            &  ZP code with $d = 3$\\	\hline
				$\mathcal{C}_{2,3}\cup T(\mathcal{C}_{2,3})$       & $(8,31)$     &  0.2581   & 0.0472   & 0.0472 & 0.6193            &  ZP code with $d = 2$\\	\hline	
				$\mathcal{C}_{2,2}\cup T(\mathcal{C}_{2,2})$       & $(6,15)$     &  0.2667   & 0.0487   & 0.0487 & 0.6511            &  ZP code with $d = 2$\\	\hline		
				$\mathcal{C}^2_{2,3}$       & $[8,5]$     &   0.3125  & 0.0365   & 0.0571 & 0.6250            &  LOZP code with $d = 2,$ $ \tau = 2$\\	\hline	
				$\mathcal{C}^2_{2,5}$       & $[12,7]$    &   0.2917   & 0.0349   & 0.0533 & 0.5833            & LOZP code with $d = 2,$ $ \tau = 2$\\	\hline
				OMP code       & $[8, 5]$    &   0.3125   & 0.0377   & 0.0571 & 0.6250            & Based on OMP sequence distribution ($d = 2,$ $ \tau = 2$)\\	\hline
				OEP code       & $[8, 5]$    &   0.3125   & 0.0427   & 0.0571 & 0.6250            & Based on OEP sequence distribution ($d = 2,$ $ \tau = 2$)\\	\hline	
				$\textbf{CW}_7$ \cite{8972472}   & $(7,20)$   &  0.2714   & 0.0494   & 0.0496 & 0.6174            & ISI-mtg code \\ \hline
				$\textbf{CW}_5$ \cite{8972472}   & $(5,7)$    & 0.2857  & 0.0526  & 0.0522 & 0.5615            & ISI-mtg code \\ \hline
				Hamming (7,4) \cite{7273857}   & $[7,4]$      &  0.5000    & 0.0914  & 0.0914 & 0.5714           & Hamming code \\ \hline
				Un-coded                       & $(7,128)$    &  0.5000   & 0.0914  & 0.0914 & 1.0000            & Un-coded case \\ \hline
			\end{tabular}
			\label{table_isi_coderate_all_codes}
		\end{table*}
		
		\subsection{Effect of ISI with Different Codebook}\label{subsubsec_channel_without_ref_isi}
		In Fig. \ref{isi_simulation}, we compare the analytical expression of the expected ISI of the considered ZP codes with the Monte Carlo simulation. Number of transmitted messages varies from 10 to $10^5$, and Fig. \ref{isi_simulation} depicts that for a large number of transmitted message blocks, the analytical expression aligns perfectly with the simulation outcomes, thus validating the analytical closed-form expression for average ISI from Theorem \ref{ISI_density_theorem} and Remark \ref{isi_zp_codes}.
		Fig. \ref{fig:fig3} illustrates the impact of channel memory on the last bit ISI for the ZP code ($\mathcal{C}\cup T(\mathcal{C})$), Hamming code, ISI-mtg code and the un-coded case for symbol duration $t_s = $ 0.2s and $t_s = $ 0.3s, which are discussed below.
		\begin{figure}[t]
			\centering
			\includegraphics[width=1\linewidth]{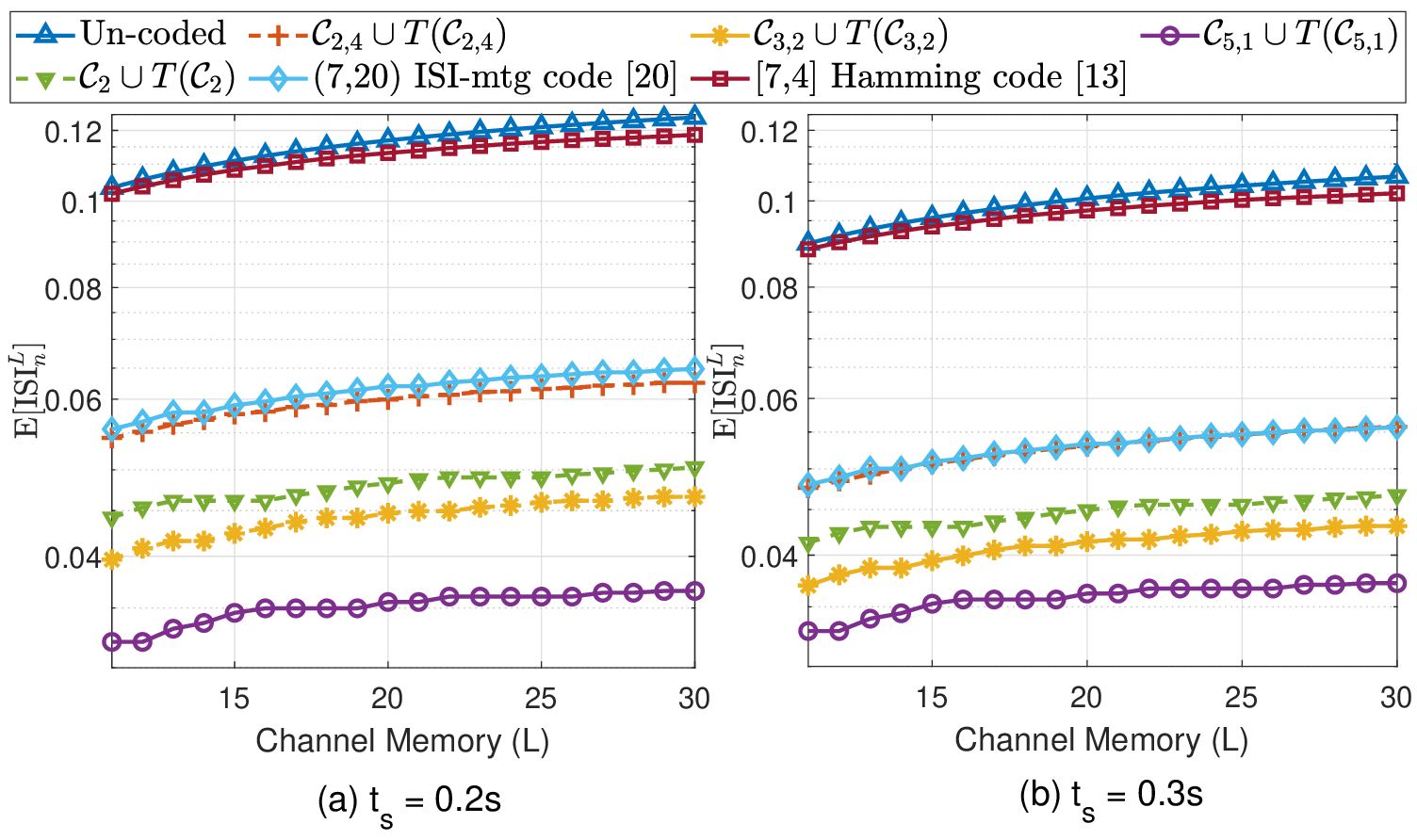}  
			\caption{Comparison of last bit expected ISI with channel memory: (a) $t_s = 0.2s$ and (b) $t_s = 0.3s$ (without channel refresh).}
			\label{fig:fig3}
		\end{figure}
		
		\begin{figure}
			\centering
			\includegraphics[width=1\linewidth]{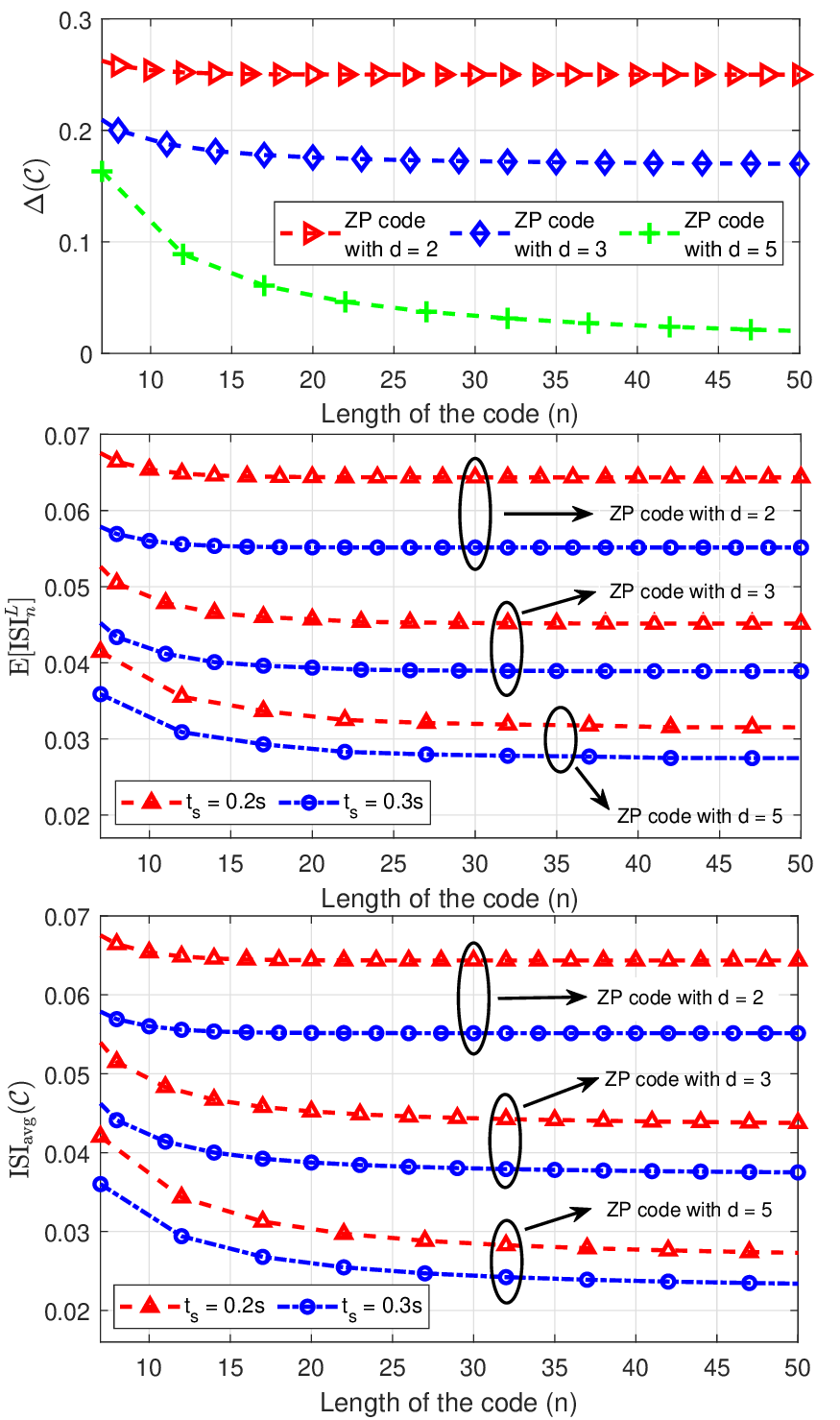}  
			\caption{Comparison of average bit-1 density, last bit expected ISI and average code ISI with length of the code for $t_s \in \{0.2\text{s}, 0.3\text{s}\}$ and $L = 11$ (without channel refresh).}
			\label{fig_isi_avg_density_vs_length}
		\end{figure}
		\subsubsection{Data rate regime with $t_s = 0.2$s} To understand the ISI performance from Fig.  \ref{fig:fig3}, we first need to obtain the average bit-1 density of the codes.
		For instance, the average bit-1 densities for the mentioned $(6, 15)$,  $(5, 7)$, $(9,15)$, $(7, 7)$ ZP codes and $(7, 20)$ ISI-mtg code are $0.2667$, $0.2286$, 
		$0.1788$, $0.1633$ and $0.2714$, respectively, which can be obtained from Remark \ref{isi_zp_codes}.
		Given that the expected ISI for any code correlates with the average density of bit-1s within the codebook from Theorem \ref{ISI_density_theorem}, the proposed family of ZP codes, due to a lower average bit-1 density, outperforms the ISI-mtg codes, which is one of the best-performing ISI-reducing channel codes reported in the literature in an MCvD system. 
		Observations from Fig. \ref{fig:fig3} (a) also confirm that the $(7, 7)$ ZP code $\mathcal{C}_{5,1}\cup T(\mathcal{C}_{5,1})$ performs better than the other considered ISI-reducing codes due to their sparse bit-1 distribution in the generator matrix and lower average bit-1 density.
		Meanwhile, the $[7, 4]$ Hamming code and the un-coded scenarios result in a similar ISI performance due to the average bit-1 density of 0.5 across all symbol positions.
		
		\subsubsection{Data rate regime with $t_s = 0.3$s} In Fig.  \ref{fig:fig3}(b), we simulate the last bit ISI at $t_s = 0.3$s. As the channel experiences less ISI at $t_s = 0.3$s than $t_s = 0.2$s, we observe an improved ISI performance with the mentioned codes.
		Table \ref{table_isi_coderate_all_codes} provides a comprehensive comparison of the average density of bit-1 in the code (from \eqref{average_bit1_density}), expected ISI on the last bit of the code (from \eqref{eq_isi_with_density}), the overall average ISI of the codebook (from \eqref{equation_avg_isi_code}), and the code rates, including the ZP and linear ZPZS codes against existing channel codes at $L = 11$ and $t_s = 0.3$s. 
		The $(n, \mathcal{S})$ ZP code exhibits a higher average number of bit-1s before the last bit compared to the $[n, k]$ ZPZS linear code, leading to a greater ISI effects on the last bit. 
		Also observe that the $(6, 15)$ ZP code (code rate of 0.6511) performs better than the $(7, 20)$ ISI-mtg code (code rate 0.6174) in terms of both the last bit ISI and the average ISI across the codebook due to a lesser average bit-1 in the code, and thereby resulting in a performance gain with a higher code rate.
		The $(5, 7)$ ZP code $\mathcal{C}_{3,1}\cup T(\mathcal{C}_{3,1})$ ensures at least two bit-0s between consecutive bit-1s, compared to the $(5, 7)$ ISI-mtg code, which requires only one bit-0 between the consecutive bit-1s. This results in a lower average bit-1 density for the $(5, 7)$ ZP code (0.2286) than the $(5, 7)$ ISI-mtg code (0.2857), which also leads to a reduced interference on bit-0 for the $(5, 7)$ ZP code. Consequently, the $(5, 7)$ ZP code achieves lower last bit ISI and average ISI than the $(5, 7)$ ISI-mtg code at the same code rate of 0.5615.
		
		Fig. \ref{fig_isi_avg_density_vs_length} compares the average bit-1 density, last bit expected ISI and the average ISI of the code in two different data rate regimes with the length of the code $(n)$ for channel memory $L = 11$. 
		The results show that as the length of the code increases, the average density of bit-1 initially decreases but eventually becomes constant after a certain length.
		Consequently, the average ISI saturates beyond a certain code length for given parameters and symbol duration.
		To improve the ISI performance, one can either choose a longer symbol duration or increase the zero padding between consecutive bit-1s. Also, for practical applications in bio-nanomachines, excessively large block lengths are undesirable. 
		Therefore, as demonstrated in Fig. \ref{fig_isi_avg_density_vs_length}, selecting shorter code parameters for the ZP code ensures both implementability and achieving satisfactory ISI performance in an MCvD system.
		
		\subsection{BER Performance with Varying Transmitted Number of Molecules}	\label{subsubsec_channel_without_ref_ber_mol}
		In Fig. \ref{fig_mol}, we demonstrate the average BER performance of our proposed family of ZP codes alongside existing channel codes, analyzing the impact of the transmitted number of molecules under two data rate regimes.
		
		\subsubsection{Data rate regime with $t_s = 0.2$s} Fig. \ref{fig_mol}(a) illustrates that in a higher data rate MCvD channel, where the impact of ISI is more pronounced, the proposed family of ZP codes performs better than the existing channel codes mentioned in this paper, across the range of transmitted molecule numbers. 
		For codes with higher bit-1 densities, such as the uncoded and Hamming codes, persistent molecules in the channel cause significant interference on bit-0 during symbol detection, degrading the channel BER. In contrast, ISI-mitigating codes, including the proposed ZP codes, ensure at least one bit-0 between consecutive bit-1s.	As the channel coefficient $p^{\mathrm{ch}}_{2}$ is more dominant among all the coefficients of $p^{\mathrm{ch}}_{i}$ for $2\leq i \leq n$, the expected ISI largely depends on the average bit-1 density of the most recent bit. Thus, with the ZP constraint, we observe a better BER performance during shorter symbol durations.
		\begin{figure*}	
				\centering
				\subfloat[$t_s = 0.2$s and $\sigma_n^2 = 0$.]{\includegraphics[width=0.5\linewidth]{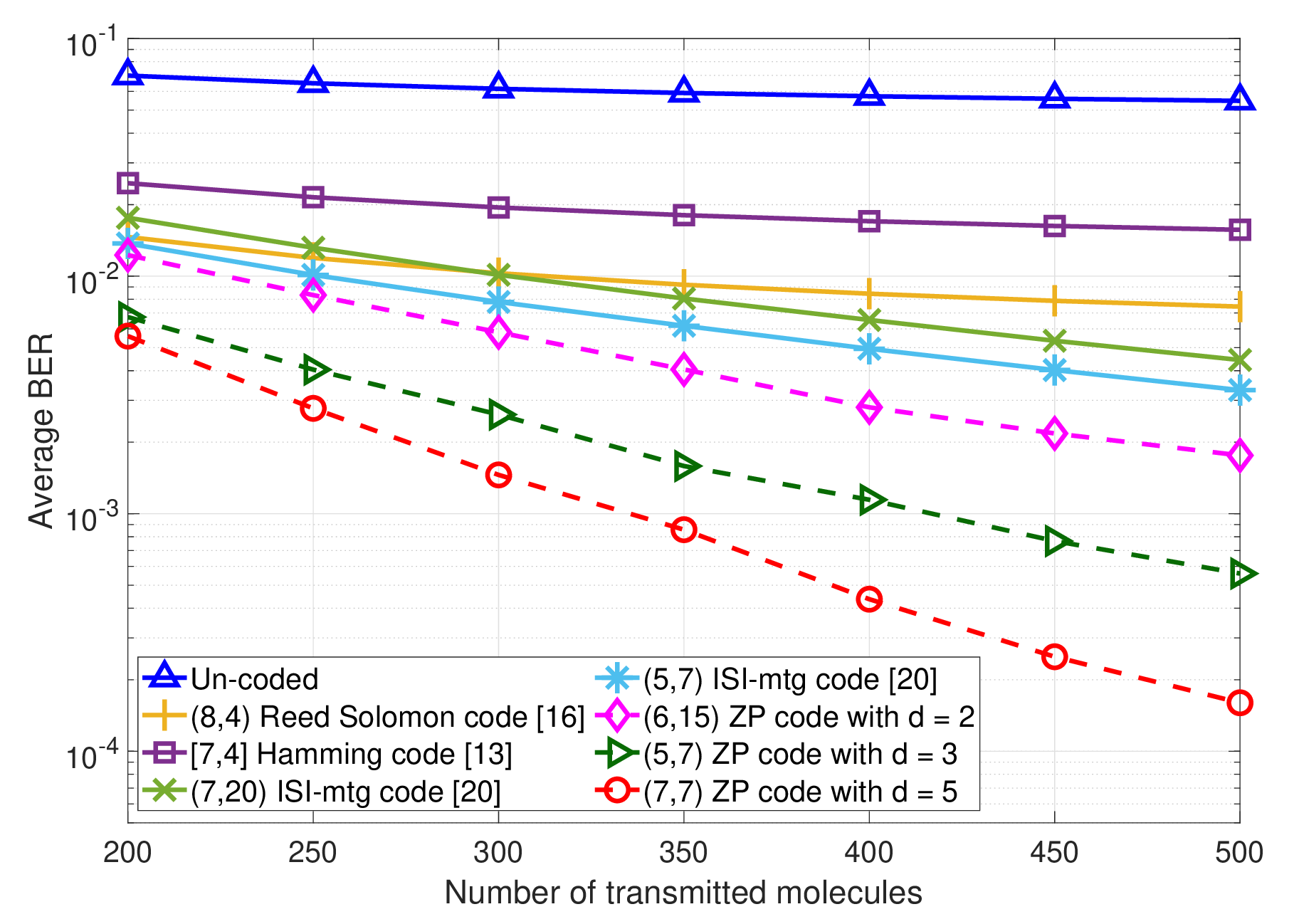} 
				\label{mol1}}
				\subfloat[$t_s = 0.3$s and $\sigma_n^2 = 0$.]{\includegraphics[width=0.5\linewidth]{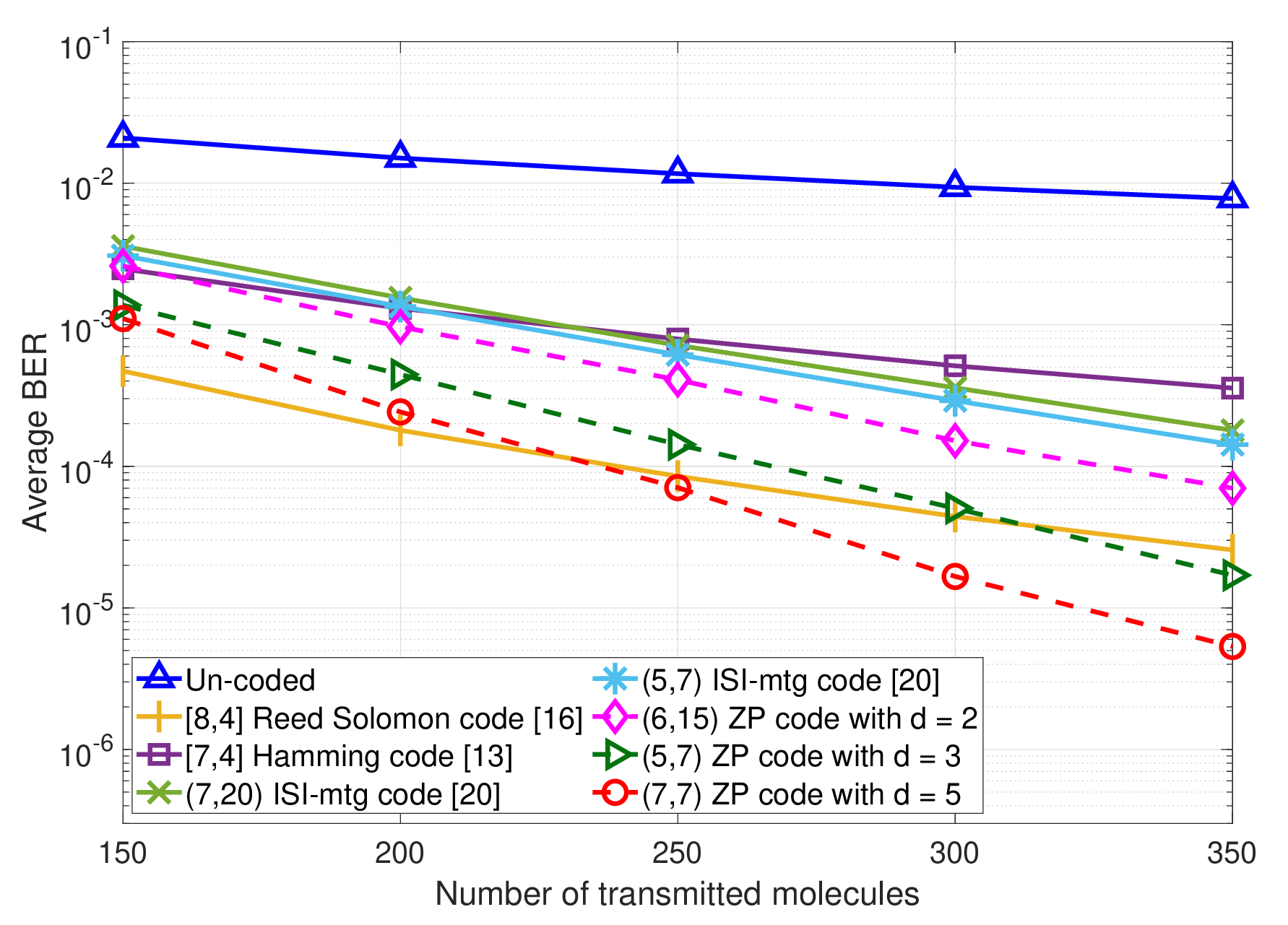}  
				\label{mol2}}
			\caption{Average BER comparison varying transmitted molecules in different data rate regimes ($t_s \in \{0.2\text{s},0.3\text{s}\}$) with $L = 40$ (without channel refresh).}
			\label{fig_mol}
		\end{figure*}
		\begin{figure*}
				\centering
				\subfloat[$t_s = 0.2$s and $M = 500$.]{\includegraphics[width=0.5\linewidth]{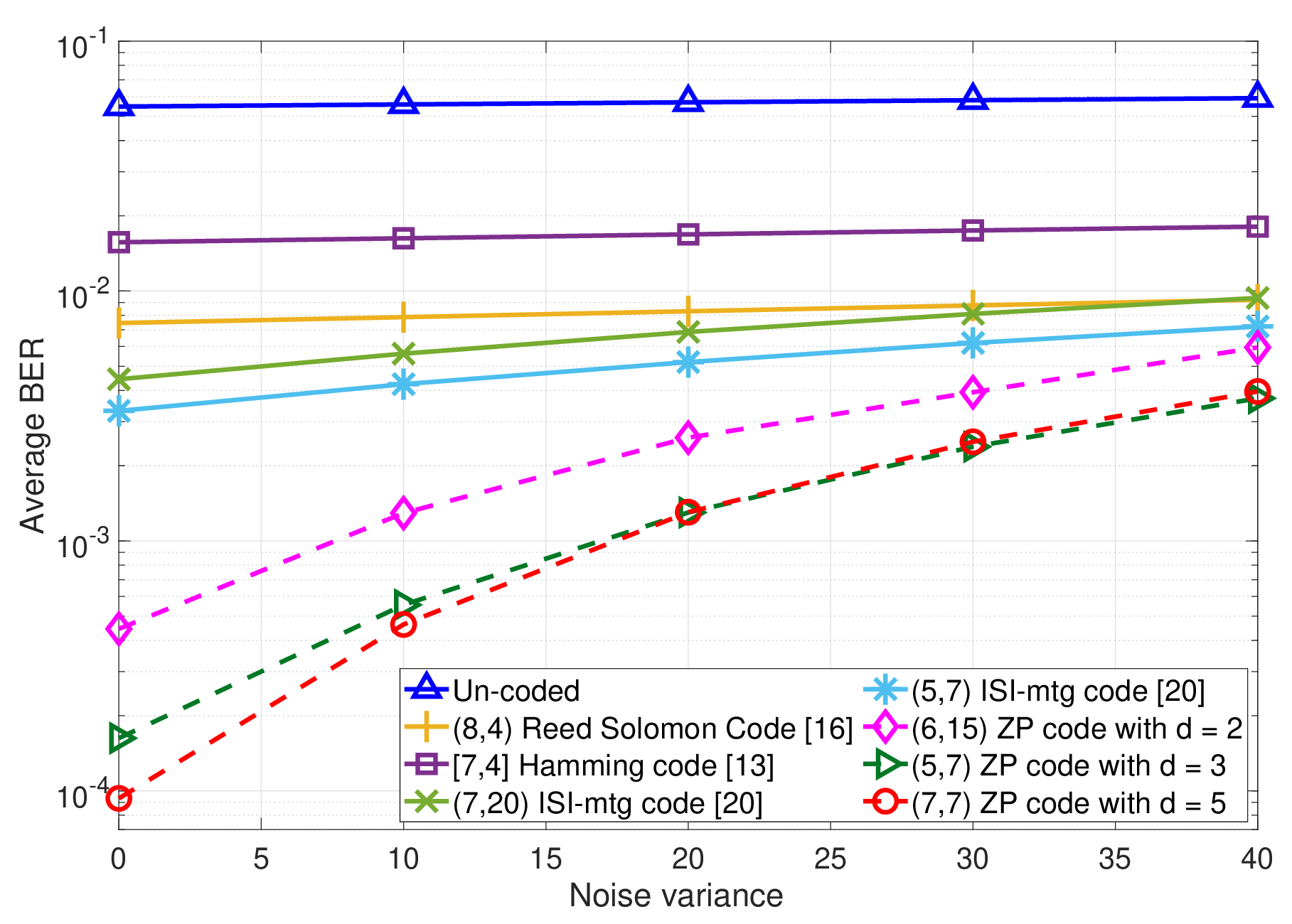}  
				\label{sigma1_zp}}
				\subfloat[$t_s = 0.3$s and $M = 350$.]{\includegraphics[width=0.5\linewidth]{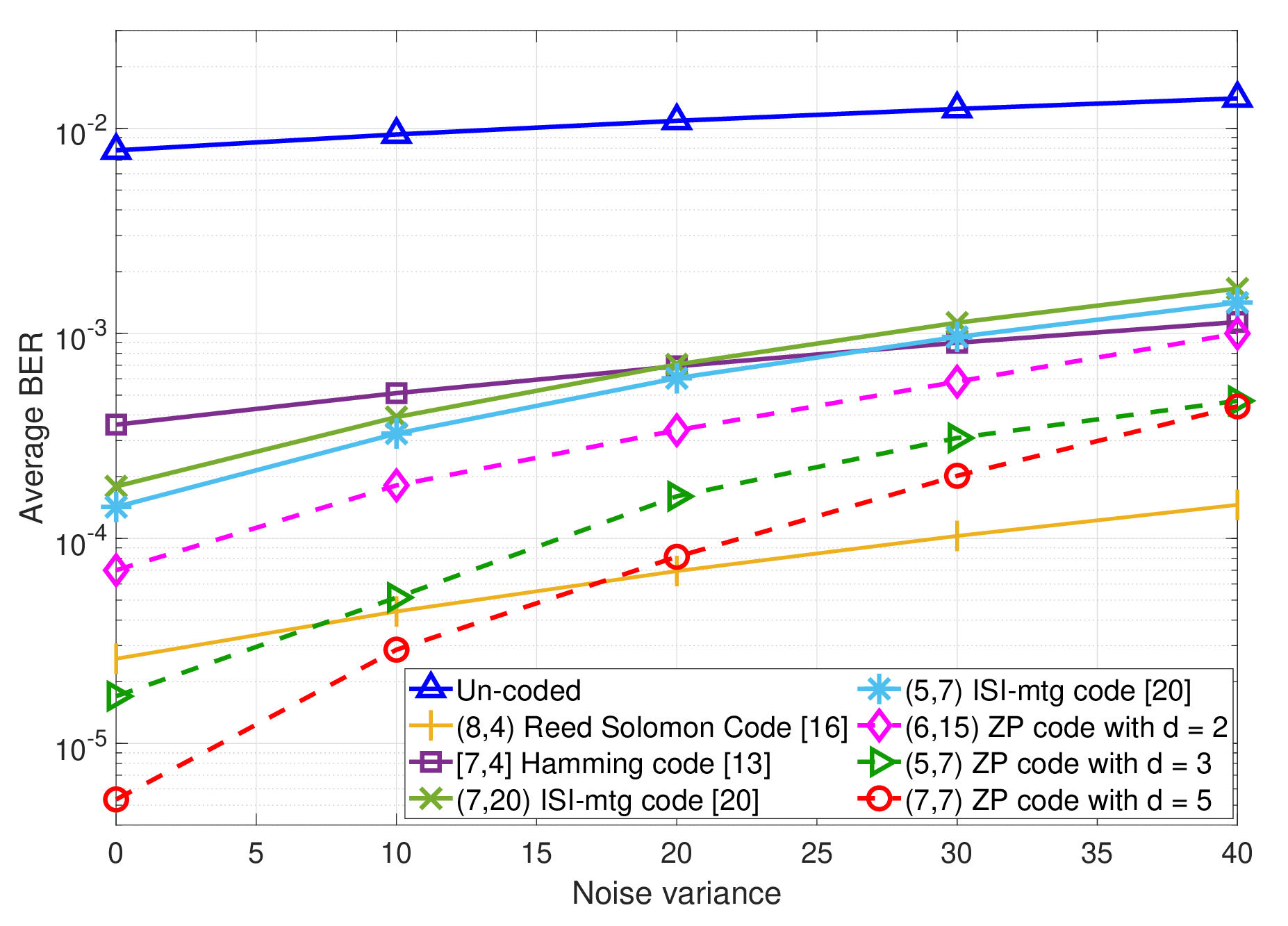}  
				\label{sigma2_zp}}
			\caption{Average BER comparison varying receiver noise in different data rate regimes ($t_s \in \{0.2\text{s},0.3\text{s}\}$) with $L = 40$ (without channel refresh).}
			\label{fig_sigma}
		\end{figure*}
	
		As the (6, 15) ZP code (code rate 0.6511) has a lower average bit-1 density than the (7, 20) ISI-mtg code (code rate 0.6174), the ZP code simultaneously shows a better BER performance over the ISI-mtg code, as depicted in Fig. \ref{fig_mol}(a). Similarly, the 
		(5, 7) ZP code demonstrates a superior BER performance compared to the (5, 7) ISI-mtg codes. Notably, the (7,7) ZP code achieves a better BER performance than the other ZP codes due to its sparse bit-1 distribution in the generator matrix. Hence, in this data rate regime with shorter symbol durations, ZP codes are more effective at minimizing ISI than conventional error correcting codes, thus emphasizing such constraints in the MCvD channel.
		
		\subsubsection{Data rate regime with $t_s = 0.3$s} In a comparatively lower data rate regime, the impact of ISI is reduced due to a longer symbol duration. 
		Therefore, from Fig. \ref{fig_mol}(b), we can observe that the $[8, 4]$ Reed Solomon code performs better than other codes due to its two error correcting property, especially when the number of transmitted molecules is relatively low.
		However, with a larger number of transmitted molecules, the residual molecules in the MCvD channel increase. The remaining molecules in the channel lead to an enhanced memory overhead to the next symbol received by the receiver and cause an error (bit flip) during the symbol detection. Thus, ISI-reducing codes (with more number of bit-0s) become crucial to minimize the memory overhead when the transmitted number of molecules is comparatively large.
		In such cases, ZP codes perform better than the Reed Solomon code due to their lower average bit-1 density and the ISI-reducing (zero padding) constraint.
		For instance, at $M = 350$ and $t_s = 0.3$s, the system achieves an average BER of $2.56\times 10^{-5}$ with the $[8, 4]$ Reed Solomon code (code rate of 0.5).
		In comparison, the $(5, 7)$ ZP code $\mathcal{C}_{3,1}\cup T(\mathcal{C}_{3,1})$, with a higher code rate of 0.5615, achieves a lower average BER $1.726\times 10^{-5}$.
		
		This observation emphasizes that in a channel with a longer symbol duration, the error correcting properties become dominant over the ISI-reducing properties. 
		However, in scenarios with an increased number of transmitted molecules, ISI-reducing ZP codes prove to be more effective than error correcting codes such as the Reed Solomon or Hamming code.
		
		\subsection{BER Performance with Varying Receiver Noise}\label{subsubsec_channel_without_ref_ber_noise}	In Fig. \ref{fig_sigma}, we illustrate the BER performance with receiver noise for the proposed ZP code and compare them with existing codes with $t_s = 0.2$s and $t_s = 0.3$s, respectively. 
		
		\subsubsection{Data rate regime with $t_s = 0.2$s} In Fig. \ref{fig_sigma}(a) with $M = 500$ and $t_s = 0.2$s, we observe that both the un-coded and the $[7,4]$ Hamming code do not improve the system performance due to the high bit-1 density and subsequently large ISI impact in the channel. 
		Also, when the channel noise is relatively small, ISI in the channel is still dominant over the receiver noise. Therefore, in such cases, we obtain a superior BER performance with the ZP codes than the existing codes due to their ISI-reducing properties (ZP constraint and low average bit-1 density).
		Whereas, when receiver noise increases and becomes more dominant than the ISI effect, the BER performance of the proposed ZP codes aligns with that of the ISI-mtg code, as neither possesses error correcting properties.  However, the (5, 7) ZP code (code rate 0.5615), due to the presence of at least two bit-0s between consecutive bit-1s, attains a better BER performance than the ISI-mtg code (code rate 0.5615) and traditional error correcting codes (code rate 0.5 in Reed Solomon and 0.5714 in Hamming code).
		
		\subsubsection{Data rate regime with $t_s = 0.3$s} Fig. \ref{fig_sigma}(b) shows that, with $M = 350$ and $t_s = 0.3$s, the impact of ISI diminishes, allowing the Reed Solomon code to perform better than other channel codes due to the two error correcting property. While the single error correcting $[7, 4]$ Hamming code also shows a similar BER performance with the $(6, 15)$ ZP code at higher noise variance ($\sigma_n^2 = 40$).
		However, at a comparatively lesser noisy region ($0\leq \sigma_n^2\leq 30$), where the ISI effects are more dominant over the channel noise, the ZP codes perform better than the existing single error correcting and the ISI-mtg codes in an MCvD channel.
		Therefore, we observe that the error correcting codes show an improved BER performance with $t_s = 0.3$s than the channel with $t_s = 0.2$s due to a lesser ISI impact, which follows a similar argument from the explanation of Fig. \ref{fig_mol}(b) in section \ref{subsubsec_channel_without_ref_ber_mol}.
		
			\subsection{BER Performance with Varying Transmission Distance}
			From \cite[Eq. (23)]{6807659}, as the transmission distance increases, the probability of one molecule reaching the Rx for a fixed time reduces.
			Therefore, for a fair comparison of the BER performance with different transmission distance $d_{\mathrm{tr}}$, we have fixed the capture probability of one molecule until time $t$ to be $0.2$ in this paper, and accordingly computed the time period $t$.
			Following are the required symbol durations to obtain the capture probability of $0.2$:
			\begin{enumerate}
				\item $d_{\mathrm{tr}} = 10\mu$m: $t = 0.223$s,
				\item $d_{\mathrm{tr}} = 10.5\mu$m: $t = 0.293$s,
				\item $d_{\mathrm{tr}} = 11\mu$m: $t = 0.381$s,
				\item $d_{\mathrm{tr}} = 11.5\mu$m: $t = 0.488$s.
			\end{enumerate}	
			Clearly, as transmission distance increases the symbol duration also increases for a fixed capture probability of one molecule.
			Hence, we can summarize the following observations combining both Fig. \ref{fig_ber_vs_ts} and Fig. \ref{fig_ber_vs_dist}:\\
			(i) \textit{Noiseless channel and shorter transmission distance:} 
			When the transmission distance is short, the symbol duration is also small. 
			Consequently, the effect on the current symbol from the previously transmitted symbols becomes very prominent and the interference in the channel becomes dominant over the noise.
			For instance, Fig. \ref{fig_ber_vs_ts} demonstrates that the considered ZP codes in this paper have a performance gain over the other codes for $0.2\text{s} \leq t_s \leq 0.25\text{s}$ in a noiseless channel.
			Also, from  Fig. \ref{fig_ber_vs_dist}, with $\sigma_n^2 = 0$ and  $10\mu\text{m} \leq d_{\mathrm{tr}} \leq 10.5\mu\text{m}$, all the ZP codes perform better compared to the other codes in the paper.\\
			(ii) \textit{Noiseless channel and larger transmission distance:}
			As the transmission distance becomes large, the symbol duration also increases to maintain the same capture probability of $0.2$.
			Therefore, the effect from the earlier symbol diminishes on the current symbol and the effect of ISI also reduces.
			In such scenarios, Both the error-correcting properties and the ISI-reducing properties of the codes can improve the BER performance.
			For instance, at $d_{\mathrm{tr}} = 11.5\mu$m, the symbol duration becomes sufficiently large to reduce the effect of ISI in the channel.
			Consequently, the two error-correcting $[8, 4]$ Reed-Solomon code shows a similar BER performance with the ISI-reducing $(5,7)$ ZP code.\\
			(iii) \textit{Noisy channel and shorter transmission distance:}
			Even in a noisy channel ($\sigma_n^2 = 30$), the effect of ISI becomes dominant over the channel noise when the symbol duration is small.
			Consequently, the ZP codes perform better compared to the Reed Solomon codes (at a comparable code rate) in the region $0.2\text{s} \leq t_s \leq 0.25\text{s}$.
			Also, for $10\mu\text{m} \leq d_{\mathrm{tr}} \leq  10.5\mu\text{m}$, the $(5, 7)$ ZP code has a performance gain over the $[8,4]$ Reed-Solomon code.\\
			(iv) \textit{Noisy channel and larger transmission distance:}
			In a noisy channel, higher the transmission distance, higher the symbol duration is, which eventually eases out the ISI effect in the channel. 
			Therefore, for larger transmission distance, the channel noise becomes dominant over the ISI. 
			For instance, with $\sigma_n^2 = 30$ and  $d_{\mathrm{tr}} = 11.5 \mu$m, the two-error correcting $[8,4]$ Reed Solomon code attains an improved system performance compared to the ZP codes.
			Also, observe that the single error-correcting $[7,4]$ Hamming code shows a similar BER performance to the $(5,7)$ ISI-mtg code in this regime.
		\begin{figure}
			\centering
			\includegraphics[width = 1\linewidth]{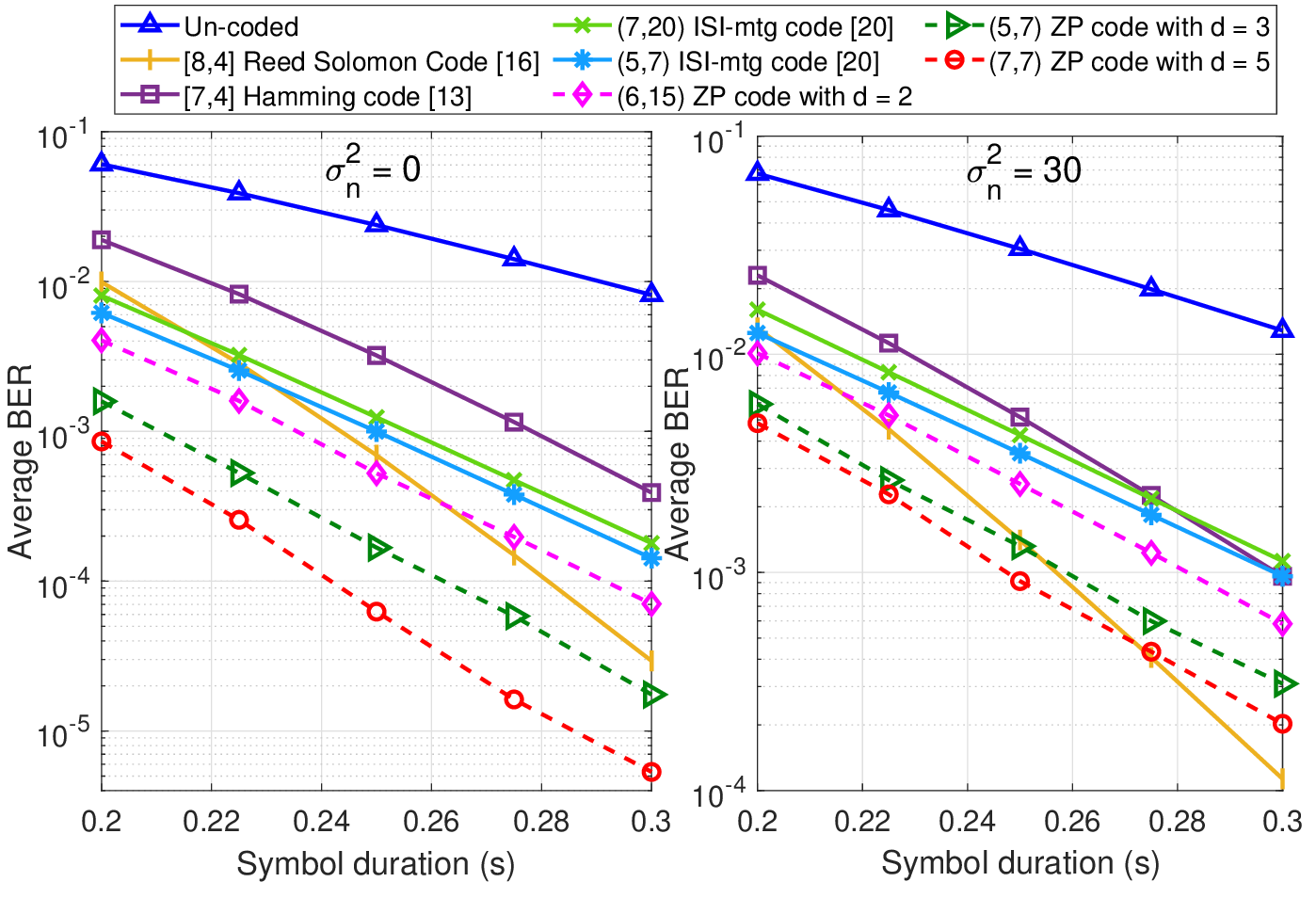}
			\caption{Average BER comparison with symbol duration $t_s$ for $L = 40$, $M = 350$ and $\sigma_n^2\in \{0, 30\}$ (without channel refresh).}
			\label{fig_ber_vs_ts}
		\end{figure}
		\begin{figure}
			\centering
			\includegraphics[width = 1\linewidth]{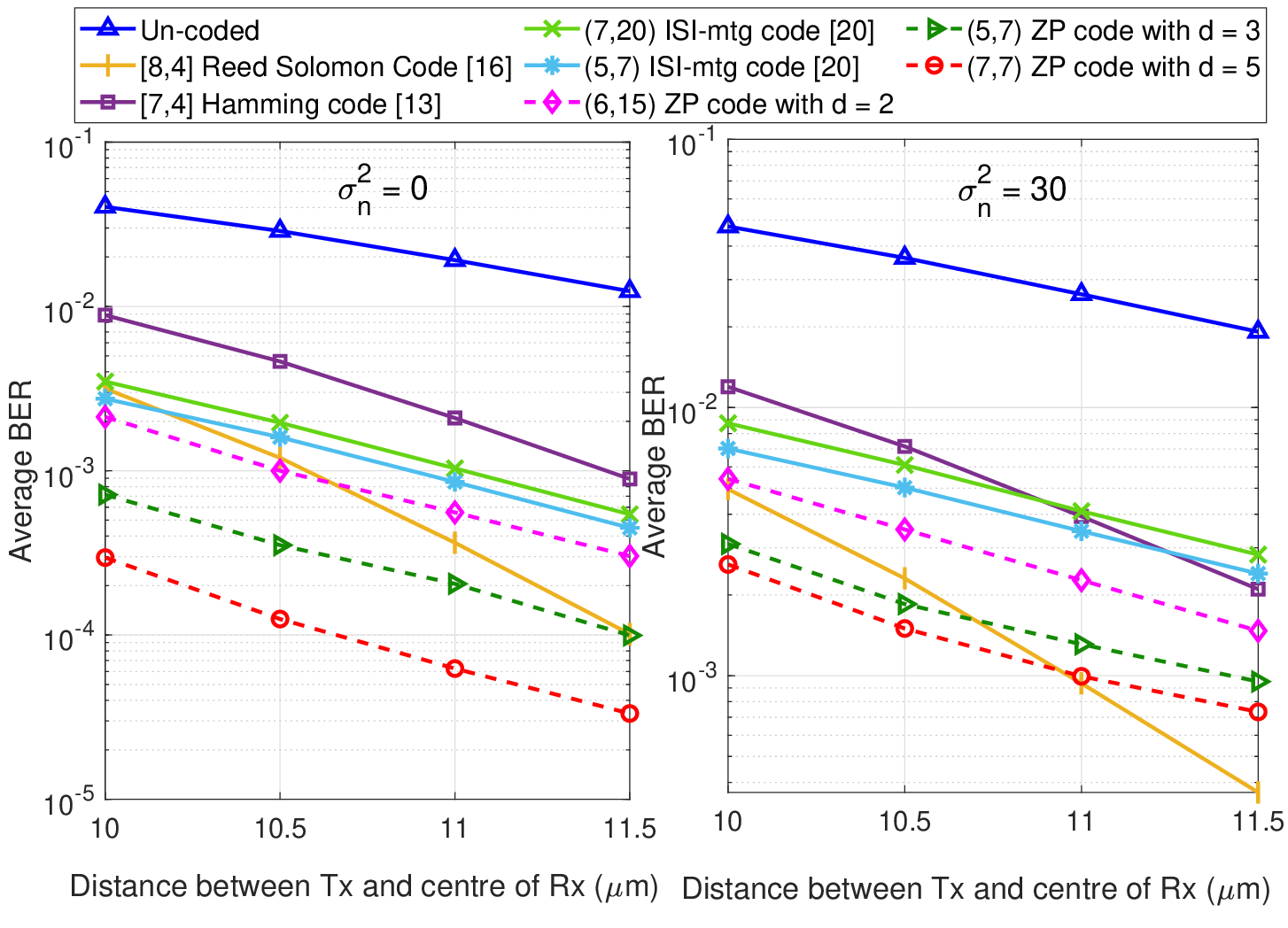}
			\caption{Average BER comparison with transmission distance $d_{\mathrm{tr}}$ for $L = 40$, $M = 350$ and $\sigma_n^2\in \{0, 30\}$ (without channel refresh).}
			\label{fig_ber_vs_dist}
		\end{figure}
	
		Therefore, combining Fig. \ref{fig_ber_vs_ts} and Fig. \ref{fig_ber_vs_dist}, we can conclude that the ZP codes are preferred over classical error-correcting codes when ISI is more dominant in the channel, whereas the error-correcting codes are preferred when channel noise in the system becomes dominant over the ISI.

		\subsection{Effect of Channel  Refresh} \label{subsubsec_channel_with_refresh}
		In this section, we compare the ISI and average BER performance, where the channel is refreshed after every successful message reception with the parameter $L = n-1$ and also compare the code performance in a channel without refresh.
		We have considered the $[8, 5]$ and $[12, 7]$ LOZP codes to compare the performance with $[7, 4]$ Hamming code, $[8, 4]$ Reed Solomon code, $(7, 20)$ ISI-mtg code, $(5, 7)$ ISI-mtg code alongside other constructed linear codes ($[8, 5]$ OMP and $[8, 5]$ OEP code) based on sequence distributions from \cite{9840783}. 
		Note that a linear ZPZS code can achieve a maximum code rate of $0.5$, as constructed in Lemma \ref{constr_1}. Therefore, the consecutive bit-1s with a linear code construction become important to improve the code rate with an accepted BER performance.
		For example, with length $n = 8$, the linear LOZP code with parameters $\tau = 2$ and $d = 2$, attains a code rate of 0.6250, whereas the ZP code with a non-linear construction can achieve a maximum (asymptotic) code rate of 0.6193.
		In such cases, placement of these consecutive bit-1s within the codeword becomes an important metric where the channel periodically gets refreshed by the enzyme emission from the receiver.
		
		In Fig. \ref{ber_with_without_refresh},  we demonstrate the average BER performance of the considered codes, analyzing the impact of channel refresh. We observe that channel refresh improves system performance across all codes by diminishing the impact of ISI.
		Note that the LOZP, OMP and OEP codes with code length $n = 8$ and dimension $k = 5$, have the identical average bit-1 density of 0.3125 (Table \ref{table_isi_coderate_all_codes}). However, among these considered codes in Fig. \ref{ber_with_without_refresh}(a) (channel without noise) and \ref{ber_with_without_refresh}(b) (channel with noise), the $[8, 5]$ LOZP code exhibits a better BER performance in a channel scenario with refresh of $L = n -1$.
		Therefore, this result validates our statement in Lemma \ref{bit_1_position_lemma}, that the consecutive bit-1s are desirable at the initial positions to reduce the ISI effect in such scenarios.
		For example, with $t_s = 0.2$s, $M = 500$ and $\sigma_n^2 = 0$, the system achieves an average BER of $1.5140\times 10^{-5}$ the $[8, 5]$ LOZP code, while with the $[8, 5]$ OMP and $[8, 5]$ OEP code, it achieves average BER of $4.3650\times 10^{-4}$ and $1.1526\times 10^{-3}$, respectively in an MCvD channel. 
		In channels without refresh, the construction of the OMP and OEP codes prevents them from having three consecutive bit-1s, with the maximum weighted codewords being	10110101 (OMP) and 10101101 (OEP). In contrast, the LOZP code's maximum weighted codeword is 11010101, which can have three consecutive bit-1s and thus contributes to higher ISI in channels without refresh. Hence, in such scenarios, OMP and OEP codes perform better than the LOZP code in both data rate regimes. However, in a channel with a refresh, the ZP code, as discussed in section \ref{subsubsec_channel_without_ref_ber_mol}, outperforms both OMP and OEP codes due to the absence of consecutive bit-1s, thus preferred over OMP and OEP codes in a channel with a refresh.
		
		The $[8, 4]$ Reed Solomon code, due to its two error correcting property, performs better than the remaining codes in a noisy channel with refresh and longer symbol duration ($t_s = 0.3$s), where noise becomes dominant over the ISI. Whereas, in the remaining cases with channel refresh, the $[12, 7]$ LOZP code achieves a better performance than the remaining mentioned codes (at a similar code rate), as it restricts more than two consecutive bit-1s at the beginning of the codeword, which is one of the ISI-reducing properties of a code.
		Hence, in addition to the average bit-1 density and ZP constraint, the location of bit-1s in a linear code also emerges as an important metric in controlling ISI. 
		
		\begin{figure*}[htbp]
				\centering
				\subfloat[Without noise ($\sigma_n^2 = 0).$]{\includegraphics[width=0.5\linewidth]{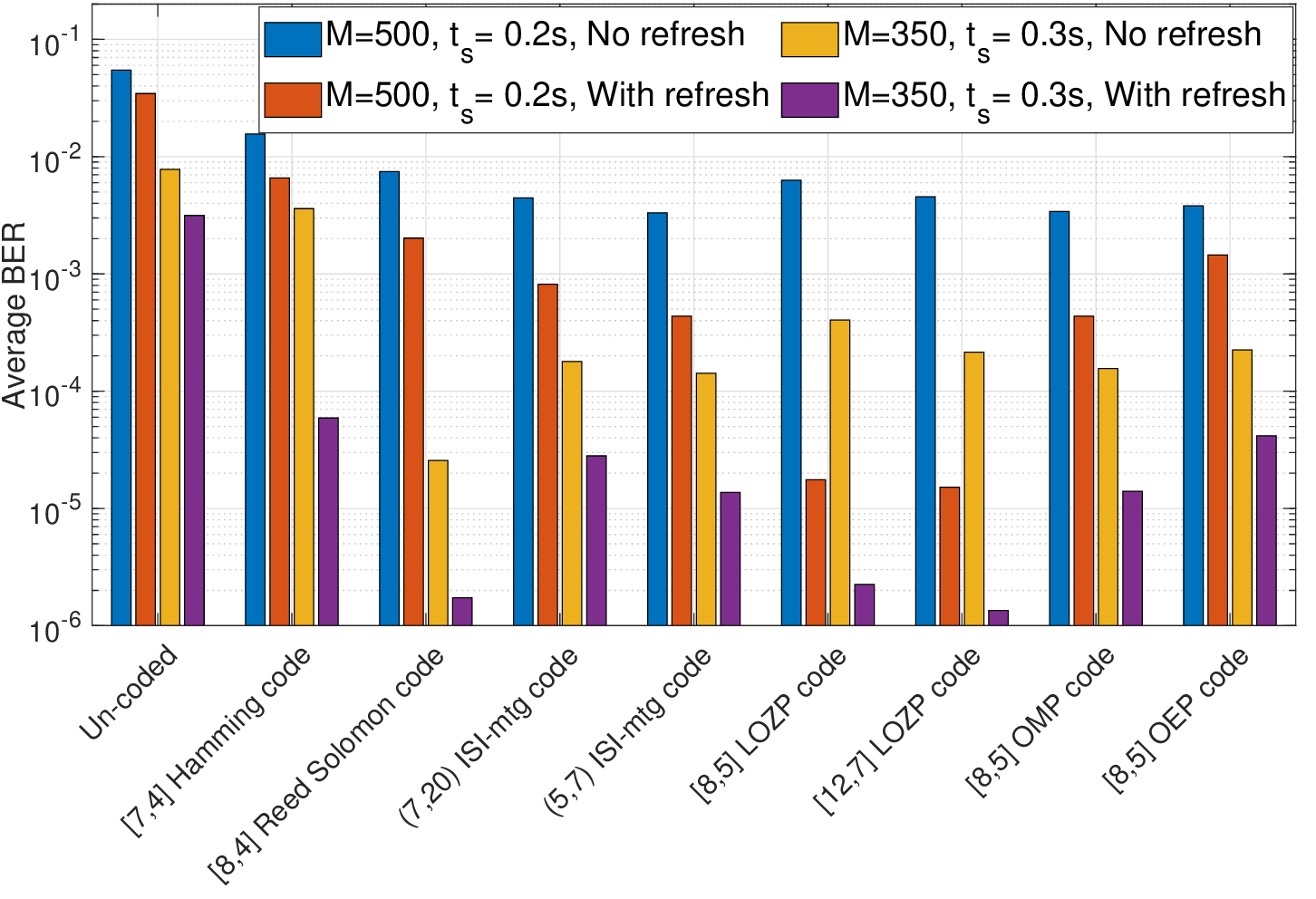}  
				\label{ber_vs_mol_with_without_refresh}}
				\subfloat[With noise ($\sigma_n^2 = 40).$]{\includegraphics[width=0.5\linewidth]{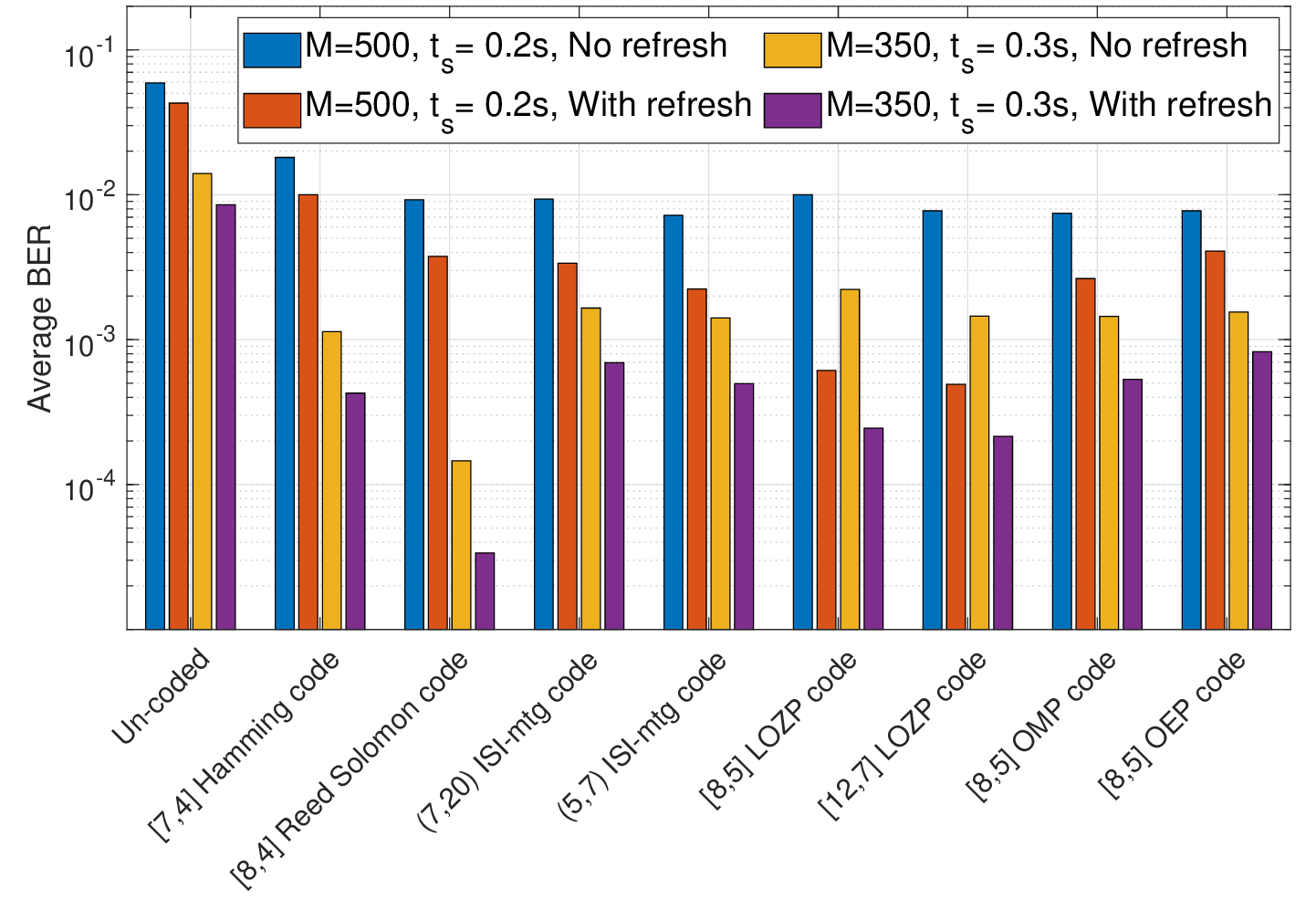}  
				\label{ber_vs_noise_with_without_refresh}}
			\caption{Average BER comparison for different linear and non-linear codes in different data rate regimes without channel refresh ($L = 40$) and with channel refresh ($L = n-1$).}
			\label{ber_with_without_refresh}
		\end{figure*}
		
		\section{Conclusion}\label{sec 7}
		In this paper, we first construct a family of binary codes based on the zero padding constraint using a linear approach and then propose a linear Leading One Zero Pad code by relaxing the zero padding constraint to achieve an improved code rate.
		We derive a closed-form expression on the expected ISI for the ZP and linear LOZP codes and show a one-to-one relation between the average bit-1 density and the expected ISI of the code.
		From the simulation results, it is evident that increasing the number of zeros between two consecutive bit-1s lowers the average bit-1 density in the code with a reduced code rate, thereby leading to an improved ISI performance. 
		
		The proposed family of binary ZP codes, such as the $(6, 15)$ ZP code, outperforms the conventional two error correcting $[8, 4]$ Reed Solomon code, single error correcting $[7,4]$ Hamming code and the existing $(7, 20)$ ISI-mtg code in terms of average BER in an MCvD channel with symbol duration $0.2$s. 
		This improvement is obtained due to the ZP code's smaller average ISI, consequently leading to a lower BER while maintaining a similar code rate compared to these codes.
		In data rate scenarios with a symbol duration of $0.3$s, where ISI effects are less severe, error correcting codes show better BER performance than the symbol duration of $0.2$s.
		However, as residual molecules in the channel increase, the $(5, 7)$ ZP code demonstrates superior BER performance compared to the error correcting codes, leveraging the ZP code's lower average bit-1 density to combat ISI.
		Due to the absence of consecutive bit-1s in the ZP code, we propose a simple location-based decoding technique for the decoding of the codes, which we call the MLR decoding algorithm. 
		
		In an MCvD channel characterized by refreshment after each successful transmission, the proposed $[8, 5]$ and $[12,  7]$ linear LOZP codes perform better over the considered linear and non-linear codes in this paper, with a symbol duration of $0.2$s. 
		Notably, in such channels, an improved BER performance with the linear LOZP codes motivates us to consider the location of bit-1s in the initial positions of the codewords as an important metric alongside the zero padding constraint and the average bit-1 density of the code. 
		Finally, we derive an upper bound on the code rate for the given code constraints and compare the code rate of the proposed codes with different parameters. 
		Combining Fig. \ref{fig_mol}, Fig. \ref{fig_sigma} and Fig. \ref{ber_with_without_refresh}, Table \ref{table_choose_best_codes} highlights the most suitable code (similar code rate) for different channel conditions. 
		For noisy channels with longer symbol duration, the $[8, 4]$ Reed Solomon code is preferred as noise becomes dominant over ISI in this scenario. 
		In other scenarios depending on channel refresh, the ZP and linear LOZP codes demonstrate superior performance due to their ISI-reducing properties.
		
		\begin{table}[t]
			\centering
			\caption{Code Selections with Different Channel Conditions.}
				\begin{tabular}{|p{1cm}|p{1.2cm}|p{1.75cm}|p{3cm}|}
					\hline
					Channel noise & Symbol duration & Channel \newline characteristics & Code \\ \hline
					&  $t_s = 0.2$s      & without refresh  & (5, 7) ZP code \\ \cline{3-4}
					&       &  with refresh  &[12, 7] LOZP code \\ \cline{2-4}
					Without noise &  $t_s = 0.3$s      &  without refresh  & (5, 7) ZP code, \newline [8, 4] Reed Solomon code  \\ \cline{3-4}
					&      &  with refresh  & [12, 7] LOZP code,\newline [8, 4] Reed Solomon code \\	\cline{1-4}
					& $t_s = 0.2$s       & without refresh  & (5, 7) ZP code  \\ \cline{3-4}
					With  &       &  with refresh  & [12, 7] LOZP code \\ 	\cline{2-4}
					noise & $t_s = 0.3$s       & without refresh  & [8, 4] Reed Solomon code \\ \cline{3-4}
					&        & with refresh  & [8, 4] Reed Solomon code \\
					\hline
				\end{tabular}
				\label{table_choose_best_codes}
			\end{table}
\bibliographystyle{IEEEtran}

\begin{thebibliography}{10}
	\providecommand{\url}[1]{#1}
	\csname url@samestyle\endcsname
	\providecommand{\newblock}{\relax}
	\providecommand{\bibinfo}[2]{#2}
	\providecommand{\BIBentrySTDinterwordspacing}{\spaceskip=0pt\relax}
	\providecommand{\BIBentryALTinterwordstretchfactor}{4}
	\providecommand{\BIBentryALTinterwordspacing}{\spaceskip=\fontdimen2\font plus
		\BIBentryALTinterwordstretchfactor\fontdimen3\font minus
		\fontdimen4\font\relax}
	\providecommand{\BIBforeignlanguage}[2]{{%
			\expandafter\ifx\csname l@#1\endcsname\relax
			\typeout{** WARNING: IEEEtran.bst: No hyphenation pattern has been}%
			\typeout{** loaded for the language `#1'. Using the pattern for}%
			\typeout{** the default language instead.}%
			\else
			\language=\csname l@#1\endcsname
			\fi
			#2}}
	\providecommand{\BIBdecl}{\relax}
	\BIBdecl
	
	\bibitem{6122529}
	B.~Atakan, O.~B. Akan, and S.~Balasubramaniam, ``Body area nanonetworks with
	molecular communications in nanomedicine,'' \emph{IEEE Communications
		Magazine}, vol.~50, no.~1, pp. 28--34, Jan. 2012.
	
	\bibitem{6881284}
	S.~Qiu, W.~Guo, S.~Wang, N.~Farsad, and A.~Eckford, ``A molecular communication
	link for monitoring in confined environments,'' in \emph{Proceedings IEEE
		International Conference on Communications Workshops (ICC)}, Sydney, NSW,
	Australia, Jun. 2014, pp. 718--723.
	
	\bibitem{suda2005exploratory}
	T.~Suda, M.~Moore, T.~Nakano, R.~Egashira, and A.~Enomoto, ``Exploratory
	research on molecular communication between nanomachines,'' in
	\emph{Proceedings Genetic and Evolutionary Computation Conference (GECCO),
		Late Breaking Papers}, vol.~25, Washington, DC, USA., Jun. 2005, p.~29.
	
	\bibitem{nakano_eckford_haraguchi_2013}
	T.~Nakano, A.~W. Eckford, and T.~Haraguchi, \emph{Molecular
		Communication}.\hskip 1em plus 0.5em minus 0.4em\relax Cambridge University
	Press, 2013.
	
	\bibitem{7841486}
	Y.~Deng, A.~Noel, W.~Guo, A.~Nallanathan, and M.~Elkashlan, ``{3D} stochastic
	geometry model for large-scale molecular communication systems,'' in
	\emph{Proceedings IEEE Global Communications Conference (GLOBECOM)},
	Washington, DC, USA, Dec. 2016, pp. 1--6.
	
	\bibitem{8742793}
	V.~Jamali, A.~Ahmadzadeh, W.~Wicke, A.~Noel, and R.~Schober, ``Channel modeling
	for diffusive molecular communication—a tutorial review,''
	\emph{Proceedings of the IEEE}, vol. 107, no.~7, pp. 1256--1301, Jul. 2019.
	
	\bibitem{9184816}
	X.~Chen, Y.~Huang, L.-L. Yang, and M.~Wen, ``Generalized molecular-shift keying
	({GMoSK}): Principles and performance analysis,'' \emph{IEEE Transactions on
		Molecular, Biological and Multi-Scale Communications}, vol.~6, no.~3, pp.
	168--183, Dec. 2020.
	
	\bibitem{10334472}
	G.~Yue, G.~Lin, Q.~Liu, and K.~Yang, ``Diffusion-based anti-interference joint
	modulation in {MIMO} molecular communication,'' \emph{IEEE Transactions on
		Molecular, Biological, and Multi-Scale Communications}, vol.~10, no.~1, pp.
	112--121, Mar. 2024.
	
	\bibitem{9839216}
	X.~Chen, M.~Wen, F.~Ji, Y.~Huang, Y.~Tang, and A.~W. Eckford, ``Detection
	interval optimization for diffusion-based molecular communication,'' in
	\emph{Proceedings IEEE International Conference on Communications (ICC)},
	Seoul, South Korea, May 2022, pp. 3691--3696.
	
	\bibitem{9768128}
	M.~C. Gursoy and U.~Mitra, ``Higher order derivative-based receiver
	preprocessing for molecular communications,'' \emph{IEEE Transactions on
		Molecular, Biological, and Multi-Scale Communications}, vol.~8, no.~3, pp.
	178--189, Sept. 2022.
	
	\bibitem{9840365}
	X.~Chen, F.~Ji, M.~Wen, Y.~Huang, Y.~Tang, and A.~W. Eckford, ``Low complexity
	first: Duration-centric {ISI} mitigation in molecular communication via
	diffusion,'' \emph{IEEE Communications Letters}, vol.~26, no.~11, pp.
	2665--2669, Nov. 2022.
	
	\bibitem{10088444}
	W.~Gao and L.-L. Yang, ``Interference mitigation-enabled signal detection in
	diffusive molecular communications systems with molecular-type spreading,''
	\emph{IEEE Internet of Things Journal}, vol.~10, no.~15, pp.
	13\,849--13\,864, Aug. 2023.
	
	\bibitem{7273857}
	Y.~Lu, M.~D. Higgins, and M.~S. Leeson, ``Comparison of channel coding schemes
	for molecular communications systems,'' \emph{IEEE Transactions on
		Communications}, vol.~63, no.~11, pp. 3991--4001, Nov. 2015.
	
	\bibitem{7248461}
	------, ``Self-orthogonal convolutional codes ({SOCC}s) for diffusion-based
	molecular communication systems,'' in \emph{Proceedings IEEE International
		Conference on Communications (ICC)}, London, UK, Jun. 2015, pp. 1049--1053.
	
	\bibitem{8633972}
	M.~B. Dissanayake, Y.~Deng, A.~Nallanathan, M.~Elkashlan, and U.~Mitra,
	``Interference mitigation in large-scale multiuser molecular communication,''
	\emph{IEEE Transactions on Communications}, vol.~67, no.~6, pp. 4088--4103,
	Jun. 2019.
	
	\bibitem{7859349}
	M.~B. Dissanayake, Y.~Deng, A.~Nallanathan, E.~M.~N. Ekanayake, and
	M.~Elkashlan, ``Reed {S}olomon codes for molecular communication with a full
	absorption receiver,'' \emph{IEEE Communications Letters}, vol.~21, no.~6,
	pp. 1245--1248, Jun. 2017.
	
	\bibitem{Minimum}
	C.~Bai, M.~S. Leeson, and M.~D. Higgins, ``Minimum energy channel codes for
	molecular communications,'' \emph{Electronics Letters}, vol.~50, no.~23, pp.
	1669--1671, Nov. 2014.
	
	\bibitem{6708566}
	P.~Shih, C.~Lee, P.~Yeh, and K.~Chen, ``Channel codes for reliability
	enhancement in molecular communication,'' \emph{IEEE Journal on Selected
		Areas in Communications}, vol.~31, no.~12, pp. 857--867, Dec. 2013.
	
	\bibitem{8648429}
	A.~O. Kislal, H.~B. Yilmaz, A.~E. Pusane, and T.~Tugcu, ``{ISI}-aware channel
	code design for molecular communication via diffusion,'' \emph{IEEE
		Transactions on NanoBioscience}, vol.~18, no.~2, pp. 205--213, Apr. 2019.
	
	\bibitem{8972472}
	A.~O. Kislal, B.~C. Akdeniz, C.~Lee, A.~E. Pusane, T.~Tugcu, and C.~Chae,
	``{ISI}-mitigating channel codes for molecular communication via diffusion,''
	\emph{IEEE Access}, vol.~8, pp. 24\,588--24\,599, Jan. 2020.
	
	\bibitem{10041114}
	P.~Hofmann, J.~A. Cabrera, R.~Bassoli, M.~Reisslein, and F.~H.~P. Fitzek,
	``Coding in diffusion-based molecular nanonetworks: A comprehensive survey,''
	\emph{IEEE Access}, vol.~11, pp. 16\,411--16\,465, Feb. 2023.
	
	\bibitem{10356127}
	Y.~Tang, F.~Ji, Q.~Wang, M.~Wen, C.-B. Chae, and L.-L. Yang, ``{Reed-Solomon}
	coded probabilistic constellation shaping for molecular communications,''
	\emph{IEEE Communications Letters}, vol.~28, no.~2, pp. 258--262, Feb. 2024.
	
	\bibitem{10250855}
	H.~Hyun, C.~Lee, M.~Wen, S.-H. Kim, and C.-B. Chae, ``{ISI}-mitigating
	character encoding for molecular communications via diffusion,'' \emph{IEEE
		Wireless Communications Letters}, vol.~13, no.~1, pp. 24--28, Jan. 2024.
	
	\bibitem{9336654}
	B.~Dhayabaran, G.~T. Raja, and M.~Magarini, ``Low complex receiver design for
	modified inverse source coded diffusion-based molecular communication
	systems,'' \emph{IEEE Transactions on Molecular, Biological and Multi-Scale
		Communications}, vol.~7, no.~4, pp. 239--252, Dec. 2021.
	
	\bibitem{6868273}
	A.~Noel, K.~C. Cheung, and R.~Schober, ``Optimal receiver design for diffusive
	molecular communication with flow and additive noise,'' \emph{IEEE
		Transactions on NanoBioscience}, vol.~13, no.~3, pp. 350--362, Sept. 2014.
	
	\bibitem{noel2014improving}
	------, ``Improving receiver performance of diffusive molecular communication
	with enzymes,'' \emph{IEEE Transactions on NanoBioscience}, vol.~13, no.~1,
	pp. 31--43, Dec. 2014.
	
	\bibitem{keshavarz2019inter}
	A.~Keshavarz-Haddad, A.~Jamshidi, and P.~Akhkandi, ``Inter-symbol interference
	reduction channel codes based on time gap in diffusion-based molecular
	communications,'' \emph{Nano Communication Networks}, vol.~19, pp. 148--156,
	Mar. 2019.
	
	\bibitem{nath2023novel}
	T.~Nath and A.~Banerjee, ``On novel {ISI}-reducing channel codes for molecular
	communication via diffusion,'' in \emph{IEEE International Symposium on
		Information Theory (ISIT)}, Taipei, Taiwan, Jun. 2023, pp. 642--647.
	
	\bibitem{bhattacharjee2022channel}
	S.~Bhattacharjee, M.~Damrath, and P.~A. Hoeher, ``Channel coding techniques in
	macroscopic air-based molecular communication,'' in \emph{Proceedings of the
		9th ACM International Conference on Nanoscale Computing and Communication},
	Barcelona, Spain, Oct., 2022, pp. 1--2.
	
	\bibitem{10134568}
	D.~Jing and A.~W. Eckford, ``Lightweight channel codes for {ISI} mitigation in
	molecular communication between bionanosensors,'' \emph{IEEE Sensors
		Journal}, vol.~23, no.~13, pp. 13\,859--13\,867, Jul. 2023.
	
	\bibitem{9840783}
	T.~Nath, K.~G. Benerjee, and A.~Banerjee, ``On effect of different sequence
	distributions on {ISI} in an {MCvD} system,'' in \emph{Proceedings IEEE
		International Conference on Signal Processing and Communications (SPCOM)},
	Bangalore, India, Jul. 2022, pp. 1--5.
	
	\bibitem{6807659}
	H.~B. Yilmaz, A.~C. Heren, T.~Tugcu, and C.-B. Chae, ``Three-dimensional
	channel characteristics for molecular communications with an absorbing
	receiver,'' \emph{IEEE Communications Letters}, vol.~18, no.~6, pp. 929--932,
	Jun. 2014.
	
	\bibitem{8412141}
	L.~Shi and L.-L. Yang, ``Error performance analysis of diffusive molecular
	communication systems with {ON-OFF} keying modulation,'' \emph{IEEE
		Transactions on Molecular, Biological and Multi-Scale Communications},
	vol.~3, no.~4, pp. 224--238, Dec. 2017.
	
	\bibitem{farsad2014channel}
	N.~Farsad, N.-R. Kim, A.~W. Eckford, and C.-B. Chae, ``Channel and noise models
	for nonlinear molecular communication systems,'' \emph{IEEE Journal on
		Selected Areas in Communications}, vol.~32, no.~12, pp. 2392--2401, Dec.
	2014.
	
	\bibitem{zhai2018anti}
	H.~Zhai, Q.~Liu, A.~V. Vasilakos, and K.~Yang, ``Anti-{ISI} demodulation scheme
	and its experiment-based evaluation for diffusion-based molecular
	communication,'' \emph{IEEE Transactions on NanoBioscience}, vol.~17, no.~2,
	pp. 126--133, Apr. 2018.
	
	\bibitem{10361887}
	F.~Vakilipoor, L.~Barletta, S.~Bregni, and M.~Magarini, ``Achievable rate
	analysis in diffusive molecular communication channels with memory,'' in
	\emph{Proceedings IEEE Latin-American Conference on Communications
		(LATINCOM)}, Nov. 2023, pp. 1--6.
	
	\bibitem{7331300}
	B.~Tepekule, A.~E. Pusane, H.~B. Yilmaz, C.-B. Chae, and T.~Tugcu, ``{ISI}
	mitigation techniques in molecular communication,'' \emph{IEEE Transactions
		on Molecular, Biological and Multi-Scale Communications}, vol.~1, no.~2, pp.
	202--216, Jun. 2015.
	
	\bibitem{983680}
	S.~Lin and D.~J. Costello, \emph{Error Control Coding, 2nd Edition}.\hskip 1em
	plus 0.5em minus 0.4em\relax USA: Prentice-Hall, Inc., 2004.
	
	\bibitem{10620232}
	S.~Angerbauer, N.~Tuccitto, G.~T. Sfrazzetto, R.~Santonocito, and W.~Haselmayr,
	``Investigation of different chemical realizations for molecular matrix
	multiplications,'' \emph{IEEE Transactions on Molecular, Biological, and
		Multi-Scale Communications}, vol.~10, no.~3, pp. 464--469, Sept. 2024.
	
\end{thebibliography}


		\end{document}